\documentclass{article}%
\usepackage{amsmath}
\usepackage{amsfonts}
\usepackage{esint}
\usepackage{tikz}
\usepackage{amssymb}
\usepackage{graphicx}%
\providecommand{\U}[1]{\protect\rule{.1in}{.1in}}
\numberwithin{equation}{section}
\newtheorem{theorem}{Theorem}

\newtheorem{definition}[theorem]{Definition}

\newtheorem{lemma}[theorem]{Lemma}

\newtheorem{proposition}[theorem]{Proposition}
\newtheorem{remark}[theorem]{Remark}

\newenvironment{proof}[1][Proof]{\noindent\textbf{#1.} }{\ \rule{0.5em}{0.5em}}
\begin{document}
\bigskip

\begin{center}
{\Large Classical non mass preserving solutions of coagulation equations.}

\bigskip

M. Escobedo\footnote{Departamento de Matem\'{a}ticas. Universidad del
Pa\'{\i}s Vasco. Apartado 644. E-48080 Bilbao, Spain. E-mail:
miguel.escobedo@ehu.es}, J. J. L. Vel\'{a}zquez\footnote{ICMAT
(CSIC-UAM-UC3M-UCM). Facultad de Matem\'{a}ticas. Universidad Complutense.
E-28040, Madrid, Spain. E-mail: JJ\_Velazquez@mat.ucm.es}

\bigskip
\end{center}

{\large Abstract:} In this paper we construct classical solutions of a family
of coagulation equations with homogeneous kernels that exhibit the behaviour
known as gelation. This behaviour consists in the loss of mass due to the fact
that some of the particles can become infinitely large in finite time.

\section{INTRODUCTION}

In this paper we prove existence of solutions of the classical coagulation
equation for which the mass is not conserved in time. The coagulation equation
reads as:%
\begin{align}
\frac{\partial f}{\partial t}\left(  t,x\right)   &  =Q\left[  f\right]
\left(  t,x\right)  \ \ \ \ ,\ \ \ \ x\geq0\ \ ,\ \ \ t>0\label{S1E1}\\
Q\left[  f\right]   &  =\frac{1}{2}\int_{0}^{x}K\left(  x-y,y\right)  f\left(
t,x-y\right)  f\left(  t,y\right)  -\int_{0}^{\infty}K\left(  x,y\right)
f\left(  t,x\right)  f\left(  t,y\right)  dy\label{S1E2}\\
f\left(  x,0\right)   &  =f_{0}\left(  x\right)  \ \ \ ,\ \ \ \ x>0
\label{S1E3}%
\end{align}
where the kernel $K$ whose specific form will be precised later, satisfies
$K\left(  x,y\right)  =K\left(  y,x\right)  \geq0.$

\bigskip

The solutions of (\ref{S1E1})-(\ref{S1E3}) satisfy formally, assuming that
Fubini's Theorem can be applied, the mass conservation property:%
\begin{equation}
\frac{d}{dt}\left(  \int_{0}^{\infty}xf\left(  t,x\right)  dx\right)  =0
\label{S1E4}%
\end{equation}

However, it is well known that for a large class of homogeneous kernels
$K\left(  x,y\right)  $ solutions of (\ref{S1E1})-(\ref{S1E3}) satisfying
(\ref{S1E4}) cannot exist globally in time (cf. \cite{EZH}, \cite{EMP},
\cite{Jeon}, \cite{McLeod}). More precisely, there exists solutions of
(\ref{S1E1})-(\ref{S1E3}) that preserve the total mass of the particles
$\int_{0}^{\infty}xf\left(  t,x\right)  dx$ during a finite time interval
$0\leq t\leq T<\infty,$ but the mass is not preserved for arbitrarily long
times. This phenomenon is usually termed as gelation.

In this paper we will restrict our attention to the study of kernels with the
form:%
\begin{equation}
K\left(  x,y\right)  =\left(  x\ y\right)  ^{\frac{\lambda}{2}}%
\ \ ,\ \ 1<\lambda<2 \label{S0E1}%
\end{equation}

The range of exponents in (\ref{S0E1}) is the one in which changes of mass of
order one can be expected in times of order one. Global weak solutions of
(\ref{S1E1}) have been obtained in \cite{L}.

\bigskip

The main goal of this paper is to construct classical solutions of
(\ref{S1E1})-(\ref{S1E3}) exhibiting gelation. We will assume that the initial
data behaves as a suitable power law for large values of $x,$ and therefore
that the loss of mass takes place since $t=0.$ In particular, in the classical
solutions obtained in this paper, it will be possible to compute a detailed
asymptotic behaviour of the solution $f\left(  t,x\right)  $ as $x\rightarrow
\infty,$ as well as the flux of mass escaping to infinity. The solutions
obtained will be local in time, since we cannot avoid the possibility of
discontinuities in the fluxes at infinity for \ positive times.

\bigskip

The results obtained in this paper rely heavily in the estimates obtained in
the papers \cite{EV1}, \cite{EV2}, where some related linear coagulation
models were studied. In particular we have obtained very detailed estimates
for the fundamental solution of the linear coagulation equation that results
linearizing (\ref{S1E1})-(\ref{S1E3}) around the power law $\bar{f}\left(
x\right)  =\frac{1}{x^{\frac{3+\lambda}{2}}}$ in \cite{EV1}. On the other
hand, we have introduced in \cite{EV2} some natural functional spaces to study
the linearized version of (\ref{S1E1})-(\ref{S1E3}) that results considering
small deviations of a bounded initial data $f_{0}\left(  x\right)  $ behaving
asymptotically as $\frac{1}{x^{\frac{3+\lambda}{2}}}$ as $x\rightarrow\infty.$
Both the fundamental solution in \cite{EV1} and the functional framework
introduced in \cite{EV2}~will be used extensively in this paper.

\bigskip

The coagulation equation is one among a large family of kinetic equations
exhibiting particle fluxes for homogeneous solutions. Several examples can be
found in \cite{BZ}. A rigorous construction of solutions exhibiting loss of
mass for small values of the energy for the so-called Uehling-Uhlenbeck
equation (or quantum Boltzmann equation) has been obtained in \cite{EMV1},
\cite{EMV2}. The type of methods used in those papers is closely related to
the ones used in this paper, although there are some technical differences.

In both cases (coagulation and Uehling-Uhlenbeck) we can think that the
obtained solutions are mass preserving measure valued solutions having a
singular part at some distinghished point and a regular part that is described
by the integro-differential equations. In the case of coagulation the singular
part of the measure (or gel) would be supported at $x=\infty,$ and in the case
of Uehling-Uhlenbeck such atomic measure (or Bose-Einstein condensate) would
correspond to a macroscopic fraction of particles with zero energy. A natural
question that arises in both cases, and in general in the study of equations
with particle fluxes is to understand the interaction\ between the singular
measure and the regular part of the measure. For the solutions obtained in
\cite{EMV1}, \cite{EMV2} and in this paper we assume that the regular part of
the measure is not affected by the singular part. However, it is well known
that such interaction could be nontrivial. For instance, in the case of
coagulation models, explicit examples for the kernel $K\left(  x,y\right)
=x\cdot y$ show that different solutions can be expected if there is
interaction between the singular part and the regular part (cf. \cite{Flory41}%
, \cite{St}) or if such interaction does not exist. For more general kernels
it is known that different dynamics can arise for different mass preserving
regularizations of the kernel $K\left(  x,y\right)  $ after passing to the
limit where gelation can occur (cf. \cite{FL}). In the case of
Uehling-Uhlenbeck the computations and physical arguments in \cite{LLPR},
\cite{ST1}, \cite{ST2} suggest the existence of solutions of this equation
exhibiting nontrivial interactions between the regular part of the particle
distribution and the Bose-Einstein condensate. We also remark that in
\cite{Lu1}, \cite{Lu2} a construction of global mass preserving weak solutions
for the Uehling-Uhlenbeck system has been given. Such a construction begins
regularizing the collision kernel for small energies and pass to the limit in
the cutoff parameter. It is not known if the solutions constructed in
\cite{Lu1}, \cite{Lu2} are the same as the ones in \cite{EMV1}, \cite{EMV2}.
In all these problems a detailed understanding of the physical regularizations
yielding cutoff mechanisms plays a crucial role (cf. also \cite{Spohn} for a
discussion about these problems).

\bigskip

The plan of this paper is the following. In Section 2 we describe the
functional framework used to prove the main Theorem of this paper and state
the main result. Section 3 gives a general sketch of the strategy of the
proof. Section 4 summarizes some results that have been proved in \cite{EV1},
\cite{EV2} that will be used in this paper. Section 5 contains some auxiliary
technical results concerning the functional spaces as well as the fundamental
solution $g\left(  \tau,x;1\right)  $ studied in \cite{EV1}. Section 6
provides some estimates for the nonlinear term. Section 7 describes the
asymptotics of the solutions of some linear equations as $x\rightarrow\infty$
in a detailed manner. Finally Section 8 explains the fixed point argument that
concludes the proof of the Theorem.

\bigskip

\section{FUNCTIONAL FRAMEWORK AND MAIN RESULT.\label{initialdata}}

In this paper we will choose the initial data in (\ref{S1E3}) satisfying
$f_{0}\in C^{3}(\mathbb{R}^{+})$. We will assume also, as in \cite{EV2}, that
the function $f_{0}$ is close to a power law for large $x$. To this end we
define:%
\begin{equation}
r=\frac{\lambda-1}{2} \label{U3E4a}%
\end{equation}
We fix also $\delta>0$ satisfying $\delta<\min\left\{  r,\frac{2-\lambda}%
{2}\right\}  .$ We will then assume that $f_{0}$ has the form:%
\begin{align}
f_{0}\left(  x\right)   &  =f_{1}\left(  x\right)  +f_{2}\left(  x\right)
+f_{3}\left(  x\right)  \ ,\ \ f_{1}\left(  x\right)  =\frac{D_{1}\xi\left(
x\right)  }{x^{\frac{3+\lambda}{2}}}\ ,\ \ f_{2}\left(  x\right)  =\frac
{D_{2}\xi\left(  x\right)  }{x^{\frac{3+\lambda}{2}+r}}\label{Z1E1a}\\
f_{1;2}\left(  x\right)   &  =f_{1}\left(  x\right)  +f_{2}\left(  x\right)
\label{Z1E2}%
\end{align}
where $D_{1}>0,\ D_{2}\in\mathbb{R}$ and:
\begin{equation}
\xi\in C^{\infty}\left[  0,\infty\right)  ,\ \xi(x)=1\text{ for }x\geq1\text{
and }\xi(x)=0\text{ if \ }0\leq x\leq1/2\ ,\ \ \xi^{\prime}(x)\geq0
\label{Z1E2a}%
\end{equation}%
\begin{equation}
\left\vert f_{3}^{k}\left(  x\right)  \right\vert \leq\frac{B}{\left(
x+1\right)  ^{\frac{3+\lambda}{2}+r+k+\delta}}\ \ ,\ \ \ k=0,1,2,3,4
\label{Z1E2b}%
\end{equation}
for some $B>0.$ The following auxiliary function will be used repeatedly:
\begin{equation}
h_{0}(x)=f_{0}(x)-f_{1}\left(  x\right)  =f_{2}\left(  x\right)  +f_{3}\left(
x\right)  \label{S1E5}%
\end{equation}

Notice that (\ref{Z1E1a})-(\ref{Z1E2b}) imply:
\begin{align}
&  \left(  1+y^{\frac{3+\lambda}{2}+r}\right)  \left\vert h_{0}(y)\right\vert
+\left(  1+y^{\frac{3+\lambda}{2}+r+1}\right)  \left\vert h_{0}^{\prime
}(y)\right\vert +\label{S1E6}\\
&  +\left(  1+y^{\frac{3+\lambda}{2}+r+2}\right)  \left\vert h_{0}%
^{\prime\prime}(y)\right\vert +\left(  1+y^{\frac{3+\lambda}{2}+r+3}\right)
\left\vert h_{0}^{\prime\prime\prime}(y)\right\vert \leq CB.\nonumber
\end{align}
for some $C>0.$ We will assume in the rest of the paper that $C$ is a generic
constant that can change from line to line and that might depend only on
$D_{1},D_{2},B,\lambda$ and $\delta$ unless some additional dependence is
written explicitly. Moreover, we will assume without loss of generality that
$D_{1}=1,$ since this parameter can be absorbed in a rescaling of $t.$

For any interval $I\subset(0,+\infty)$ we will denote as $L^{2}\left(
I\right)  $ the usual Lebesgue space of square integrable functions. For any
$\sigma>0$ we denote as $H^{\sigma}(I)$ the usual Sobolev space $W^{\sigma
,2}(I)$. The corresponding norms will be denoted $||\cdot||_{L^{2}}$ and
$||\cdot||_{H^{\sigma}}$. Dealing with functions depending on variables $x$
and $t$ we will write $H_{x}^{\sigma}$ or $L_{t}^{2}$ in order to indicate the
argument with respect to which the norm is taken.

In order to define suitable functional spaces we define, for any $T>0,\ R>0:$

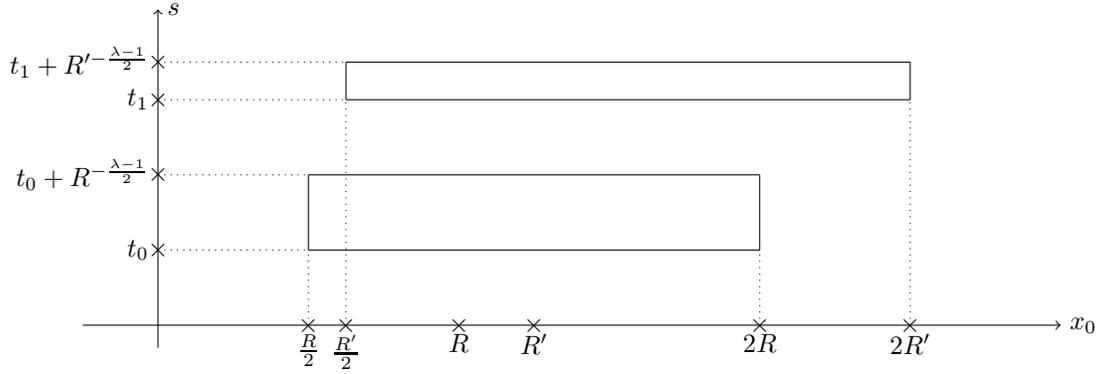
\begin{figure}
\begin{tikzpicture}
\draw[->] (-1, 0) -- (12, 0) node[right]{$x_0$};
\draw[->] (0, -0.3) -- (0, 4.2) node[right]{$s$};
\draw(2, 0) node {$\times$};
\draw(2, 0) node[below]{$\frac{R}{2}$};
\draw(4, 0) node {$\times$};
\draw(4, 0) node[below]{$R$};
\draw(8, 0) node {$\times$};
\draw(8, 0) node[below]{$2R$};
\draw(0, 1) node {$\times$};
\draw(0, 1) node[left]{$t_0$};
\draw(0, 2) node {$\times$};
\draw(0, 2) node[left]{$t_0+R^{-\frac{\lambda-1}{2}}$};
\draw(2, 1)--(8,1);
\draw(2, 1)--(2, 2);
\draw(2, 2)--(8, 2);
\draw(8, 1)--(8, 2);
\draw[dotted](2, 0)--(2, 1);
\draw[dotted](8, 0)--(8, 1);
\draw[dotted](0, 1)--(2,1);
\draw[dotted](0, 2)--(2,2);

\draw(2.5, 0) node {$\times$};
\draw(2.5, 0) node[below]{$\frac{R'}{2}$};
\draw(5, 0) node {$\times$};
\draw(5, 0) node[below]{$R'$};
\draw(10, 0) node {$\times$};
\draw(10, 0) node[below]{$2R'$};
\draw(0, 3) node {$\times$};
\draw(0, 3) node[left]{$t_1$};
\draw(0, 3.5) node {$\times$};
\draw(0, 3.5) node[left]{$t_1+R'^{-\frac{\lambda-1}{2}}$};
\draw(2.5, 3)--(10,3);
\draw(2.5, 3)--(2.5, 3.5);
\draw(2.5, 3.5)--(10, 3.5);
\draw(10, 3)--(10, 3.5);
\draw[dotted](2.5, 0)--(2.5, 3);
\draw[dotted](10, 0)--(10, 3);
\draw[dotted](0, 3)--(2.5,3);
\draw[dotted](0, 3.5)--(2.5, 3.5);

\end{tikzpicture}
\caption{Two cubes of the kind appearing in the norms $N_{2, \sigma}$ and $N_\infty$ defined below.}
\end{figure}

\begin{equation}
N_{2;\,\sigma}(f;\,t_{0},R)=\left(  R^{\frac{\lambda-1}{2}+2\sigma-1}%
\int_{t_{0}}^{\min(t_{0}+R^{-(\lambda-1)/2},T)}||D_{x}^{\sigma}f(t)||_{L^{2}%
(R/2,2R)}^{2}dt\right)  ^{1/2}\ \ ,\ \ \sigma\geq0 \label{S1E8}%
\end{equation}

\begin{equation}
M_{2;\sigma}(f;\,R)=\left(  R^{2\sigma-1}\int_{0}^{T}||D_{x}^{\sigma
}f(t)||_{L^{2}(R/2,2R)}^{2}dt\right)  ^{1/2}\ \ ,\ \ \sigma\geq0 \label{S1E10}%
\end{equation}%
\begin{align*}
N_{\infty}(f;\,t_{0},R)  &  =\left(  R^{\frac{\lambda-1}{2}}\int_{t_{0}}%
^{\min(t_{0}+R^{-(\lambda-1)/2},T)}||f(t)||_{L^{\infty}(R/2,2R)}^{2}dt\right)
^{1/2}\ \ \\
M_{\infty}(f;\,R)  &  =\left(  \int_{0}^{T}||f(t)||_{L^{\infty}(R/2,2R)}%
^{2}dt\right)  ^{1/2}%
\end{align*}

Then, for any $\sigma>0$ we define the following norms:%
\begin{align*}
\left\Vert f\right\Vert _{Y_{q,p}^{\sigma}\left(  T\right)  }  &
=\sup_{0<R\leq1}R^{q}M_{2;\,0}(f;R)+\sup_{0<R\leq1}R^{q}M_{2;\,\sigma}(f;R)+\\
&  +\sup_{0\leq t_{0}\leq T}\sup_{R\geq1}R^{p}N_{2;0}(f;t_{0},R)+\sup_{0\leq
t_{0}\leq T}\sup_{R\geq1}R^{p}N_{2;\,\sigma}(f;t_{0},R)
\end{align*}%
\[
||f||_{X_{q,p}\left(  T\right)  }=\sup_{0<R\leq1}R^{q}M_{\infty}%
(f;R))+\sup_{0\leq t_{0}\leq T}\sup_{R\geq1}R^{p}N_{\infty}(f;t_{0},R)
\]%
\begin{align}
\left\vert \left\vert \left\vert f\right\vert \right\vert \right\vert _{q,p}
&  =\sup_{0\leq x\leq1}\left\{  x^{q}\left\vert f\left(  x\right)  \right\vert
\right\}  +\sup_{x>1}\left\{  x^{p}\left\vert f\left(  x\right)  \right\vert
\right\} \label{M2E8}\\
\left\vert \left\vert \left\vert f\right\vert \right\vert \right\vert
_{\sigma}  &  =\sup_{0\leq t\leq T}\left\vert \left\vert \left\vert
f\right\vert \right\vert \right\vert _{\frac{3}{2},\frac{3+\lambda}{2}%
}+\left\Vert f\right\Vert _{Y_{\frac{3}{2},\frac{3+\lambda}{2}}^{\sigma
}\left(  T\right)  } \label{M2E9}%
\end{align}
and the following spaces:
\[
Y_{q,p}^{\sigma}\left(  T\right)  =\left\{  f:||f||_{Y_{q,p}^{\sigma}%
(T)}<\infty\right\}  \ \ \ \ ,\ \ X_{q,p}\left(  T\right)  =\left\{
f:||f||_{X_{q,p}(T)}<\infty\right\}
\]%
\[
\mathcal{E}_{T;\sigma}=\left\{  f:\left\vert \left\vert \left\vert
f\right\vert \right\vert \right\vert _{\sigma}<\infty\right\}
\]

Throughout this paper we will assume that%
\begin{equation}
\sigma\in\left(  1,2\right)  \ \label{Z1E4}%
\end{equation}

Therefore, Sobolev embeddings imply $Y_{q,p}^{\sigma}\left(  T\right)  \subset
X_{q,p}\left(  T\right)  .$ Actually such embeddings would take place assuming
the weaker condition $\sigma>\frac{1}{2}.$ The main reason for the choice of
$\sigma$ as in (\ref{Z1E4}) is purely technical, and it is due to the fact
that the Theorem proved in \cite{EV2} to solve a suitable linearized problem
(cf. for instance (\ref{S2E4})) requires such a regularity. It is likely that
using the\ "almost half derivatives" that we introduce now would be possible
to weaken the condition on $\sigma$ to $\frac{1}{2}<\sigma<1$ both for the
results of \cite{EV2} and this paper (cf. Remark 6.4 in \cite{EV2}).

We will solve (\ref{S1E1})-(\ref{S1E3}) using a functional space that measures
in a natural way the regularizing effects of the coagulation equation as
$x\rightarrow\infty$ that have been studied in \cite{EV2}. Let $\eta\in
C^{\infty}\left(  \mathbb{R}^{+}\right)  $ a cutoff function satisfying
$\eta(x)=1$ for $x\in\left(  \frac{1}{4},3\right)  ,$ $\eta(x)=0$ for
$x\not \in \left(  \frac{1}{8},4\right)  .$ Given $f\in C\left(
\mathbb{R}^{+}\right)  ,\ t_{0}\in\lbrack0,T],\ R\geq1$ we define:
\begin{equation}
F_{R,t_{0}}(\theta,X)=\eta(RX)f\left(  t_{0}+\theta R^{-(\lambda
-1)/2},RX\right)  \label{S1E10a}%
\end{equation}
and:%
\begin{equation}
\left[  f\right]  _{p}^{\sigma;\frac{1}{2}}=\sup_{R\geq1}\sup_{0\leq t_{0}\leq
T}R^{p}\left(  \int_{t_{0}}^{\min(t_{0}+R^{-(\lambda-1)/2},T)}\int
_{\mathbb{R}}|\widehat{F}_{R,t_{0}}(\theta,k)|^{2}Q_{R,\sigma}\left(
k\right)  dk\,d\theta\right)  ^{1/2} \label{S1E11}%
\end{equation}
where $Q_{R,\sigma}\left(  k\right)  =\left(  1+|k|^{2\,\sigma}\right)
\left(  1+\min\{|k|,\,R\}\right)  .$%

\begin{align}
\left\Vert f\right\Vert _{\mathcal{Z}_{p}^{\sigma;\frac{1}{2}}\left(
T\right)  }  &  =\left\Vert f\right\Vert _{L^{2}\left(  \left(  0,T\right)
;H_{x}^{\sigma}\left(  0,2\right)  \right)  }+\left[  f\right]  _{p}%
^{\sigma;\frac{1}{2}}+\sup_{0\leq t\leq T}\left\vert \left\vert \left\vert
f\right\vert \right\vert \right\vert _{\frac{3}{2},p}+||f||_{Y_{\frac{3}{2}%
,p}^{\sigma}\left(  T\right)  }\ \ \label{S1E12a}\\
\mathcal{Z}_{p}^{\sigma;\frac{1}{2}}\left(  T\right)   &  =\left\{
f:|\left\Vert f\right\Vert _{\mathcal{Z}_{p}^{\sigma;\frac{1}{2}}\left(
T\right)  }<\infty\right\} \nonumber
\end{align}

The intuition behind these spaces is the following. As it has been seen in
\cite{EV2} the main terms in the coagulation equation for solutions that are
close to the power law $x^{-\frac{3+\lambda}{2}}$ as $x\rightarrow\infty$ can
be thought as a perturbation of the half-derivative operator. However, since
the integral operator $Q\left[  f\right]  $ in (\ref{S1E2}) is an integral
operator the equation (\ref{S1E1}), cannot be expected to have smoothing
effects. Nevertheless, it has been seen in \cite{EV2} that the equation
(\ref{S1E1}) has some kind of regularizing effect\ due to the fact that the
right hand side of (\ref{S1E1}) can be thought, for solutions close to
$x^{-\frac{3+\lambda}{2}}$ as $x\rightarrow\infty$ as the half-derivative
operator, if we restrict ourselves to incremental quotients with length $x$
larger than one. This is the source of the regularizing effects that will be
studied using the functionals (\ref{S1E11}), (\ref{S1E12a}).

In order to gain some intuition about the spaces $X_{q,p}\left(  T\right)
,\ Y_{q,p}^{\sigma}\left(  T\right)  ,$ $\mathcal{Z}_{p}^{\sigma;\frac{1}{2}%
}\left(  T\right)  $ it is useful to think about them as functions that can be
estimated like $x^{-p}$ as $x\rightarrow\infty$ and $x^{-q}$ as $x\rightarrow
0$ in the case of the spaces $X_{q,p}\left(  T\right)  ,\ Y_{q,p}^{\sigma
}\left(  T\right)  $ and $x^{-\frac{3}{2}}$ in the case of $\mathcal{Z}%
_{p}^{\sigma;\frac{1}{2}}\left(  T\right)  .$ Concerning regularity, the
functions in $X_{q,p}\left(  T\right)  $ are estimated pointwise, the
functions in $Y_{q,p}^{\sigma}\left(  T\right)  $ have $\sigma$ derivatives in
space and the functions in $\mathcal{Z}_{p}^{\sigma;\frac{1}{2}}\left(
T\right)  $ have almost $\left(  \sigma+\frac{1}{2}\right)  $ derivatives in
the sense of the definition (\ref{S1E11}).

The main result of this paper is the following:

\begin{theorem}
\label{Th1}Suppose that $f_{0}$ satisfies (\ref{Z1E1a})-(\ref{Z1E2b}),
$\sigma$ is as in (\ref{Z1E4}) and $K$ is as in (\ref{S0E1}). Then, there
exists a classical solution $f\in\mathcal{Z}_{\frac{3+\lambda}{2}}%
^{\sigma;\frac{1}{2}}$ of (\ref{S1E1})-(\ref{S1E3}) with $f_{t}\in L^{\infty
}\left(  \left(  0,T\right)  \times\mathbb{R}^{+}\right)  .$ Moreover, this
solution is unique in the class of functions satisfying:%
\[
f\left(  t,x\right)  =\lambda\left(  t\right)  \xi(x)x^{-\frac{3+\lambda}{2}%
}+h\left(  t,x\right)
\]
with $\lambda\in C\left[  0,T\right]  ,\ h\in\mathcal{Z}_{\bar{p}}%
^{\sigma;\frac{1}{2}}\left(  T\right)  ,$ $\lim_{\bar{t}\rightarrow
0}\left\Vert h\right\Vert _{\mathcal{Z}_{\bar{p}}^{\sigma;\frac{1}{2}}\left(
\bar{t}\right)  }=0,\ $where $\bar{p}=\frac{3+\lambda}{2}+\bar{\delta}$ with
$0<\bar{\delta}<r,$ and $T$ small enough.
\end{theorem}

\begin{remark}
Assumptions (\ref{Z1E1a})-(\ref{Z1E2b}) seem a very strong condition. However,
this condition is analogous to the type of compatibility conditions that must
be assumed solving boundary value problems in order to obtain smooth
solutions, or also to assume that the initial data has as many derivatives
appear in the equation solving a parabolic problem. It is likely that
(\ref{Z1E1a})-(\ref{Z1E2b}) could be weakened to the form $f_{0}\left(
x\right)  =D_{1}x^{-\frac{3+\lambda}{2}}+O\left(  x^{-\frac{3+\lambda}%
{2}-\delta}\right)  $ as $x\rightarrow\infty$ for some $\delta>0.$ However, to
prove this would require to obtain some delicate regularizing effects that we
have preferred to avoid in this paper that is already rather technical. The
specific value of $r$ will play a role in the proof of Proposition
\ref{PropAsympt} (cf. Remark \ref{rationalef_0}) as well as in the Proof of
Proposition \ref{Proposition_psi}.
\end{remark}

\section{GENERAL STRATEGY OF THE PROOF.\label{strategy}}

\bigskip

The general plan that we will use to prove Theorem \ref{Th1} is the following.
We look for a solution of (\ref{S1E1})-(\ref{S1E3}) in the form:%
\begin{equation}
f\left(  t,x\right)  =\lambda\left(  t\right)  f_{0}\left(  x\right)
+h\left(  t,x\right)  \label{S2E1}%
\end{equation}
where $f_{0}$ is the initial data (cf. (\ref{S1E3})) and $h$ will be a small
perturbation for short times. The function $\lambda$ is a differentiable
function to be prescribed satisfying $\lambda\left(  0\right)  =1.$ Then
$h,\ \lambda$ solve:%
\begin{equation}
h_{t}=\lambda\left(  t\right)  {\mathcal{L}}_{f_{0}}\left[  h\right]
+Q\left[  h\right]  +\left(  \lambda\left(  t\right)  \right)  ^{2}Q\left[
f_{0}\right]  -\lambda_{t}f_{0}\left(  x\right)  \label{S2E2}%
\end{equation}
where the linear operator ${\mathcal{L}}_{f_{0}}$ is as in \cite{EV2}:
\begin{align}
{\mathcal{L}}_{f_{0}}\left[  h\right]   &  =\int_{0}^{x}(x-y)^{\lambda/2}%
f_{0}(x-y)y^{\lambda/2}h(y)\,dy-\label{S2E3}\\
&  -x^{\lambda/2}\,f_{0}(x)\int_{0}^{\infty}y^{\lambda/2}h(y)dy-x^{\lambda
/2}h(x)\int_{0}^{\infty}y^{\lambda/2}f_{0}(y)dy.\nonumber
\end{align}

Our strategy is to solve (\ref{S2E2}) by means of a fixed point argument for a
suitable operator $\mathcal{T}$ defined in $\mathcal{Z}_{\bar{p}}%
^{\sigma;\frac{1}{2}}\left(  T\right)  $ with $r$ as in (\ref{U3E4a}%
),$\ \sigma$ as in (\ref{Z1E4}) and $T$ sufficiently small (cf. (\ref{S1E12a}%
)). It is convenient first, in order to apply the well-posedness results in
\cite{EV2} to introduce a new time scale. We will assume in all the paper that
$\left\vert \lambda\left(  t\right)  -1\right\vert \leq\frac{1}{2}.$ We can
then define a new time scale $\tau$ and a new function $\Lambda$ by means of:%
\begin{equation}
d\tau=\lambda\left(  t\right)  dt\ \ ,\ \ \ \tau=0\ \ \text{at\ \ }%
t=0\ \ \ ,\ \ \ \ \Lambda\left(  \tau\right)  =\lambda\left(  t\right)
\label{time}%
\end{equation}

Then (\ref{S2E2}) becomes:%
\[
h_{\tau}={\mathcal{L}}_{f_{0}}\left[  h\right]  +\frac{Q\left[  h\right]
}{\Lambda\left(  \tau\right)  }+\Lambda\left(  \tau\right)  Q\left[
f_{0}\right]  -\Lambda_{\tau}f_{0}\left(  x\right)
\]
where we will write $h\left(  t,x\right)  =h\left(  \tau,x\right)  $ by convenience.

Given $h\in\mathcal{Z}_{\bar{p}}^{\sigma;\frac{1}{2}}\left(  T\right)  $ and
$\Lambda\in C^{1}\left(  \left[  0,T\right]  \right)  $ we will define
$\tilde{h}=\tilde{h}\left[  \Lambda\right]  $ as the unique solution of:%
\begin{equation}
\tilde{h}_{\tau}={\mathcal{L}}_{f_{0}}\left[  \tilde{h}\right]  +\frac
{Q\left[  h\right]  }{\Lambda\left(  \tau\right)  }+\Lambda\left(
\tau\right)  Q\left[  f_{0}\right]  -\Lambda_{\tau}f_{0}\left(  x\right)
\label{S2E4}%
\end{equation}
in $\mathcal{E}_{T;\sigma}.$ The existence of such a solution will be a
consequence of the results in \cite{EV2}. In order to apply such a results we
will need to show that $Q\left[  f_{0}\right]  ,Q\left[  h\right]  \in
Y_{\frac{3}{2},\left(  2+\bar{\delta}\right)  }^{\sigma}\left(  T\right)  .$
In the case of $Q\left[  f_{0}\right]  $ this will be a consequence of
(\ref{S1E5}), (\ref{S1E6}). In order to derive this property for $Q\left[
h\right]  $ we will use the decay and regularity properties of the functions
$h\in\mathcal{Z}_{\bar{p}}^{\sigma;\frac{1}{2}}\left(  T\right)  .$ The
details will be given in Section \ref{SectionQh}.

After obtaining $\tilde{h}=\tilde{h}\left[  \Lambda\right]  $ we proceed to
determine $\Lambda\left(  \tau\right)  .$ To this end we will argue as
follows. The asymptotic behaviour of $\tilde{h}$ as $x\rightarrow\infty$ is
given by:%
\begin{equation}
\tilde{h}\left(  \tau,x\right)  \sim\left[  \mathcal{G}\left[  \tau
;h,\Lambda\right]  -\int_{0}^{\tau}a\left(  \tau-s\right)  \Lambda_{\tau
}\left(  s\right)  ds\right]  x^{-\frac{3+\lambda}{2}}\ \ \text{as\ \ }%
x\rightarrow\infty\ \ ,\ \ 0\leq\tau\leq T \label{S2E5}%
\end{equation}
where $a\left(  \cdot\right)  $ is a function depending on $f_{0}$ and
$\mathcal{G}\left[  \cdot;h,\Lambda\right]  $ a functional that will be
precised later (cf. Proposition \ref{Propositionh_1}, Proposition
\ref{Proposition_psi} and Lemma \ref{Ash2} for a precise formulation of this result).

The asymptotics (\ref{S2E5}) will be obtained using the properties of the
fundamental solution constructed in \cite{EV2}. In order to close the fixed
point argument, we need to choose $\Lambda\left(  \tau\right)  $ in such a way
that $\tilde{h}\left(  \tau,x\right)  =o\left(  x^{-\frac{3+\lambda}{2}%
}\right)  $ as $x\rightarrow\infty.$ This can be achieved assuming that
$\Lambda$ solves the equation:%
\begin{equation}
\int_{0}^{\tau}a\left(  \tau-s\right)  \Lambda_{\tau}\left(  s\right)
ds-\mathcal{G}\left[  \tau;h,\Lambda\right]  =0\ \ ,\ \ 0\leq\tau\leq T
\label{S2E6}%
\end{equation}

A detailed analysis of the function $a\left(  \tau\right)  $ (see Subsection
\ref{Functionh_2}) will allow to transform (\ref{S2E6}) in something more like
a first order Volterra integral equation:%
\begin{equation}
a\left(  0\right)  \Lambda\left(  \tau\right)  -\int_{0}^{\tau}\frac{da}%
{d\tau}\left(  \tau-s\right)  \Lambda\left(  s\right)  ds-a\left(
\tau\right)  -\mathcal{G}\left[  \tau;h,\Lambda\right]  =0\ \ ,\ \ 0\leq
\tau\leq T \label{S2E7}%
\end{equation}
with $a\left(  0\right)  =1.$ This equation can be solved by means of a
standard fixed point argument, and this gives the desired $\Lambda$ that will
be denoted as $\tilde{\Lambda}.$ We then define $\mathcal{T}\left[  h\right]
=\tilde{h}\left[  \tilde{\Lambda}\right]  .$ Notice that $\mathcal{T}\left[
h\right]  \left(  \tau,x\right)  =o\left(  x^{-\frac{3+\lambda}{2}}\right)  $
as $x\rightarrow\infty.$ Actually, a more careful analysis of (\ref{S2E4}),
(\ref{S2E7}) shows that $\mathcal{T}\left[  h\right]  \in\mathcal{Z}_{\bar{p}%
}^{\sigma;\frac{1}{2}}\left(  T\right)  .$ Moreover, the operator
$\mathcal{T}$ is contractive in $\mathcal{Z}_{\bar{p}}^{\sigma;\frac{1}{2}%
}\left(  T\right)  $ if $T$ is sufficiently small and a suitable choice of
$\bar{\delta}$.

\section{\label{Summary}SUMMARY OF SOME OF THE RESULTS IN \cite{EV1},
\cite{EV2}.}

We recall in this Section several results that have been obtained in
\cite{EV1}, \cite{EV2} and that will be used repeatedly in this paper.

In order to study the asymptotic behaviour of $\tilde{h}$ defined in the
previous Section, we will need some properties of the semigroup defined by the
operator:%
\begin{align}
L(h)  &  =\int_{0}^{\frac{x}{2}}\left[  (x-y)^{\lambda/2}G(x-y)-x^{\lambda
/2}G(x)\right]  y^{\lambda/2}h(y)\,dy+\label{S3E1}\\
&  +\int_{0}^{\frac{x}{2}}\left[  (x-y)^{\lambda/2}h(x-y)-x^{\lambda
/2}h(x)\right]  y^{-\frac{3}{2}}dy-x^{-\frac{3}{2}}\int_{\frac{x}{2}}^{\infty
}y^{\lambda/2}h(y)dy-2\sqrt{2}x^{\frac{\lambda-1}{2}}h(x)\nonumber
\end{align}
where $G\left(  x\right)  =\frac{1}{x^{\frac{3+\lambda}{2}}}.$ We have studied
in \cite{EV1} the solution of the following problem:%
\begin{equation}
\partial_{\tau}g\left(  \tau,x\right)  =L\left[  g\right]  \left(
\tau,x\right)  \ \ \ ,\ \ \ x>0\ \ ,\ \ \tau>0\ \ ,\ \ g\left(  0,x,x_{0}%
\right)  =\delta\left(  x-x_{0}\right)  \label{S3E3}%
\end{equation}

In particular we have proved there the following results:

\begin{theorem}
[cf. Theorem 3.8 in \cite{EV1}]\label{ThFS}There exists a unique solution
$g(\tau,\cdot,x_{0})\in C^{\infty}\left(  \mathbb{R}^{+}\right)  $ of
(\ref{S3E3}) that has the following properties. There exist $\varepsilon
_{1}>0$ and $\varepsilon_{2}>0$ depending only on $\lambda$ such that, for any
$0<\varepsilon<\varepsilon_{1}$ the following statements hold.

The function $g(\tau,\cdot,x_{0})$ has the following self-similar structure:
\begin{equation}
g(\tau,x,x_{0})=\frac{1}{x_{0}}g\left(  \tau x_{0}^{\frac{\lambda-1}{2}}%
,\frac{x}{x_{0}},1\right)  \label{T1E0a}%
\end{equation}

For all $\tau\geq1:$
\begin{equation}
g(\tau,x,1)=\tau^{\frac{2}{\lambda-1}}\varphi_{1}(\rho)+\varphi_{2}(\tau
,\rho)\ \ ,\ \ \ \rho=\tau^{\frac{2}{\lambda-1}}x \label{T1E0}%
\end{equation}
with:%
\begin{equation}
\varphi_{1}(\mathcal{\rho})=\left\{
\begin{array}
[c]{c}%
a_{1}\mathcal{\rho}^{-\frac{3}{2}}+O_{\varepsilon}\left(  \mathcal{\rho
}^{-\frac{4-\lambda}{2}+\varepsilon}\right)  \ \ ,\ \ 0\leq\mathcal{\rho}%
\leq1\\
a_{2}\mathcal{\rho}^{-\frac{3+\lambda}{2}}+O_{\varepsilon}\left(
\mathcal{\rho}^{-\left(  1+\lambda-\varepsilon\right)  }\right)
\ \ ,\ \ \mathcal{\rho}>1
\end{array}
\right.  \label{T1E1}%
\end{equation}
where $a_{1},\ a_{2}$ are two explicit constants.
\begin{equation}
\varphi_{2}(\tau,\mathcal{\rho})=\left\{
\begin{array}
[c]{c}%
b_{1}\left(  \tau\right)  \mathcal{\rho}^{-\frac{3}{2}}+O\left(  \tau
^{\frac{2}{\lambda-1}-\varepsilon_{2}}\mathcal{\rho}^{-\frac{3}{2}%
+\varepsilon_{2}}\right)  \ \ ,\ \ 0\leq\mathcal{\rho}\leq1\\
b_{2}\left(  \tau\right)  \mathcal{\rho}^{-\frac{3+\lambda}{2}}+O\left(
\tau^{\frac{2}{\lambda-1}-\varepsilon_{2}}\mathcal{\rho}^{-\frac{3+\lambda}%
{2}-\varepsilon_{2}}\right)  \ \ ,\ \ \mathcal{\rho}>1
\end{array}
\right.  \label{T1E2}%
\end{equation}
where $b_{1},b_{2}\in$ are two continuous functions such that $\left\vert
b_{1}\left(  \tau\right)  \right\vert +\left\vert b_{2}\left(  \tau\right)
\right\vert \leq C\tau^{\frac{2}{\lambda-1}-\varepsilon_{2}}.$

For $0<\tau\leq1$ we have:%
\begin{equation}
g(\tau,x,1)=\left\{
\begin{array}
[c]{c}%
\tau x^{-\frac{3}{2}}+b_{3}\left(  \tau\right)  x^{-\frac{3}{2}}+O\left(  \tau
x^{-\frac{3}{2}+\varepsilon_{2}}\right)  \ \ ,\ \ 0\leq x\leq\frac{1}{2}\\
a_{3}\tau x^{-\frac{3+\lambda}{2}}+b_{4}\left(  \tau\right)  x^{-\frac
{3+\lambda}{2}}+O\left(  \tau x^{-\frac{3+\lambda}{2}-\varepsilon_{2}}\right)
\ \ ,\ \ x\geq\frac{3}{2}\\
O_{\varepsilon}\left(  \frac{t^{1-2\varepsilon}}{\left\vert x-1\right\vert
^{\frac{3}{2}-\varepsilon}}\right)  \ \ \text{for\ \ }t^{2}<\left\vert
x-1\right\vert <\frac{1}{2}%
\end{array}
\right.  \label{T1E3}%
\end{equation}
where $a_{3}$ is an explicit numerical constant and $b_{3},b_{4}$ are
continuous functions such that $\left\vert b_{3}\left(  \tau\right)
\right\vert +\left\vert b_{4}\left(  \tau\right)  \right\vert \leq
C\tau^{1+\varepsilon_{2}}.$ Moreover:%
\[
\lim_{t\rightarrow0}t^{2}g\left(  t,1+t^{2}\chi,1\right)  =\Psi\left(
\chi\right)  \ \ \text{uniformly on compact subsets of }\mathbb{R}\text{ }%
\]
where the function $\Psi$ is given by:%
\begin{equation}
\Psi\left(  \chi\right)  =\frac{2}{\pi}\frac{\exp\left(  -\frac{\pi}%
{\chi^{3/2}}\right)  }{\chi^{3/2}}\ \text{for}\ \ \chi>0\ \ ,\ \ \Psi\left(
\chi\right)  =0\ \text{for}\ \ \chi<0 \label{T1E5}%
\end{equation}

\end{theorem}

\begin{remark}
The functions $O_{\varepsilon}\left(  \cdot\right)  $ depend on $\varepsilon.$
\end{remark}

\begin{remark}
Notice that (\ref{T1E1})-(\ref{T1E3}) imply the existence of a function
$\Theta=\Theta\left(  \tau\right)  $ and $\varepsilon>0$ such that:%
\begin{align}
\left\vert g(\tau,x,1)-\Theta\left(  \tau\right)  x^{-\frac{3+\lambda}{2}%
}\right\vert  &  \leq C\tau x^{-\frac{3+\lambda}{2}-\varepsilon}%
\ \ ,\ \ \tau\leq1\ \ ,\ \ x\geq1\label{G1E1}\\
\left\vert g(\tau,x,1)-\Theta\left(  \tau\right)  x^{-\frac{3+\lambda}{2}%
}\right\vert  &  \leq\frac{C}{\tau^{\frac{\lambda+1}{\lambda-1}+\frac
{2\varepsilon}{\lambda-1}}x^{\frac{3+\lambda}{2}+\varepsilon}}\ \ ,\ \ \tau
\geq1\ \ ,\ \ x\geq1 \label{G1E2}%
\end{align}
where:%
\begin{equation}
\Theta\left(  \tau\right)  =\left\{
\begin{array}
[c]{c}%
a_{4}\tau+b_{4}\left(  \tau\right)  \ \ ,\ \ \left\vert b_{4}\left(
\tau\right)  \right\vert \leq C\tau^{1+\varepsilon}\ \ ,\ \ \tau\leq1\\
a_{2}\tau^{-\frac{\lambda+1}{\lambda-1}}+b_{2}\left(  \tau\right)
\ \ ,\ \ \left\vert b_{2}\left(  \tau\right)  \right\vert \leq C\tau
^{-\frac{\lambda+1}{\lambda-1}-\varepsilon}\ \ ,\ \ \tau\geq1
\end{array}
\right.  \label{G1E3}%
\end{equation}

\end{remark}

We will need improved estimates for $g\left(  \tau,x,1\right)  .$ More
precisely we need to compute the next order in the expansion of $g$ as
$x\rightarrow\infty.$ To this end we obtain the representation formulas for
the function $g(\tau,x,1)$ that we have obtained in the Proof of Lemma 7.10 of
\cite{EV1}.

\begin{theorem}
[cf. Lemma 5.1 in \cite{EV1}]\label{ReFor}The function $g\left(
\tau,x,1\right)  $ described in Theorem \ref{ThFS} can be written as $g\left(
\tau,x,1\right)  =G\left(  \tau,X\right)  \ \ \ ,\ \ \ x=e^{X}$ with:%
\begin{align}
G\left(  \tau,X\right)   &  =-\frac{\mathcal{V}\left(  2i\right)  i}%
{2\pi\left(  \lambda-1\right)  }e^{-\frac{3+\lambda}{2}X}\int
_{\operatorname{Im}\left(  Y\right)  =-\gamma_{1}}dY\frac{\tau^{-\frac
{2iY}{\lambda-1}}}{\mathcal{V}\left(  \frac{\left(  3+\lambda\right)  i}%
{2}+Y\right)  }\Gamma\left(  \frac{2iY}{\lambda-1}\right)  +\label{S3E4}\\
&  +\frac{i}{\pi\left(  \lambda-1\right)  }\int_{\operatorname{Im}\left(
\xi\right)  =\beta}d\xi e^{i\xi X}\int_{\operatorname{Im}\left(  Y\right)
=-\gamma_{1}}dY\frac{\mathcal{V}\left(  \xi\right)  \tau^{-\frac{2iY}%
{\lambda-1}}}{\mathcal{V}\left(  \xi+Y\right)  }\Gamma\left(  \frac
{2iY}{\lambda-1}\right) \nonumber
\end{align}
where $\left(  \beta-\frac{3+\lambda}{2}\right)  >0$ and $\gamma_{1}>0$ are
sufficiently small. The function $\mathcal{V}\left(  \xi\right)  $ is given
by:
\begin{align*}
\mathcal{V}\left(  \xi\right)   &  =\exp\left(  -\frac{2i}{\lambda-1}%
\int_{\operatorname{Im}\left(  \xi\right)  =\beta_{1}}\log\left(  -\Phi\left(
\eta\right)  \right)  \left[  \frac{1}{1-e^{\frac{4\pi\left(  \xi-\eta\right)
}{\lambda-1}}}-\frac{1}{1+e^{-\frac{4\pi\eta}{\lambda-1}}}\right]
d\eta\right)  ,\ \beta_{1}\in\left(  \frac{2+\lambda}{2},\frac{3+\lambda}%
{2}\right) \\
\ \ \Phi\left(  \eta\right)   &  =-\frac{2\sqrt{\pi}\ \Gamma\left(
i\eta+1+\frac{\lambda}{2}\right)  }{\Gamma\left(  i\eta+\frac{\lambda+1}%
{2}\right)  }\ \ ,\ \ \lim_{\operatorname{Re}\left(  \eta\right)
\rightarrow\infty}\arg\left(  -\Phi\left(  \eta\right)  \right)  =\frac{\pi
}{4}%
\end{align*}

\end{theorem}

On the other hand we have proved the following results in \cite{EV2}:

\begin{theorem}
[cf Theorem 2.1 in \cite{EV2}]\label{ExistenceLredonda} For any $\sigma
\in(1,2)$, $\bar{\delta}>0$ and any $f_{0}$ satisfying (\ref{S1E5}),
(\ref{S1E6})\ there exists $T>0$ such that for all $\mu\in Y_{3/2,\,2+\bar
{\delta}}^{\sigma}$ the Cauchy problem
\begin{equation}
h_{\tau}={\mathcal{L}}_{f_{0}}(h)+\mu\ \ \ \ ,\ \ \ \ h\left(  0\right)
=0\ \label{S3E5}%
\end{equation}
has a unique solution $h$ in $\mathcal{E}_{T;\,\sigma}$. Moreover
$|||h|||_{\sigma}\leq C||\mu||_{Y_{3/2,\,\left(  2+\bar{\delta}\right)
}^{\sigma}}$ for some \ positive constant $C$ depending on $T$, $\sigma$,
$\bar{\delta}$ as well as $A$, $B$ and $\gamma$ in (\ref{S1E5}), (\ref{S1E6})
but not on $\mu$.
\end{theorem}

\begin{theorem}
[cf. Theorem 2.2 in \cite{EV2}]\label{RegularityLredonda} For any $\sigma
\in(1,2)$, $\bar{\delta}>0$ and for any $f_{0}$ satisfying (\ref{S1E5}),
(\ref{S1E6}), the solution of the Cauchy problem (\ref{S3E5}) obtained in
Theorem \ref{ExistenceLredonda} satisfies%
\[
\lbrack h]_{\frac{3+\lambda}{2}}^{\sigma;\frac{1}{2}}\leq C||\mu
||_{Y_{3/2,\,2+\bar{\delta}}^{\sigma}}%
\]
for some \ positive constant $C$ depending on $T$, $\sigma$, $\bar{\delta}$ as
well as $A$, $B$ and $\gamma$ in (\ref{S1E5}), (\ref{S1E6}) but not on $\mu$.
\end{theorem}

\bigskip

This is a regularity result proved in \cite{EV2} that will be used repeatedly
in the following:

\begin{theorem}
[cf. Theorem 3.1 in \cite{EV2}]\label{S8T3-101} (i) Suppose that $Q\in
L_{t}^{2}(0,1;H_{x}^{\sigma}(1/2,2))$, $P\in L_{t}^{2}(0,1;H_{x}^{\sigma
-1/2}(1/2,2))$ with $\sigma\in(1/2,2)$, $\kappa\in(0,1]$ and $f\in L^{\infty
}((1/4,2)\times(0,1))\cap L^{2}(0,1;H^{1/2}(1/4,2))\cap H^{1}(0,1;L^{2}%
(1/4,2))$ is such that $f=0$ if $x<1/8$ or $x>4$ and satisfies
\[
\frac{\partial f}{\partial t}=\kappa\,T_{\varepsilon,R}\left(  M_{\lambda
/2}\,f\right)  +Q+P
\]
for all $x\in(1/4,2)$, $t\in(0,1)$ and $f(x,0)=0$. Then:
\begin{equation}
||f||_{L_{t}^{2}(0,1;H_{x}^{\sigma}(3/4,5/4))}\leq C\left(  ||Q||_{L_{t}%
^{2}(0,1;H_{x}^{\sigma}(1/2,2))}+\frac{1}{\varepsilon\,\kappa}||P||_{L_{t}%
^{2}(0,1;H_{x}^{\sigma-1/2}(1/2,2))}+||f||_{L^{\infty}((1/4,2)\times
(0,1))}\right)  \nonumber\label{S8T3-10198756}%
\end{equation}
for some positive constant $C$ independent of $\varepsilon$ and $R$.

(ii) Suppose that $Q\in L_{t}^{2}(0,T_{max};H_{x}^{\sigma}(1/2,2))$, $P\in
L_{t}^{2}(0,T_{max};H_{x}^{\sigma-1/2}(1/2,2))$, $f\in L^{\infty
}((1/4,2)\times(0,T_{max}))\cap C_{t}^{1}(0,T_{max};H_{x}^{1/2}(1/4,2)))$ for
some $T_{max}>0$ is such that $f=0$ if $x<1/8$ or $x>4$ and satisfies
\begin{align}
\frac{\partial f}{\partial t}  &  =T_{\varepsilon,R}\left(  M_{\lambda
/2}\,f\right)  +Q+P-a(x,t)\,f,\qquad x\in(1/4,2),\,t>0)\label{S8T3-101E143}\\
f(x,0)  &  =0 \label{S8T3-101E143MAs}%
\end{align}
for some function $a\in L^{\infty}({0,T_{max};H^{\sigma}(1/2,2)})$, $a\geq
A>0$. Then, for all $T\in\lbrack0,T_{max}-1]$:
\begin{align}
&  \sup_{0\leq T\leq T_{max}}\left(  \int_{T}^{\min(T+1,T_{max})}%
||f(t)||_{H^{\sigma}(3/4,5/4)}^{2}\,dt\right)  ^{1/2}\leq
\label{S8T3-101E143bis}\\
&  \hskip1cmC\sup_{0\leq T\leq T_{max}}\left(  \int_{T}^{\min(T+1,T_{max}%
)}||Q(t)||_{H^{\sigma}(1/2,2)}^{2}\,dt\right)  ^{1/2}+\nonumber\\
&  \hskip1.5cm+\frac{C}{\varepsilon}\sup_{0\leq T\leq T_{max}}\left(  \int
_{T}^{\min(T+1,T_{max})}||P(t)||_{H^{\sigma-1/2}(1/2,2)}^{2}\,dt\right)
^{1/2}+\nonumber\\
&  \hskip2.3cm+C\,||f||_{L^{\infty}((1/4,2)\times(0,T_{max}))}\nonumber
\end{align}

(iii) Suppose that for some $T_{max}>0$, $Q\in L_{t}^{2}(0,T_{max}%
;H_{x}^{\sigma}(1/2,2))$, $f\in L^{\infty}((1/4,2)\times(0,T_{max}))\cap
C_{t}^{1}(0,T_{max};H_{x}^{1/2}(1/4,2)))$ is such that $f=0$ if $x<1/8$ or
$x>4$ and satisfies (\ref{S8T3-101E143}) (\ref{S8T3-101E143MAs}) with $P=0$
and $\varepsilon=0$. Then
\begin{align}
&  \left(  \int_{T}^{\min(T+1,T_{max})}\int_{\mathbb{R}}|\widehat{F}%
(k,t)|^{2}|k|^{2\sigma}\min\{|k|,R\}dk\right)  ^{1/2}\leq
\nonumber\label{S8T3-101E143bisbis}\\
&  C\sup_{0\leq T\leq T_{max}}\left(  \int_{T}^{\min(T+1,T_{max}%
)}||Q(t)||_{H^{\sigma}(1/2,2)}^{2}\,dt\right)  ^{1/2}+C\,||f||_{L^{\infty
}((1/4,2)\times(0,T_{max}))}\nonumber\\
&
\end{align}
where $F(x,t)=\eta(x)\,f(x,t)$, $\eta\in C^{\infty}$ is a cutoff satisfying
$\eta\left(  x\right)  =1$ if $x\in\left(  \frac{3}{4},\frac{5}{4}\right)  $
and $\eta\left(  x\right)  =0$ if $x\notin\left(  \frac{1}{8},\frac{1}%
{4}\right)  $. The constant $C$ is independent of $R$.
\end{theorem}

\section{SOME AUXILIARY RESULTS.}

In this Section we collect two estimates that will be used in the Proof of
Theorem \ref{Th1}.

\subsection{Remarks about notation.}

We will use in the arguments several different symbols. Specific letters have
been reserved for quantities with precise meanings. We write them shortly here
as a guide for the reader.

The letter $r=\frac{\lambda-1}{2}$ will denote the first order correction to
the asymptotics of $f_{0}$ as $x\rightarrow\infty$ (cf. (\ref{U3E4a}%
)-(\ref{Z1E2b}))$.$ We will use $\delta$ to denote the exponent of the second
order correction of $f_{0}$ as $x\rightarrow\infty$ $.$ It will be assumed in
the whole paper that $\delta<\min\left\{  r,\frac{2-\lambda}{2}\right\}  $.

The parameter $\bar{\delta}$ characterizes the functional space where the
solution of the equation will be obtained (cf. Theorem \ \ref{Th1}). It will
be always assumed that $\bar{\delta}<\min\left\{  r,\delta\right\}  .$ We will
use also the notation $\bar{p}=\frac{3+\lambda}{2}+\bar{\delta}.$

The symbols $\varepsilon$ 's will be used for the fundamental solution
associated to $g_{t}=L\left[  g\right]  $ (cf. Theorem \ref{ThFS}).

We use $\sigma$ to denote the spatial regularity of the solutions. We assume
$\sigma\in\left(  1,2\right)  .$

\subsection{A general estimate for the functions in $\mathcal{Z}_{p}%
^{\sigma;\frac{1}{2}}\left(  T\right)  .$}

\begin{lemma}
\label{LemmaInteg}Suppose that $\phi\in\mathcal{Z}_{p}^{\sigma;\frac{1}{2}%
}\left(  T\right)  $ for $\sigma\in\left(  1,2\right)  ,\ p>0\ .$ Let us
define:%
\begin{equation}
\omega\left(  t,x\right)  =\int_{0}^{t}\phi\left(  s,x\right)  ds\ \ ,\ \ x\in
\mathbb{R}^{+},\ 0\leq t\leq T \label{N1}%
\end{equation}

Then, there exists $C>0$ independent of $T,\ \phi$ such that:%
\begin{equation}
\left\Vert \omega\right\Vert _{\mathcal{Z}_{p}^{\sigma;\frac{1}{2}}\left(
T\right)  }\leq4\sqrt{T}\left\Vert \phi\right\Vert _{\mathcal{Z}_{p}%
^{\sigma;\frac{1}{2}}\left(  T\right)  }\ \ ,\ \ \label{N2}%
\end{equation}

\end{lemma}

\begin{proof}
Due to (\ref{S1E12a}) to estimate $\left\Vert \omega\right\Vert _{\mathcal{Z}%
_{p}^{\sigma;\frac{1}{2}}\left(  T\right)  }$ we need to obtain bounds for
$\left\Vert \omega\right\Vert _{L^{2}\left(  \left(  0,T\right)
;H_{x}^{\sigma}\left(  0,2\right)  \right)  },$ $\left[  \omega\right]
_{p}^{\sigma;\frac{1}{2}},$ $\sup_{0\leq t\leq T}\left\vert \left\vert
\left\vert \omega\right\vert \right\vert \right\vert _{\frac{3}{2},p},$
$||\omega||_{Y_{\frac{3}{2},p}^{\sigma}\left(  T\right)  }.$ Using (\ref{N1})
and Cauchy-Schwartz we obtain:%
\begin{equation}
\left\Vert \omega\right\Vert _{L^{2}\left(  \left(  0,T\right)  ;H_{x}%
^{\sigma}\left(  0,2\right)  \right)  }\leq T\left\Vert \phi\right\Vert
_{L^{2}\left(  \left(  0,T\right)  ;H_{x}^{\sigma}\left(  0,2\right)  \right)
} \label{N3}%
\end{equation}

Using (\ref{M2E8}):%
\begin{equation}
\sup_{0\leq t\leq T}\left\vert \left\vert \left\vert \omega\right\vert
\right\vert \right\vert _{\frac{3}{2},p}\leq T\left\vert \left\vert \left\vert
\phi\right\vert \right\vert \right\vert _{\frac{3}{2},p} \label{N4}%
\end{equation}

To estimate $||\omega||_{Y_{\frac{3}{2},p}^{\sigma}\left(  T\right)  }$ we
need to control $N_{2;\sigma}\left(  \omega;t_{0},R\right)  ,$ $M_{2;\sigma
}\left(  \omega;R\right)  $ (cf. (\ref{S1E8}), (\ref{S1E10})). Using again
Cauchy-Schwartz inequality we arrive at:%
\begin{equation}
N_{2;\sigma}\left(  \omega;t_{0},R\right)  \leq\sqrt{T}N_{2;\sigma}\left(
\phi;t_{0},R\right)  \ ,\ R>1\ \ ,\ \ \ M_{2;\sigma}\left(  \omega;R\right)
\leq\sqrt{T}M_{2;\sigma}\left(  \phi;R\right)  \ \ ,\ \ R\leq1 \label{N5}%
\end{equation}

Finally we can estimate $\left[  \omega\right]  _{p}^{\sigma;\frac{1}{2}}$
using also Cauchy-Schwartz for each value of $R$ (cf. (\ref{S1E11})):%
\begin{equation}
\left[  \omega\right]  _{p}^{\sigma;\frac{1}{2}}\leq\sqrt{T}\left[
\phi\right]  _{p}^{\sigma;\frac{1}{2}} \label{N6}%
\end{equation}
where we use that $t_{0}+\theta R^{-\frac{\left(  \lambda-1\right)  }{2}}$
(cf. (\ref{S1E10a})) is bounded by $T.$ Combining (\ref{N2})-(\ref{N5}) we
obtain (\ref{N2}).
\end{proof}

\subsection{Improved estimates for $g\left(  \tau,x,1\right)  .$}

We will need to compute detailed asymptotics for the function $g\left(
\tau,x,1\right)  $ in Theorem \ref{ThFS} as $x\rightarrow\infty,$ since the
main corrective terms coming from the asymptotics of $g\left(  \tau
,x,1\right)  $ have the same order of magnitude as the ones due to the natural
sources in the problem for the approach indicated in Section \ref{strategy}.

\begin{proposition}
\label{PropositionImproved}Let $g\left(  \tau,x,1\right)  $ be as in Theorem
\ref{ThFS}. Suppose that $\tau\geq1.$ Then:%
\[
g(\tau,x,1)=\tau^{\frac{2}{\lambda-1}}\varphi_{1}(\mathcal{\rho})+\varphi
_{2}(\tau,\mathcal{\rho})\ \ ,\ \ \ \mathcal{\rho}=\tau^{\frac{2}{\lambda-1}%
}x
\]
with:%
\begin{equation}
\varphi_{1}(\mathcal{\rho})=a_{2}\mathcal{\rho}^{-\frac{3+\lambda}{2}}%
+a_{5}\mathcal{\rho}^{-\left(  \frac{3+\lambda}{2}+r\right)  }+O\left(
\mathcal{\rho}^{-\left(  1+\lambda+\varepsilon_{1}\right)  }\right)
\ \ ,\ \ \mathcal{\rho}>1 \label{Z3E1}%
\end{equation}
for some $\varepsilon_{1}>0$. Moreover:%
\begin{equation}
\varphi_{2}(\tau,\mathcal{\rho})=b_{2}\left(  \tau\right)  \mathcal{\rho
}^{-\frac{3+\lambda}{2}}+O_{\varepsilon_{2}}\left(  \tau^{\frac{2}{\lambda
-1}-\varepsilon_{2}}\mathcal{\rho}^{-\left(  \frac{3+\lambda}{2}+r\right)
}\right)  \ \ ,\ \ \mathcal{\rho}>1 \label{Z3E1a}%
\end{equation}
where $\left\vert b_{2}\left(  \tau\right)  \right\vert \leq C_{\varepsilon
_{2}}\tau^{\frac{2}{\lambda-1}-\varepsilon_{2}}$ for any $\varepsilon_{2}>0$.

Suppose now that $\tau\leq1.$ Then:
\begin{equation}
g(\tau,x,1)=a_{3}\tau x^{-\frac{3+\lambda}{2}}+b_{4}\left(  \tau\right)
x^{-\frac{3+\lambda}{2}}+O\left(  \tau x^{-\left(  \frac{3+\lambda}%
{2}+r\right)  }\right)  \ \ ,\ \ x\geq\frac{3}{2} \label{Z3E1b}%
\end{equation}
where $\left\vert b_{4}\left(  \tau\right)  \right\vert \leq C\tau
^{1+\varepsilon_{3}}$ for some $\varepsilon_{3}>0$ sufficiently small.
\end{proposition}

\begin{proof}
The argument is similar to the one in \cite{EV1}. More precisely we deform the
contour of integration in (\ref{S3E4}). Crossing the singularities of the
integrand we obtain contributions using residues that yield the main terms in
the asymptotics. The only difference with the argument in \cite{EV1} is that
we have to cross also the singularity at $\xi=\left(  \frac{3+\lambda}%
{2}+r\right)  i$. This yields the second term on the right-hand side of
(\ref{Z3E1}).

More precisely. Suppose first that $\tau\geq1.$ We then use the representation
formula (cf. \cite{EV1}, Subsection 9.2):%
\[
G\left(  \tau,X\right)  =\left(  \tau\right)  ^{\frac{2}{\lambda-1}}\Psi
_{1}\left(  \theta\right)  +G_{1}\left(  \tau,X\right)  \ \ ,\ \ \theta
=X+\frac{2}{\lambda-1}\log\left(  \tau\right)
\]
where:%
\[
\Psi_{1}\left(  \theta\right)  =\frac{1}{\pi\left(  \lambda-1\right)
i\mathcal{V}\left(  \frac{\lambda+1}{2}i\right)  }\int_{\operatorname{Im}%
\left(  \xi\right)  =\beta_{2}}d\xi e^{i\xi\theta}\mathcal{V}\left(
\xi\right)  \Gamma\left(  -\frac{2i}{\lambda-1}\left(  \xi-i\right)  \right)
\]%
\[
G_{1}\left(  \tau,X\right)  =\frac{i}{\pi\left(  \lambda-1\right)  }%
\int_{\operatorname{Im}\left(  \xi\right)  =\beta_{2}}d\xi e^{i\xi\theta}%
\int_{\operatorname{Im}\left(  y\right)  =\beta_{3}}dy\frac{\mathcal{V}\left(
\xi\right)  \tau^{-\frac{2iy}{\lambda-1}}}{\mathcal{V}\left(  y\right)
}\Gamma\left(  -\frac{2i}{\lambda-1}\left(  \xi-y\right)  \right)
\]
with $\beta_{2}\in\left(  \beta_{0},\frac{3-\lambda}{2}\right)  ,$ $\beta
_{0}\in\left(  \frac{3}{2},2\right)  ,\ \beta_{3}\in\left(  \frac{3-\lambda
}{2},1\right)  .$ The asymptotics of the function $\Psi_{1}\left(
\theta\right)  $ as $\theta\rightarrow\infty$ has been obtained in \cite{EV1},
Proposition 9.8 moving upwards the contour of integration $\left\{
\operatorname{Im}\left(  \xi\right)  =\beta_{2}\right\}  $ in order to make it
cross the first singularity found of $\mathcal{V}$ at $\xi=\frac{3+\lambda}%
{2}i.$ To obtain better estimates we just move the contour of integration
above the line $\left\{  \operatorname{Im}\left(  \xi\right)  =\lambda
+1\right\}  .$ We then obtain the following generalization of formula (9.27)
in \cite{EV1}:%
\begin{align}
\Psi_{1}\left(  \theta\right)   &  =-\frac{\Gamma\left(  \frac{\lambda
+1}{\lambda-1}\right)  }{2\pi i}\frac{\mathcal{V}\left(  2i\right)
e^{-\frac{3+\lambda}{2}\theta}}{\mathcal{V}\left(  \frac{\lambda+1}%
{2}i\right)  }-\frac{2\pi i\Gamma\left(  \frac{\lambda}{\lambda-1}\right)
\mathcal{V}\left(  2i\right)  e^{-\left(  \frac{3+\lambda}{2}+r\right)
\theta}}{\pi\left(  \lambda-1\right)  i\mathcal{V}\left(  \frac{\lambda+1}%
{2}i\right)  \Phi\left(  \left(  \lambda+1\right)  i\right)  \Phi^{\prime
}\left(  \left(  \frac{3+\lambda}{2}i\right)  \right)  }+\nonumber\\
&  +\frac{1}{\pi\left(  \lambda-1\right)  i\mathcal{V}\left(  \frac{\lambda
+1}{2}i\right)  }\int_{\operatorname{Im}\left(  \xi\right)  =1+\lambda
+\varepsilon_{1}}d\xi e^{i\xi\theta}\mathcal{V}\left(  \xi\right)
\Gamma\left(  -\frac{2i}{\lambda-1}\left(  \xi-i\right)  \right)
\label{Z3E1c}%
\end{align}
with $\varepsilon>0$ small. We have computed $\operatorname*{Res}\left(
\mathcal{V};\xi=\left(  \lambda+1\right)  i\right)  $ using Proposition 4.1
and (5.11) in \cite{EV1}.

The first term on the right-hand side of (\ref{Z3E1c}) is the first one on the
right-hand side of (\ref{Z3E1}). The last one can be estimated by
$Ce^{-\left(  \frac{3+\lambda}{2}+r+\varepsilon_{1}\right)  \theta}$ for
$\theta>0.$ This gives (\ref{Z3E1}). We now estimate $G_{1}\left(
\tau,X\right)  .$ This can be made as the estimate of $G_{1}$ in Lemma 9.9 of
\cite{EV1}. Deforming the contour $\left\{  \operatorname{Im}\left(
\xi\right)  =\beta_{2}\right\}  $ as in the derivation of (9.36) of
\cite{EV1}, but moving it above the line $\left\{  \operatorname{Im}\left(
\xi\right)  =\lambda+1\right\}  $ we obtain:%
\[
G_{1}\left(  \tau,X\right)  =b_{2}\left(  \tau\right)  e^{-\frac{3+\lambda}%
{2}\theta}+\tilde{b}_{2}\left(  \tau\right)  e^{-\left(  \frac{3+\lambda}%
{2}+r\right)  \theta}+\tilde{Q}_{1}\left(  \tau,X\right)
\]
where the function $b_{2}\left(  \tau\right)  $ is exactly as in \cite{EV1},
the function $\tilde{b}_{2}\left(  \tau\right)  $ has a similar formula, with
slightly different terms arising from the integration by residues, and
$\tilde{Q}_{1}\left(  \tau,X\right)  $ is similar to (9.37) in \cite{EV1} with
the only difference that $\beta_{6}=\left(  1+\lambda\right)  +\varepsilon
_{1},$ with $\varepsilon_{1}>0.$ Arguing exactly as in \cite{EV1} we obtain:%
\begin{align*}
\left\vert b_{2}\left(  \tau\right)  \right\vert +\left\vert \tilde{b}%
_{2}\left(  \tau\right)  \right\vert  &  \leq C\left(  \tau\right)  ^{\frac
{2}{\lambda-1}-\varepsilon_{2}}\ \ \text{for\ \ }\tau\geq1\\
\left\vert \tilde{Q}_{1}\left(  \tau,X\right)  \right\vert  &  \leq C\left(
\tau\right)  ^{\frac{2}{\lambda-1}-\varepsilon_{2}}e^{-\left[  \left(
1+\lambda\right)  +\varepsilon_{1}\right]  \theta}%
\end{align*}

This gives (\ref{Z3E1a}). On the other hand, in order to derive (\ref{Z3E1b})
we argue as in \cite{EV1}, Proof of Lemma 9.10, (9.45). Indeed, moving the
contour of integration $\left\{  \operatorname{Im}\left(  \xi\right)
=\beta_{7}\right\}  $ to $\left\{  \operatorname{Im}\left(  \xi\right)
=\left(  \lambda+1+\varepsilon_{1}\right)  \right\}  ,$ with $\varepsilon
_{1}>0$ we obtain:
\begin{align}
G\left(  \tau,X\right)   &  =-\frac{\mathcal{V}\left(  2i\right)  i}%
{2\pi\left(  \lambda-1\right)  }e^{-\frac{3+\lambda}{2}X}\int
_{\operatorname{Im}\left(  Y\right)  =-\gamma_{1}}dY\frac{\tau^{-\frac
{2iY}{\lambda-1}}}{\mathcal{V}\left(  \frac{\left(  3+\lambda\right)  i}%
{2}+Y\right)  }\Gamma\left(  \frac{2iY}{\lambda-1}\right)  +\\
&  -\frac{2\mathcal{V}\left(  2i\right)  e^{-\left(  \frac{3+\lambda}%
{2}+r\right)  X}}{\left(  \lambda-1\right)  \Phi\left(  \left(  \lambda
+1\right)  i\right)  \Phi^{\prime}\left(  \left(  \frac{3+\lambda}{2}i\right)
\right)  }\int_{\operatorname{Im}\left(  Y\right)  =-\gamma_{1}}dY\frac
{\tau^{-\frac{2iY}{\lambda-1}}}{\mathcal{V}\left(  \left(  1+\lambda\right)
i+Y\right)  }\Gamma\left(  \frac{2iY}{\lambda-1}\right) \\
&  +\frac{i}{\pi\left(  \lambda-1\right)  }\int_{\operatorname{Im}\left(
\xi\right)  =\lambda+1+\varepsilon_{1}}d\xi e^{i\xi X}\int_{\operatorname{Im}%
\left(  Y\right)  =-\gamma_{1}}dY\frac{\mathcal{V}\left(  \xi\right)
\tau^{-\frac{2iY}{\lambda-1}}}{\mathcal{V}\left(  \xi+Y\right)  }\Gamma\left(
\frac{2iY}{\lambda-1}\right) \nonumber
\end{align}

The time dependence of the integral terms can be obtained as in \cite{EV1},
since this one comes from the integration in the $Y$ variable.
\end{proof}

\bigskip

\section{ESTIMATES FOR THE QUADRATIC TERM $Q\left[  h\right]  .$%
\label{SectionQh}}

The following Lemma will be used to show smallness of the quadratic terms
$Q\left[  h\right]  .$

\begin{proposition}
\label{LemmaQuad}For any $\sigma\in\left(  1,2\right)  ,$ and any $\bar
{\delta}>0,$ there exists $C=C\left(  \sigma,\bar{\delta}\right)  $ such that
for any $h\in\mathcal{Z}_{\bar{p}}^{\sigma;\frac{1}{2}}\left(  T\right)  :$%
\[
\left\Vert Q\left[  h\right]  \right\Vert _{Y_{\frac{3}{2},\left(
2+\bar{\delta}\right)  }^{\sigma}\left(  T\right)  }\leq C\left\Vert
h\right\Vert _{\mathcal{Z}_{\bar{p}}^{\sigma;\frac{1}{2}}\left(  T\right)
}^{2}%
\]
with $Q\left[  \cdot\right]  $ as in (\ref{S1E2}), (\ref{S0E1}).
\end{proposition}

\bigskip

In order to prove Proposition \ref{LemmaQuad} we rewrite $Q\left[  h\right]  $
as:%
\begin{equation}
Q\left[  h\right]  \left(  \tau,x\right)  =I_{1}+I_{2} \label{M1E0}%
\end{equation}
where:%
\begin{align}
I_{1}  &  =-\int_{\frac{x}{2}}^{\infty}\left(  xy\right)  ^{\frac{\lambda}{2}%
}h\left(  \tau,x\right)  h\left(  \tau,y\right)  dy\label{M1E1}\\
I_{2}  &  =-\int_{\frac{x}{2}}^{\infty}\left(  xy\right)  ^{\frac{\lambda}{2}%
}h\left(  \tau,x\right)  h\left(  \tau,y\right)  dy+\int_{0}^{\frac{x}{2}%
}y^{\frac{\lambda}{2}}h\left(  \tau,y\right)  \left[  \left(  x-y\right)
^{\frac{\lambda}{2}}h\left(  \tau,x-y\right)  -x^{\frac{\lambda}{2}}h\left(
\tau,x\right)  \right]  dy \label{M1E2}%
\end{align}

We begin estimating $I_{1}:$

\begin{lemma}
\label{LQ1}Let $I_{1}$ be as in (\ref{M1E2}) and $\bar{\delta}>0$. Then:%
\begin{equation}
||I_{1}||_{Y_{\frac{3}{2},\left(  2+\bar{\delta}_{1}\right)  }^{\sigma}\left(
T\right)  }\leq C\left\Vert h\right\Vert _{\mathcal{Z}_{\bar{p}}^{\sigma
;\frac{1}{2}}\left(  T\right)  }^{2} \label{S5E2}%
\end{equation}
where $C$ is uniformly bounded for $0\leq T\leq1.$
\end{lemma}

\begin{proof}
[Proof of Lemma \ref{LQ1}]We just need to estimate the functionals
$N_{2;\,\sigma}(I_{1};\,\tau_{0},R)$ for $R\geq1$ and$\ M_{2;\sigma}%
(I_{1};\,R)\ $for $R\leq1$ (cf. (\ref{S1E8}), (\ref{S1E10})). Suppose first
that $R>1.$ We introduce the rescaling $x=RX,\ y=RY,\ \tau=\tau_{0}%
+R^{-\frac{\lambda-1}{2}}\theta,$ $H_{R}\left(  \theta,X\right)  =h\left(
\tau,x\right)  .$ Then:%
\begin{align}
I_{1}  &  =-R^{\lambda+1}X^{\frac{\lambda}{2}}H_{R}\left(  \theta,X\right)
\int_{\frac{X}{2}}^{2}Y^{\frac{\lambda}{2}}H_{R}\left(  \theta,Y\right)
dY-\label{S4E7a}\\
&  -R^{\lambda+1}X^{\frac{\lambda}{2}}H_{R}\left(  \theta,X\right)  \int
_{2}^{\infty}Y^{\frac{\lambda}{2}}H_{R}\left(  \theta,Y\right)  dY\nonumber\\
&  \equiv I_{1,1}+I_{1,2}%
\end{align}

We begin estimating $I_{1,2}.$ Notice that:
\[
\left\vert R^{\frac{\lambda}{2}+1}\int_{2}^{\infty}Y^{\frac{\lambda}{2}}%
H_{R}\left(  \theta,Y\right)  dY\right\vert \leq\frac{C}{R^{\frac{1}{2}%
+\bar{\delta}_{1}}}\left\Vert h\right\Vert _{\mathcal{Z}_{\bar{p}}%
^{\sigma;\frac{1}{2}}\left(  T\right)  }%
\]

Then:%
\begin{equation}
N_{2;\,\sigma}(I_{1,2};\,\tau_{0},R)\leq\frac{C}{R^{\left(  2+2\bar{\delta
}\right)  }}\left\Vert h\right\Vert _{\mathcal{Z}_{\bar{p}}^{\sigma;\frac
{1}{2}}\left(  T\right)  }^{2}\ ,\ \ R\geq1\ \ ,\ \ \tau_{0}\in\left[
0,T\right]  \label{M1E3}%
\end{equation}

On the other hand using the inequality:
\begin{equation}
\left\Vert fg\right\Vert _{H_{x}^{\sigma}\left(  \frac{1}{2},2\right)  }\leq
C\left(  \left\Vert f\right\Vert _{H_{x}^{\sigma}\left(  \frac{1}{2},2\right)
}\left\Vert g\right\Vert _{H_{x}^{\sigma}\left(  \frac{1}{2},2\right)
}+\left\Vert f\right\Vert _{H_{x}^{\sigma}\left(  \frac{1}{2},2\right)
}\left\Vert g\right\Vert _{H_{x}^{\sigma}\left(  \frac{1}{2},2\right)
}\right)  \ \label{M1E3a}%
\end{equation}
for $\sigma>\frac{1}{2}$ (cf. \cite{RS}, Theorem 1, Section 4.6.4, p. 221) we
obtain:%
\begin{equation}
N_{2;\,\sigma}(I_{1,1};\,\tau_{0},R)\leq\frac{C}{R^{\left(  2+2\bar{\delta
}\right)  }}\left\Vert h\right\Vert _{\mathcal{Z}_{\bar{p}}^{\sigma;\frac
{1}{2}}\left(  T\right)  }^{2}\ \ ,\ \ R\geq1\ \ ,\ \ \tau_{0}\in\left[
0,T\right]  \label{M1E4}%
\end{equation}

Therefore, combining (\ref{M1E3}), (\ref{M1E4}):%
\begin{equation}
N_{2;\,\sigma}(I_{1};\,\tau_{0},R)\leq\frac{C}{R^{\left(  2+2\bar{\delta
}\right)  }}\left\Vert h\right\Vert _{\mathcal{Z}_{\bar{p}}^{\sigma;\frac
{1}{2}}\left(  T\right)  }^{2}\ \ ,\ \ R\geq1\ \ ,\ \ \tau_{0}\in\left[
0,T\right]  \label{S4E8}%
\end{equation}

Suppose now that $R\leq1.$ We introduce now the rescaling $x=RX,\ y=RY,$
$H_{R}\left(  \tau,X\right)  =h\left(  \tau,x\right)  .$ Then:%
\begin{align*}
I_{1}  &  =-R^{\lambda+1}X^{\frac{\lambda}{2}}H_{R}\left(  \tau,X\right)
\int_{\frac{X}{2}}^{2}Y^{\frac{\lambda}{2}}H_{R}\left(  \tau,Y\right)  dY-\\
&  -R^{\lambda+1}X^{\frac{\lambda}{2}}H_{R}\left(  \tau,X\right)  \int
_{2}^{\infty}Y^{\frac{\lambda}{2}}H_{R}\left(  \tau,Y\right)  dY\\
&  \equiv I_{1,1}+I_{1,2}%
\end{align*}

Notice that:%
\begin{equation}
\left\vert R^{\frac{\lambda}{2}+1}\int_{2}^{\infty}Y^{\frac{\lambda}{2}}%
H_{R}\left(  \tau,Y\right)  dY\right\vert \leq C\left\Vert h\right\Vert
_{\mathcal{Z}_{\bar{p}}^{\sigma;\frac{1}{2}}\left(  T\right)  } \label{S4E9}%
\end{equation}

On the other hand using $\left\Vert D_{X}^{\sigma}H_{R}\left(  \tau
,\cdot\right)  \right\Vert _{L^{2}\left(  \frac{1}{2},2\right)  }%
^{2}=R^{2\sigma-1}\left\Vert D_{x}^{\sigma}h\left(  \tau,\cdot\right)
\right\Vert _{L^{2}\left(  \frac{R}{2},2R\right)  }^{2}$ we obtain:%
\begin{align}
\left(  \int_{0}^{T}\left\Vert D_{X}^{\sigma}H_{R}\left(  \tau,\cdot\right)
\right\Vert _{L^{2}\left(  \frac{1}{2},2\right)  }^{2}d\tau\right)  ^{\frac
{1}{2}}  &  \leq\left(  R^{2\sigma-1}\int_{0}^{T}\left\Vert D_{x}^{\sigma
}h\left(  \tau,\cdot\right)  \right\Vert _{L^{2}\left(  \frac{R}{2},2R\right)
}^{2}d\tau\right)  ^{\frac{1}{2}}\label{S4E10}\\
&  \leq R^{-\frac{3}{2}}\left\Vert h\right\Vert _{\mathcal{Z}_{\bar{p}%
}^{\sigma;\frac{1}{2}}\left(  T\right)  }\nonumber
\end{align}

Moreover:%
\begin{equation}
\left(  \int_{0}^{T}\left\Vert H_{R}\left(  \tau,\cdot\right)  \right\Vert
_{L^{2}\left(  \frac{1}{2},2\right)  }^{2}d\tau\right)  ^{\frac{1}{2}}\leq
CT\sup_{0\leq\tau\leq T}\left\Vert H_{R}\left(  \tau,\cdot\right)  \right\Vert
_{L^{\infty}\left(  \frac{1}{2},2\right)  }\leq CTR^{-\frac{3}{2}}\left\Vert
h\right\Vert _{\mathcal{Z}_{\bar{p}}^{\sigma;\frac{1}{2}}\left(  T\right)
}\ \label{S4E11}%
\end{equation}
henceforth, using (\ref{S4E9}), (\ref{S4E10}), (\ref{S4E11}):%
\[
R^{\frac{3}{2}}M_{2;\sigma}(I_{1,2};\,R)+\frac{R^{\frac{3}{2}}M_{2;0}%
(I_{1,2};\,R)}{T}\leq CR^{\frac{\lambda}{2}}\left\Vert h\right\Vert
_{\mathcal{Z}_{\bar{p}}^{\sigma;\frac{1}{2}}\left(  T\right)  }^{2}%
\]

On the other hand:%
\begin{align*}
&  \left(  \int_{0}^{T}\left\Vert D_{X}^{\sigma}\left[  \int_{\frac{X}{2}}%
^{2}Y^{\frac{\lambda}{2}}H_{R}\left(  \tau,Y\right)  dY\right]  \right\Vert
_{L^{2}\left(  \frac{1}{2},2\right)  }^{2}d\tau\right)  ^{\frac{1}{2}}+\\
&  +\sup_{0\leq\tau\leq T}\left\Vert \int_{\frac{X}{2}}^{2}Y^{\frac{\lambda
}{2}}H_{R}\left(  \tau,Y\right)  dY\right\Vert _{L^{\infty}\left(  \frac{1}%
{2},2\right)  }\\
&  \leq CR^{-\frac{3}{2}}\left\Vert h\right\Vert _{\mathcal{Z}_{\bar{p}%
}^{\sigma;\frac{1}{2}}\left(  T\right)  }%
\end{align*}
where we just estimate the $L^{2}$ norm of $D_{X}^{\sigma}$ combining the
estimates of the derivative and the function itself by interpolation. Then,
using also (\ref{S4E10}), (\ref{S4E11}) we obtain $R^{\frac{3}{2}}M_{2;\sigma
}(I_{1,1};\,R)+R^{\frac{3}{2}}M_{2;0}(I_{1,1};\,R)\leq CR^{\lambda-\frac{1}%
{2}}\left\Vert h\right\Vert _{\mathcal{Z}_{\bar{p}}^{\sigma;\frac{1}{2}%
}\left(  T\right)  }^{2}$ whence:
\begin{equation}
R^{\frac{3}{2}}M_{2;\sigma}(I_{1};\,R)+R^{\frac{3}{2}}M_{2;0}(I_{1};\,R)\leq
C\left\Vert h\right\Vert _{\mathcal{Z}_{\bar{p}}^{\sigma;\frac{1}{2}}\left(
T\right)  }^{2}\ \ ,\ \ R\leq1 \label{S5E1}%
\end{equation}

Combining (\ref{S4E8}), (\ref{S5E1}) we obtain (\ref{S5E2}) and the Lemma follows.
\end{proof}

\bigskip

In order to estimate $I_{2}$ we will need the following auxiliary Lemma:

\begin{lemma}
\label{RegulFourier} Suppose that $\phi\in C_{0}^{\infty}\left(
\mathbb{R}^{+}\right)  .$ There exists $C>0$ depending only on $\phi$ and its
derivatives such that the following inequality holds for any $R>1$:%
\begin{align}
&  \int_{\mathbb{R}}\left(  1+\left\vert \xi\right\vert ^{2\sigma}\right)
\left\vert \left(  \hat{\phi}\ast\hat{G}\right)  \left(  \xi\right)
\right\vert ^{2}\left(  \min\left\{  \sqrt{\left\vert \xi\right\vert }%
,\sqrt{R}\right\}  \right)  ^{2}d\xi\nonumber\\
&  \leq C\int_{\mathbb{R}}\left(  1+\left\vert \xi\right\vert ^{2\sigma
}\right)  \left\vert \hat{G}\left(  \xi\right)  \right\vert ^{2}\left(
1+\min\left\{  \sqrt{\left\vert \xi\right\vert },\sqrt{R}\right\}  \right)
^{2}d\xi\label{M1E6}%
\end{align}

\end{lemma}

The proof of Lemma \ref{RegulFourier} will be based in the following inequality:

\begin{lemma}
\label{Inequality} Let $W_{R}\left(  \xi\right)  =\min\left\{  \sqrt
{\left\vert \xi\right\vert },\sqrt{R}\right\}  =\sqrt{R}\min\left\{
\sqrt{\frac{\left\vert \xi\right\vert }{R}},1\right\}  .$ There exists a
constant $C>0$ such that, for any $R>0$ and any $\xi,\eta\in\mathbb{R}:$%
\begin{equation}
\left\vert W_{R}\left(  \xi\right)  -W_{R}\left(  \eta\right)  \right\vert
\leq C\frac{\left\vert \xi-\eta\right\vert }{\left\vert \eta\right\vert }%
W_{R}\left(  \eta\right)  \label{M1E5}%
\end{equation}

\end{lemma}

\begin{proof}
[Proof of Lemma \ref{Inequality}]This Lemma can be thought as a particular
case of Lemma 3.6 in \cite{EV2}. However, we give here an elementary proof.
Due to the scale invariance of the inequality (\ref{M1E5}) we can restrict
ourselves to the case $R=1.$ The inequality is then elementary if
$\max\left\{  \left\vert \xi\right\vert ,\left\vert \eta\right\vert \right\}
\geq1.$ Suppose then that $\max\left\{  \left\vert \xi\right\vert ,\left\vert
\eta\right\vert \right\}  \leq1.$ Then (\ref{M1E5}) reduces to $\left\vert
\sqrt{\left\vert \xi\right\vert }-\sqrt{\left\vert \eta\right\vert
}\right\vert \leq C\frac{\left\vert \xi-\eta\right\vert }{\left\vert
\eta\right\vert }\sqrt{\left\vert \eta\right\vert }$ which follows immediately
multiplying both sides of the inequality by $\left(  \sqrt{\left\vert
\xi\right\vert }+\sqrt{\left\vert \eta\right\vert }\right)  .$
\end{proof}

\bigskip

\begin{proof}
[Proof of Lemma \ref{RegulFourier}]Using the regularity properties of $\phi$
we have:%
\[
\left\vert \left(  \hat{\phi}\ast\hat{G}\right)  \left(  \xi\right)
\right\vert \leq C_{m}\int_{\left\{  \left\vert \eta\right\vert \leq1\right\}
}\frac{\left\vert \hat{G}\left(  \eta\right)  \right\vert }{1+\left\vert
\xi-\eta\right\vert ^{m}}+C_{m}\int_{\left\{  \left\vert \eta\right\vert
>1\right\}  }\frac{\left\vert \hat{G}\left(  \eta\right)  \right\vert
}{1+\left\vert \xi-\eta\right\vert ^{m}}\equiv J_{1}\left(  \xi\right)
+J_{2}\left(  \xi\right)
\]
where $m$ can be assumed to be arbitrarily large. Using then $\frac{\left(
1+\left\vert \xi\right\vert ^{\sigma}\right)  W_{R}\left(  \xi\right)
}{1+\left\vert \xi-\eta\right\vert ^{m}}\leq\frac{C}{1+\left\vert
\xi\right\vert }$ for $\left\vert \eta\right\vert \leq1$ we obtain:%
\begin{align*}
&  \int_{\mathbb{R}}\left(  1+\left\vert \xi\right\vert ^{2\sigma}\right)
\left\vert \left(  \hat{\phi}\ast\hat{G}\right)  \left(  \xi\right)
\right\vert ^{2}\left(  \min\left\{  \sqrt{\left\vert \xi\right\vert }%
,\sqrt{R}\right\}  \right)  ^{2}d\xi\\
&  \leq C\int_{\left\{  \left\vert \eta\right\vert \leq1\right\}  }\left\vert
\hat{G}\left(  \eta\right)  \right\vert ^{2}d\eta+\\
&  +\int_{\mathbb{R}}\left(  1+\left\vert \xi\right\vert ^{2\sigma}\right)
\left(  1+\left(  W_{R}\left(  \xi\right)  \right)  ^{2}\right)  \left(
J_{2}\left(  \xi\right)  \right)  ^{2}d\xi\\
&  \leq C\int_{\mathbb{R}}\left(  1+\left\vert \xi\right\vert ^{2\sigma
}\right)  \left(  1+\left(  W_{R}\left(  \xi\right)  \right)  ^{2}\right)
\left\vert \hat{G}\left(  \xi\right)  \right\vert ^{2}d\xi+\\
&  +C\int_{\mathbb{R}}\left(  1+\left\vert \xi\right\vert ^{\sigma}\right)
^{2}\left(  W_{R}\left(  \xi\right)  \right)  ^{2}\left(  J_{2}\left(
\xi\right)  \right)  ^{2}d\xi\\
&  \equiv K_{1}+K_{2}%
\end{align*}

In order to estimate $K_{2}$ we use Lemma \ref{Inequality} to obtain for
$\left\vert \eta\right\vert \geq1:$
\begin{align*}
&  \left\vert \left(  1+\left\vert \xi\right\vert ^{\sigma}\right)
W_{R}\left(  \xi\right)  -\left(  1+\left\vert \eta\right\vert ^{\sigma
}\right)  W_{R}\left(  \eta\right)  \right\vert \\
&  \leq C\frac{\left\vert \xi-\eta\right\vert }{\left\vert \eta\right\vert
+1}W_{R}\left(  \eta\right)  +C\frac{\left\vert \xi-\eta\right\vert
}{\left\vert \eta\right\vert ^{1-\sigma}+1}W_{R}\left(  \eta\right)
+C\left\vert \xi-\eta\right\vert ^{\sigma}W_{R}\left(  \xi\right) \\
&  \leq C\left\vert \xi-\eta\right\vert W_{R}\left(  \eta\right)  +C\left\vert
\xi-\eta\right\vert ^{\sigma}W_{R}\left(  \eta\right)  +C\frac{\left\vert
\xi-\eta\right\vert ^{\sigma+1}}{\left\vert \eta\right\vert +1}W_{R}\left(
\eta\right) \\
&  \leq C\left(  \left\vert \xi-\eta\right\vert ^{\sigma}+\left\vert \xi
-\eta\right\vert ^{\sigma+1}\right)  W_{R}\left(  \eta\right)
\end{align*}

Using this inequality to estimate the terms in $W_{R}\left(  \xi\right)  \cdot
J_{2}\left(  \xi\right)  $ and using Young's inequality, as well as the fact
that the integration $J_{2}\left(  \xi\right)  $ takes place in $\left\vert
\eta\right\vert \geq1$ we obtain:%
\[
K_{2}\leq C\int_{\mathbb{R}}\left(  1+\left\vert \xi\right\vert ^{2\sigma
}\right)  \left(  W_{R}\left(  \xi\right)  \right)  ^{2}\left\vert \hat
{G}\left(  \xi\right)  \right\vert ^{2}d\xi
\]
whence Lemma \ref{RegulFourier} follows.
\end{proof}

We now estimate $I_{2}.$ The bounds for this operator are the crucial step in
the argument from the point of view of the regularity of the functions,
because this operator can be estimated as some regularized version of the
half-derivative operator. It will be essential to use the seminorm $\left[
\cdot\right]  _{\frac{3+\lambda}{2}}^{\sigma;\frac{1}{2}}$ (cf. (\ref{S1E11})).

\begin{lemma}
\label{LQ2}Suppose that $I_{2}$ is as in (\ref{M1E2}) and $\bar{\delta}>0$.
Then:%
\[
\left\Vert I_{2}\right\Vert _{Y_{\frac{3}{2},2+\bar{\delta}}^{\sigma}}\leq
CR^{-\left(  2+\bar{\delta}\right)  }\left\Vert h\right\Vert _{\mathcal{Z}%
_{\bar{p}}^{\sigma;\frac{1}{2}}\left(  T\right)  }^{2}%
\]

\end{lemma}

\begin{proof}
[Proof of Lemma \ref{LQ2}]Suppose first that $R\geq1.$ Using the rescaling
$x=RX,\ y=RY,\ \tau=\tau_{0}+R^{-\frac{\lambda-1}{2}}\theta,\ h\left(
\tau,x\right)  =R^{-\left(  \frac{3+\lambda}{2}+\bar{\delta}\right)  }%
G_{R}\left(  \theta,X\right)  $%
\begin{equation}
I_{2}=R^{-\left(  2+2\bar{\delta}\right)  }\int_{0}^{\frac{X}{2}}%
Y^{\frac{\lambda}{2}}G_{R}\left(  \theta,Y\right)  \left[  \left(  X-Y\right)
^{\frac{\lambda}{2}}G_{R}\left(  \theta,X-Y\right)  -X^{\frac{\lambda}{2}%
}G_{R}\left(  \theta,X\right)  \right]  dY \label{S5E4}%
\end{equation}

Notice that:%
\begin{equation}
\left\vert G_{R}\left(  \theta,Y\right)  \right\vert \leq\left\Vert
h\right\Vert _{\mathcal{Z}_{\bar{p}}^{\sigma;\frac{1}{2}}\left(  T\right)
}\min\left(  Y^{-\left(  \frac{3+\lambda}{2}+\bar{\delta}\right)  }%
,R^{\frac{\lambda}{2}+\bar{\delta}}Y^{-\frac{3}{2}}\right)  \label{S5E3}%
\end{equation}

We rewrite (\ref{S5E4}) as:,%
\begin{align*}
I_{2}  &  =I_{2,-}+\sum_{\left\{  k=0,1,...;\frac{2^{k}}{R}\leq\frac{1}%
{4}\right\}  }I_{2,k}+I_{2,+}\\
I_{2,-}  &  =R^{-\left(  2+2\bar{\delta}\right)  }\int_{0}^{\frac{1}{R}%
}Y^{\frac{\lambda}{2}}G_{R}\left(  \theta,Y\right)  J\left(  G_{R}%
;\theta,X,Y\right)  dY\\
I_{2,k}  &  =R^{-\left(  2+2\bar{\delta}\right)  }\int_{\frac{2^{k-1}}{R}%
}^{\frac{2^{k}}{R}}Y^{\frac{\lambda}{2}}G_{R}\left(  \theta,Y\right)  J\left(
G_{R};\theta,X,Y\right)  dY\ \ ,\ \ k=0,1,...\\
I_{2,+}  &  =R^{-\left(  2+2\bar{\delta}\right)  }\int_{\frac{2^{k_{\max}}}%
{R}}^{\frac{X}{2}}Y^{\frac{\lambda}{2}}G_{R}\left(  \theta,Y\right)  J\left(
G_{R};\theta,X,Y\right)  dY
\end{align*}%
\begin{equation}
J\left(  G;\theta,X,Y\right)  =\left[  \left(  X-Y\right)  ^{\frac{\lambda}%
{2}}G\left(  \theta,X-Y\right)  \eta\left(  X-Y\right)  -X^{\frac{\lambda}{2}%
}G\left(  \theta,X\right)  \eta\left(  X\right)  \right]  \label{S5E3a}%
\end{equation}
where $\frac{2^{k_{\max}}}{R}\leq\frac{1}{4}<\frac{2^{k_{\max}+1}}{R}$ and
$\eta\left(  X\right)  $ is the cutoff function used in (\ref{S1E10a}).
(Notice that $\eta\left(  X\right)  =1$ in all the regions of integration,
since $X\in\left(  \frac{1}{2},2\right)  $). Let us write $\psi_{R}\left(
\theta,X\right)  =X^{\frac{\lambda}{2}}G_{R}\left(  \theta,X\right)
\eta\left(  X\right)  .$ In order to estimate these terms in $H_{X}^{\sigma
}\left(  \frac{1}{2},2\right)  $ we use Fourier:%
\[
\psi_{R}\left(  \theta,X\right)  =\frac{1}{\sqrt{2\pi}}\int_{\mathbb{R}}%
\hat{\psi}_{R}\left(  \theta,\xi\right)  e^{i\xi X}d\xi\ \
\]

Since the functions $I_{2,-}\left(  \theta,X\right)  ,\ I_{2,k}\left(
\theta,X\right)  ,\ I_{2,+}\left(  \theta,X\right)  $ are defined for
$X\in\mathbb{R}$ we can compute their Fourier transforms. Using the
convolution property for Fourier transforms we have:%
\[
\hat{I}_{2,k}\left(  \theta,\xi\right)  =\hat{\psi}_{R}\left(  \theta
,\xi\right)  \left[  R^{-\left(  2+2\bar{\delta}\right)  }\int_{\frac{2^{k-1}%
}{R}}^{\frac{2^{k}}{R}}Y^{\frac{\lambda}{2}}G_{R}\left(  \theta,Y\right)
\left(  e^{-i\xi Y}-1\right)  dY\right]
\]

Using (\ref{S5E3}):%
\begin{align*}
&  \left\Vert D_{X}^{\sigma}I_{2,k}\right\Vert _{L^{2}\left(  \frac{1}%
{2},2\right)  }\\
&  \leq\left\Vert D_{X}^{\sigma}I_{2,k}\right\Vert _{L^{2}\left(
\mathbb{R}\right)  }\\
&  \leq\left\Vert h\right\Vert _{\mathcal{Z}_{\bar{p}}^{\sigma;\frac{1}{2}%
}\left(  T\right)  }\frac{R^{-\left(  2+\bar{\delta}\right)  }}{\left(
2^{k-1}\right)  ^{\bar{\delta}}}\left(  \int_{\mathbb{R}}\left\vert
\xi\right\vert ^{2\sigma}\left\vert \hat{\psi}_{R}\left(  \theta,\xi\right)
\right\vert ^{2}\left(  \int_{\frac{2^{k-1}}{R}}^{\frac{2^{k}}{R}}%
\frac{\left\vert e^{-i\xi Y}-1\right\vert }{Y^{\frac{3}{2}}}dY\right)
^{2}d\xi\right)  ^{\frac{1}{2}}%
\end{align*}

We now use that:%
\[
\int_{\frac{2^{k-1}}{R}}^{\frac{2^{k}}{R}}\frac{\left\vert e^{-i\xi
Y}-1\right\vert }{Y^{\frac{3}{2}}}dY\leq C\min\left\{  \sqrt{\left\vert
\xi\right\vert },\sqrt{R}\right\}
\]
whence:%
\[
\left\Vert D_{X}^{\sigma}I_{2,k}\right\Vert _{L^{2}\left(  \frac{1}%
{2},2\right)  }\leq C\left\Vert h\right\Vert _{\mathcal{Z}_{\bar{p}}%
^{\sigma;\frac{1}{2}}\left(  T\right)  }\frac{R^{-\left(  2+\bar{\delta
}\right)  }}{\left(  2^{k-1}\right)  ^{\bar{\delta}}}\left(  \int_{\mathbb{R}%
}\left\vert \xi\right\vert ^{2\sigma}\left\vert \hat{\psi}_{R}\left(
\theta,\xi\right)  \right\vert ^{2}\left(  \min\left\{  \sqrt{\left\vert
\xi\right\vert },\sqrt{R}\right\}  \right)  ^{2}d\xi\right)  ^{\frac{1}{2}}%
\]

Using Lemma \ref{RegulFourier} it follows that:%
\begin{equation}
\left\Vert D_{X}^{\sigma}I_{2,k}\right\Vert _{L^{2}\left(  \frac{1}%
{2},2\right)  }\leq\frac{CR^{-\left(  2+\bar{\delta}\right)  }}{\left(
2^{k-1}\right)  ^{\bar{\delta}}}\left\Vert h\right\Vert _{\mathcal{Z}_{\bar
{p}}^{\sigma;\frac{1}{2}}\left(  T\right)  }^{2} \label{M1E7}%
\end{equation}

The term $I_{2,+}$ can be estimated similarly:%
\begin{equation}
\left\Vert D_{X}^{\sigma}I_{2,+}\right\Vert _{L^{2}\left(  \frac{1}%
{2},2\right)  }\leq\frac{CR^{-\left(  2+\bar{\delta}\right)  }}{\left(
2^{k_{\max}-1}\right)  ^{\bar{\delta}}}\left\Vert h\right\Vert _{\mathcal{Z}%
_{\bar{p}}^{\sigma;\frac{1}{2}}\left(  T\right)  }^{2} \label{M1E8}%
\end{equation}

We now estimate $I_{2,-}.$ A similar argument yields:%
\[
\left\Vert D_{X}^{\sigma}I_{2,-}\right\Vert _{L^{2}\left(  \frac{1}%
{2},2\right)  }\leq\left\Vert D_{X}^{\sigma}I_{2,-}\right\Vert _{L^{2}\left(
\mathbb{R}\right)  }\leq\left\Vert h\right\Vert _{\mathcal{Z}_{\bar{p}%
}^{\sigma;\frac{1}{2}}\left(  T\right)  }R^{-\left(  2+\bar{\delta}\right)
}\left(  \int_{\mathbb{R}}\left\vert \xi\right\vert ^{2\sigma}\left\vert
\hat{\psi}_{R}\left(  \theta,\xi\right)  \right\vert ^{2}\left(  \Omega
_{R}\left(  Y\right)  \right)  ^{2}d\xi\right)  ^{\frac{1}{2}}%
\]
where:%
\begin{equation}
\Omega_{R}\left(  Y\right)  =\int_{0}^{\frac{1}{R}}\frac{\left\vert e^{-i\xi
Y}-1\right\vert }{Y^{\frac{3}{2}}}\left(  RY\right)  ^{\frac{\lambda}{2}%
}dY=\sqrt{R}\int_{0}^{1}\frac{\left\vert e^{-i\frac{\xi}{R}y}-1\right\vert
}{y^{\frac{3}{2}}}\left(  y\right)  ^{\frac{\lambda}{2}}dy\leq CW_{R}\left(
\xi\right)  \label{M1E10}%
\end{equation}
with $W_{R}\left(  \xi\right)  $ as in Lemma \ref{Inequality}. The last
inequality follows computing the asymptotics of the second integral in
(\ref{M1E10}) for $\frac{\xi}{R}\rightarrow0$ and $\frac{\xi}{R}%
\rightarrow\infty.$

Therefore%
\[
\left\Vert D_{X}^{\sigma}I_{2,-}\right\Vert _{L^{2}\left(  \frac{1}%
{2},2\right)  }\leq C\left\Vert h\right\Vert _{\mathcal{Z}_{\bar{p}}%
^{\sigma;\frac{1}{2}}\left(  T\right)  }R^{-\left(  2+\bar{\delta}\right)
}\left(  \int_{\mathbb{R}}\left\vert \xi\right\vert ^{2\sigma}\left\vert
\hat{\psi}_{R}\left(  \theta,\xi\right)  \right\vert ^{2}\left(  W_{R}\left(
\xi\right)  \right)  ^{2}d\xi\right)  ^{\frac{1}{2}}%
\]
whence:%
\begin{equation}
\left\Vert D_{X}^{\sigma}I_{2,-}\right\Vert _{L^{2}\left(  \frac{1}%
{2},2\right)  }\leq C\left\Vert h\right\Vert _{\mathcal{Z}_{\bar{p}}%
^{\sigma;\frac{1}{2}}\left(  T\right)  }^{2}R^{-\left(  2+\bar{\delta}\right)
} \label{M1E9}%
\end{equation}

To conclude the proof of Lemma \ref{LQ2} it only remains to estimate the
contributions of the region where $R\leq1.$ The estimate of $R^{\frac{3}{2}%
}M_{2;\sigma}(I_{2};\,R),$ $R^{\frac{3}{2}}M_{2;0}(I_{2};\,R)$ can be made in
exactly the same way as the estimate (\ref{S5E1}) for $I_{1}.$ Notice that the
two terms in $I_{2}$ yield integrals that converge separately since\thinspace
$h\left(  \tau,y\right)  $ can be estimated as $\frac{1}{y^{\frac{3}{2}}}$ for
$y\leq1$ and then, the term $y^{\frac{\lambda}{2}-\frac{3}{2}}$ gives
integrability. Therefore:%
\begin{equation}
R^{\frac{3}{2}}M_{2;\sigma}(I_{2};\,R)+R^{\frac{3}{2}}M_{2;0}(I_{2};\,R)\leq
C\left\Vert h\right\Vert _{\mathcal{Z}_{\bar{p}}^{\sigma;\frac{1}{2}}\left(
T\right)  }^{2}\ \ ,\ \ R\leq1 \label{M1E11}%
\end{equation}

Combining (\ref{M1E7}), (\ref{M1E8}), (\ref{M1E9}), (\ref{M1E11}) Lemma
\ref{LQ2} follows.
\end{proof}

\begin{proof}
[Proof of Proposition \ref{LemmaQuad}]It is just a consequence of
(\ref{M1E1}), (\ref{M1E2}), Lemma \ref{LQ1}, Lemma \ref{LQ2}.
\end{proof}

We can also prove the following Lipschitz property for the functional
$Q\left[  \cdot\right]  :$

\begin{proposition}
\label{LemmaQuadLip}For any $\sigma\in\left(  1,2\right)  $ and any
$\bar{\delta}>0$ there exists $C=C\left(  \sigma,\bar{\delta}\right)  $ such
that for any $h_{1,}h_{2}\in\mathcal{Z}_{\bar{p}}^{\sigma;\frac{1}{2}}\left(
T\right)  :$%
\[
\left\Vert Q\left[  h_{1}\right]  -Q\left[  h_{2}\right]  \right\Vert
_{Y_{\frac{3}{2},\left(  2+\bar{\delta}\right)  }^{\sigma}\left(  T\right)
}\leq C\left(  \sum_{k=1}^{2}\left\Vert h_{k}\right\Vert _{\mathcal{Z}%
_{\bar{p}}^{\sigma;\frac{1}{2}}\left(  T\right)  }\right)  \left\Vert
h_{1}-h_{2}\right\Vert _{\mathcal{Z}_{\bar{p}}^{\sigma;\frac{1}{2}}\left(
T\right)  }%
\]
with $Q\left[  \cdot\right]  $ as in (\ref{S1E2}), (\ref{S0E1}).
\end{proposition}

\begin{proof}
We have $Q\left[  h_{1}\right]  \left(  \tau,x\right)  -Q\left[  h_{2}\right]
\left(  \tau,x\right)  =I_{1}+I_{2}$ with:%
\begin{align*}
I_{1}  &  \equiv-\left(  \int_{\frac{x}{2}}^{\infty}\left(  xy\right)
^{\frac{\lambda}{2}}h_{1}\left(  \tau,x\right)  h_{1}\left(  \tau,y\right)
dy-\int_{\frac{x}{2}}^{\infty}\left(  xy\right)  ^{\frac{\lambda}{2}}%
h_{2}\left(  \tau,x\right)  h_{2}\left(  \tau,y\right)  dy\right) \\
I_{2}  &  \equiv\left(  \int_{0}^{\frac{x}{2}}y^{\frac{\lambda}{2}}%
h_{1}\left(  \tau,y\right)  \left[  \left(  x-y\right)  ^{\frac{\lambda}{2}%
}h_{1}\left(  \tau,x-y\right)  -x^{\frac{\lambda}{2}}h_{1}\left(
\tau,x\right)  \right]  dy-\right. \\
&  \left.  -\int_{0}^{\frac{x}{2}}y^{\frac{\lambda}{2}}h_{2}\left(
\tau,y\right)  \left[  \left(  x-y\right)  ^{\frac{\lambda}{2}}h_{2}\left(
\tau,x-y\right)  -x^{\frac{\lambda}{2}}h_{2}\left(  \tau,x\right)  \right]
dy\right)
\end{align*}

To estimate $I_{1}$ we need to estimate the functionals $N_{2;\,\sigma}%
(I_{1};\,\tau_{0},R)$ for $R\geq1$ and$\ M_{2;\sigma}(I_{1};\,R)\ $for
$R\leq1$ (cf. (\ref{S1E8}), (\ref{S1E10})). Suppose first that $R>1.$ We
introduce the rescaling $x=RX,\ y=RY,\ \tau=\tau_{0}+R^{-\frac{\lambda-1}{2}%
}\theta,$ $H_{R,1}\left(  \theta,X\right)  =h_{1}\left(  \tau,x\right)
,\ H_{R,2}\left(  \theta,X\right)  =h_{2}\left(  \tau,x\right)  .$ Then:%
\begin{align}
I_{1}  &  =-R^{\lambda+1}X^{\frac{\lambda}{2}}\left[  H_{R,1}\left(
\theta,X\right)  -H_{R,2}\left(  \theta,X\right)  \right]  \int_{\frac{X}{2}%
}^{\infty}Y^{\frac{\lambda}{2}}H_{R,1}\left(  \theta,Y\right)
dY-\label{T4E7a}\\
&  -R^{\lambda+1}X^{\frac{\lambda}{2}}H_{R,2}\left(  \theta,X\right)
\int_{\frac{X}{2}}^{\infty}Y^{\frac{\lambda}{2}}\left[  H_{R,1}\left(
\theta,Y\right)  -H_{R,2}\left(  \theta,Y\right)  \right]  dY\nonumber
\end{align}

Notice that:
\begin{align*}
&  \sup_{X\in\left(  \frac{1}{2},2\right)  }\left\vert \int_{X}^{\infty
}Y^{\frac{\lambda}{2}}H_{R,1}\left(  \theta,Y\right)  dY\right\vert
+\left\Vert \int_{\left(  \cdot\right)  }^{\infty}Y^{\frac{\lambda}{2}}%
H_{R,1}\left(  \theta,Y\right)  dY\right\Vert _{H_{x}^{\sigma}\left(  \frac
{1}{2},2\right)  }\\
&  \leq\frac{C}{R^{\frac{3+\lambda}{2}+\bar{\delta}}}\left\Vert h_{1}%
\right\Vert _{\mathcal{Z}_{\bar{p}}^{\sigma;\frac{1}{2}}\left(  T\right)  }%
\end{align*}%
\begin{align*}
&  \sup_{X\in\left(  \frac{1}{2},2\right)  }\left\vert \int_{\frac{X}{2}%
}^{\infty}Y^{\frac{\lambda}{2}}\left[  H_{R,1}\left(  \theta,Y\right)
-H_{R,2}\left(  \theta,Y\right)  \right]  dY\right\vert +\\
&  +\left\Vert \int_{\left(  \cdot\right)  }^{\infty}Y^{\frac{\lambda}{2}%
}\left[  H_{R,1}\left(  \theta,Y\right)  -H_{R,2}\left(  \theta,Y\right)
\right]  dY\right\Vert _{H_{x}^{\sigma}\left(  \frac{1}{2},2\right)  }%
\leq\frac{C}{R^{\frac{3+\lambda}{2}+\bar{\delta}}}\left\Vert h_{1}%
-h_{2}\right\Vert _{\mathcal{Z}_{\bar{p}}^{\sigma;\frac{1}{2}}\left(
T\right)  }%
\end{align*}

Using (\ref{M1E3a}) in (\ref{T4E7a}) we obtain:%
\begin{equation}
N_{2;\,\sigma}(I_{1};\,\tau_{0},R)\leq\frac{C}{R^{\left(  2+2\bar{\delta
}\right)  }}\left(  \sum_{k=1}^{2}\left\Vert h_{k}\right\Vert _{\mathcal{Z}%
_{\bar{p}}^{\sigma;\frac{1}{2}}\left(  T\right)  }\right)  \left\Vert
h_{1}-h_{2}\right\Vert _{\mathcal{Z}_{\bar{p}}^{\sigma;\frac{1}{2}}\left(
T\right)  }\ \ \label{M2E1}%
\end{equation}
for$\ R\geq1\ \ ,\ \ \tau_{0}\in\left[  0,T\right]  .$ Suppose now that
$R\leq1.$ We introduce the rescaling $x=RX,\ y=RY,$ $H_{R,1}\left(
\tau,X\right)  =h_{1}\left(  \tau,x\right)  ,\ H_{R,2}\left(  \tau,X\right)
=h_{2}\left(  \tau,x\right)  .$Then:%
\[
I_{1}=I_{1,1}+I_{1,2}%
\]%
\begin{align}
I_{1,1}  &  \equiv-\left[  R^{\lambda+1}X^{\frac{\lambda}{2}}\left[
H_{R,1}\left(  \tau,X\right)  -H_{R,2}\left(  \tau,X\right)  \right]
\int_{\frac{X}{2}}^{2}Y^{\frac{\lambda}{2}}H_{R,1}\left(  \tau,Y\right)
dY+\right. \nonumber\\
&  \left.  +R^{\lambda+1}X^{\frac{\lambda}{2}}H_{R,2}\left(  \tau,X\right)
\int_{\frac{X}{2}}^{2}Y^{\frac{\lambda}{2}}\left[  H_{R,1}\left(
\tau,Y\right)  -H_{R,2}\left(  \tau,Y\right)  \right]  dY\right]
\label{M2E2}\\
I_{1,2}  &  \equiv-\left[  R^{\lambda+1}X^{\frac{\lambda}{2}}\left[
H_{R,1}\left(  \tau,X\right)  -H_{R,2}\left(  \tau,X\right)  \right]  \int
_{2}^{\infty}Y^{\frac{\lambda}{2}}H_{R,1}\left(  \tau,Y\right)  dY+\right.
\nonumber\\
&  \left.  +R^{\lambda+1}X^{\frac{\lambda}{2}}H_{R,2}\left(  \tau,X\right)
\int_{2}^{\infty}Y^{\frac{\lambda}{2}}\left[  H_{R,1}\left(  \tau,Y\right)
-H_{R,2}\left(  \tau,Y\right)  \right]  dY\right]  \label{M2E3}%
\end{align}

Notice that:%
\begin{equation}
\left\vert R^{\frac{\lambda}{2}+1}\int_{2}^{\infty}Y^{\frac{\lambda}{2}%
}H_{R,1}\left(  \tau,Y\right)  dY\right\vert \leq C\left\Vert h_{1}\right\Vert
_{\mathcal{Z}_{\bar{p}}^{\sigma;\frac{1}{2}}\left(  T\right)  } \label{T4E9}%
\end{equation}%
\begin{equation}
\left\vert R^{\frac{\lambda}{2}+1}\int_{2}^{\infty}Y^{\frac{\lambda}{2}%
}\left[  H_{R,1}\left(  \tau,Y\right)  -H_{R,2}\left(  \tau,Y\right)  \right]
dY\right\vert \leq C\left\Vert h_{1}-h_{2}\right\Vert _{\mathcal{Z}_{\bar{p}%
}^{\sigma;\frac{1}{2}}\left(  T\right)  } \label{T4E10}%
\end{equation}

On the other hand, using the definition of $H_{R,1},\ H_{R,2}$ we arrive at:%
\begin{equation}
\int_{0}^{T}\left\Vert D_{X}^{\sigma}H_{R,2}\left(  \tau,\cdot\right)
\right\Vert _{L^{2}\left(  \frac{1}{2},2\right)  }^{2}d\tau\leq\left\Vert
h_{2}\right\Vert _{\mathcal{Z}_{\bar{p}}^{\sigma;\frac{1}{2}}\left(  T\right)
}^{2} \label{U4E10}%
\end{equation}

\begin{equation}
\int_{0}^{T}\left\Vert D_{X}^{\sigma}\left[  H_{R,1}\left(  \tau,X\right)
-H_{R,2}\left(  \tau,X\right)  \right]  \right\Vert _{L^{2}\left(  \frac{1}%
{2},2\right)  }^{2}d\tau\leq\left\Vert h_{1}-h_{2}\right\Vert _{\mathcal{Z}%
_{\bar{p}}^{\sigma;\frac{1}{2}}\left(  T\right)  }^{2} \label{U4E10a}%
\end{equation}

Moreover:%
\begin{equation}
\left(  \int_{0}^{T}\left\Vert H_{R,2}\left(  \tau,\cdot\right)  \right\Vert
_{L^{2}\left(  \frac{1}{2},2\right)  }^{2}d\tau\right)  ^{\frac{1}{2}}\leq
CTR^{-\frac{3}{2}}\left\Vert h\right\Vert _{\mathcal{Z}_{\bar{p}}%
^{\sigma;\frac{1}{2}}\left(  T\right)  }\ \ \label{T4E11}%
\end{equation}%
\begin{equation}
\left(  \int_{0}^{T}\left\Vert \left[  H_{R,1}\left(  \tau,\cdot\right)
-H_{R,2}\left(  \tau,\cdot\right)  \right]  \right\Vert _{L^{2}\left(
\frac{1}{2},2\right)  }^{2}d\tau\right)  ^{\frac{1}{2}}\leq CTR^{-\frac{3}{2}%
}\left\Vert h_{1}-h_{2}\right\Vert _{\mathcal{Z}_{\bar{p}}^{\sigma;\frac{1}%
{2}}\left(  T\right)  } \label{T4E11a}%
\end{equation}
henceforth, using (\ref{M1E3a}), (\ref{M2E3})-(\ref{T4E11a}):
\begin{equation}
R^{\frac{3}{2}}M_{2;0}(I_{1,2};\,R)+R^{\frac{3}{2}}M_{2;\sigma}(I_{1,2}%
;\,R)\leq R^{\frac{\lambda}{2}}\left(  \sum_{k=1}^{2}\left\Vert h_{k}%
\right\Vert _{\mathcal{Z}_{\bar{p}}^{\sigma;\frac{1}{2}}\left(  T\right)
}\right)  \left\Vert h_{1}-h_{2}\right\Vert _{\mathcal{Z}_{\bar{p}}%
^{\sigma;\frac{1}{2}}\left(  T\right)  } \label{M2E4}%
\end{equation}

On the other hand for $R\leq1$ we have:%
\begin{align*}
&  \left(  \int_{0}^{T}\left\Vert D_{X}^{\sigma}\left[  \int_{\frac{X}{2}}%
^{2}Y^{\frac{\lambda}{2}}\left[  H_{R,1}\left(  \tau,\cdot\right)
-H_{R,2}\left(  \tau,\cdot\right)  \right]  dY\right]  \right\Vert
_{L^{2}\left(  \frac{1}{2},2\right)  }^{2}d\tau\right)  ^{\frac{1}{2}}+\\
&  +\sup_{0\leq\tau\leq T}\left\Vert \int_{\frac{X}{2}}^{2}Y^{\frac{\lambda
}{2}}H_{R,1}\left(  \tau,\cdot\right)  -H_{R,2}\left(  \tau,\cdot\right)
dY\right\Vert _{L^{\infty}\left(  \frac{1}{2},2\right)  }\\
&  \leq CR^{-\frac{3}{2}}\left\Vert h_{1}-h_{2}\right\Vert _{\mathcal{Z}%
_{\bar{p}}^{\sigma;\frac{1}{2}}\left(  T\right)  }%
\end{align*}
where we just estimate the $L^{2}$ norm of $D_{X}^{\sigma}$ by $D_{X}$ and the
function itself by interpolation. Then, using also (\ref{S4E10}),
(\ref{S4E11}) as well as the fact that $R\leq1:$%
\begin{equation}
R^{\frac{3}{2}}M_{2;\sigma}(I_{1,1};\,R)+R^{\frac{3}{2}}M_{2;0}(I_{1,1}%
;\,R)\leq C\left(  \sum_{k=1}^{2}\left\Vert h_{k}\right\Vert _{\mathcal{Z}%
_{\bar{p}}^{\sigma;\frac{1}{2}}\left(  T\right)  }\right)  \left\Vert
h_{1}-h_{2}\right\Vert _{\mathcal{Z}_{\bar{p}}^{\sigma;\frac{1}{2}}\left(
T\right)  } \label{S5E1a}%
\end{equation}

Combining (\ref{M2E4}), (\ref{S5E1a}):%
\begin{equation}
||I_{1}||_{Y_{\frac{3}{2},\left(  2+\bar{\delta}\right)  }^{\sigma}\left(
T\right)  }\leq\left(  \sum_{k=1}^{2}\left\Vert h_{k}\right\Vert
_{\mathcal{Z}_{\bar{p}}^{\sigma;\frac{1}{2}}\left(  T\right)  }\right)
C\left\Vert h_{1}-h_{2}\right\Vert _{\mathcal{Z}_{\bar{p}}^{\sigma;\frac{1}%
{2}}\left(  T\right)  } \label{T5E2}%
\end{equation}
where $C$ is uniformly bounded for $0\leq T\leq1.$

We now estimate $I_{2}.$ Suppose first that $R\geq1.$ Using the rescaling
$x=RX,$ $y=RY,$ $\tau=\tau_{0}+R^{-\frac{\lambda-1}{2}}\theta,$ $h_{1}\left(
\tau,x\right)  =R^{-\left(  \frac{3+\lambda}{2}+\bar{\delta}\right)  }%
G_{R,1}\left(  \theta,X\right)  ,$ $h_{2}\left(  \tau,x\right)  =R^{-\left(
\frac{3+\lambda}{2}+\bar{\delta}\right)  }G_{R,2}\left(  \theta,X\right)  $
and using (\ref{S5E3a}) we obtain:%
\[
I_{2}=I_{2,1}+I_{2,1}%
\]%
\begin{align*}
I_{2}  &  \equiv\left(  \int_{0}^{\frac{x}{2}}y^{\frac{\lambda}{2}}%
h_{1}\left(  \tau,y\right)  \left[  \left(  x-y\right)  ^{\frac{\lambda}{2}%
}h_{1}\left(  \tau,x-y\right)  -x^{\frac{\lambda}{2}}h_{1}\left(
\tau,x\right)  \right]  dy-\right. \\
&  \left.  -\int_{0}^{\frac{x}{2}}y^{\frac{\lambda}{2}}h_{2}\left(
\tau,y\right)  \left[  \left(  x-y\right)  ^{\frac{\lambda}{2}}h_{2}\left(
\tau,x-y\right)  -x^{\frac{\lambda}{2}}h_{2}\left(  \tau,x\right)  \right]
dy\right)
\end{align*}%
\begin{align}
I_{2,1}  &  =R^{-\left(  2+2\bar{\delta}\right)  }\int_{0}^{\frac{X}{2}%
}Y^{\frac{\lambda}{2}}\left(  G_{R,1}\left(  \theta,Y\right)  -G_{R,2}\left(
\theta,Y\right)  \right)  J\left(  G_{R,1};\theta,X,Y\right)  dY
\label{T5E4a}\\
I_{2,2}  &  =R^{-\left(  2+2\bar{\delta}\right)  }\int_{0}^{\frac{X}{2}%
}Y^{\frac{\lambda}{2}}G_{R,2}\left(  \theta,Y\right)  J\left(  G_{R,1}%
-G_{R,2};\theta,X,Y\right)  dY \label{T5E4b}%
\end{align}

Notice that:%
\begin{equation}
\left\vert G_{R,1}\left(  \theta,Y\right)  -G_{R,2}\left(  \theta,Y\right)
\right\vert \leq\left\Vert h_{1}-h_{2}\right\Vert _{\mathcal{Z}_{\bar{p}%
}^{\sigma;\frac{1}{2}}\left(  T\right)  }\min\left(  Y^{-\left(
\frac{3+\lambda}{2}+\bar{\delta}\right)  },R^{\frac{\lambda}{2}+\bar{\delta}%
}Y^{-\frac{3}{2}}\right)  \label{T5E3}%
\end{equation}

We now argue exactly as in the proof of Lemma \ref{LQ2} in order to estimate
$I_{2,1},\ I_{2,2}.$ Notice that estimating these terms it is crucial to use
the boundedness of the seminorm $\left[  \cdot\right]  _{\frac{3+\lambda}{2}%
}^{\sigma;\frac{1}{2}}$ in (\ref{S1E11}) for the sources. On the other hand,
the argument in the Proof of Lemma \ref{LQ2} shows that the pointwise estimate
(\ref{T5E3}) is needed. A similar argument and estimate allows to estimate the
terms $I_{2,1},\ I_{2,2}$ in (\ref{T5E4a}), (\ref{T5E4b}). Therefore, after
some computations:%
\begin{equation}
N_{2,0}\left(  I_{2},\tau_{0},R\right)  +N_{2,\sigma}\left(  I_{2},\tau
_{0},R\right)  \leq\frac{C}{R^{2+\bar{\delta}}}\left(  \sum_{k=1}%
^{2}\left\Vert h_{k}\right\Vert _{\mathcal{Z}_{\bar{p}}^{\sigma;\frac{1}{2}%
}\left(  T\right)  }\right)  \left\Vert h_{1}-h_{2}\right\Vert _{\mathcal{Z}%
_{\bar{p}}^{\sigma;\frac{1}{2}}\left(  T\right)  } \label{M2E5}%
\end{equation}

It only remains to estimate the region where $R\leq1.$ The estimate of
$R^{\frac{3}{2}}M_{2;\sigma}(I_{2};\,R),$ $R^{\frac{3}{2}}M_{2;0}(I_{2};\,R)$
can be made exactly in the same way as the estimate of similar terms for
$I_{1}.$ Notice that the two terms in $I_{2}$ yield integrals that converge
separately since\thinspace$h\left(  \tau,y\right)  $ can be estimated as
$\frac{1}{y^{\frac{3}{2}}}$ for $y\leq1$ and then, the term $y^{\frac{\lambda
}{2}-\frac{3}{2}}$ is integrable near the origin. Then:%
\begin{equation}
R^{\frac{3}{2}}M_{2;0}(I_{2};\,R)+R^{\frac{3}{2}}M_{2;\sigma}(I_{2};\,R)\leq
C\left(  \sum_{k=1}^{2}\left\Vert h_{k}\right\Vert _{\mathcal{Z}_{\bar{p}%
}^{\sigma;\frac{1}{2}}\left(  T\right)  }\right)  \left\Vert h_{1}%
-h_{2}\right\Vert _{\mathcal{Z}_{\bar{p}}^{\sigma;\frac{1}{2}}\left(
T\right)  } \label{M2E6}%
\end{equation}

Combining (\ref{M2E5}), (\ref{M2E6}) we obtain:
\begin{equation}
||I_{2}||_{Y_{\frac{3}{2},\left(  2+\bar{\delta}\right)  }^{\sigma}\left(
T\right)  }\leq C\left(  \sum_{k=1}^{2}\left\Vert h_{k}\right\Vert
_{\mathcal{Z}_{\bar{p}}^{\sigma;\frac{1}{2}}\left(  T\right)  }\right)
\left\Vert h_{1}-h_{2}\right\Vert _{\mathcal{Z}_{\bar{p}}^{\sigma;\frac{1}{2}%
}\left(  T\right)  } \label{M2E7}%
\end{equation}

The proof of the Lemma is then concluded using (\ref{T5E2}) and (\ref{M2E7}).
\end{proof}

\section{DERIVATION OF THE ASYMPTOTICS $x^{-\frac{3+\lambda}{2}}$ AS
$x\rightarrow\infty$.}

The main result in this Section is the following.

\begin{proposition}
\label{PropAsympt}Suppose that $\varphi\in\mathcal{Z}_{\frac{3+\lambda}{2}%
}^{\sigma;\frac{1}{2}}\left(  T\right)  $ solves:%
\begin{equation}
\varphi_{\tau}={\mathcal{L}}_{f_{0}}\left[  \varphi\right]  +F\left(
\tau,x\right)  \ ,\ \ \ x>0\ \ ,\ \ 0\leq t\leq T\ \ ,\ \ \varphi\left(
0,x\right)  =0 \label{G2E3}%
\end{equation}
where $F\in Y_{\frac{3}{2},2+\bar{\delta}}^{\sigma}\left(  T\right)  $ and
$\bar{\delta}<r.$ Then, the following asymptotics holds:%
\begin{equation}
\varphi\left(  \tau,x\right)  =\mathcal{W}\left(  \tau\right)  x^{-\frac
{3+\lambda}{2}}+\varphi_{R}\left(  \tau,x\right)  \ \ \ \ \ \ \ \text{as\ \ }%
x\rightarrow\infty\label{G2E2}%
\end{equation}
where:%
\begin{equation}
\mathcal{W}\left(  \tau\right)  =\int_{0}^{\tau}ds\int_{0}^{\infty}%
\frac{dx_{0}}{x_{0}}\Theta\left(  \left(  \tau-s\right)  x_{0}^{\frac
{\lambda-1}{2}}\right)  x_{0}^{\frac{3+\lambda}{2}}\left[  F\left(
s,x_{0}\right)  +\left(  {\mathcal{L}}_{f_{0}}-L\right)  \left[
\varphi\right]  \left(  s,x_{0}\right)  \right]  \label{G2E5}%
\end{equation}
with $\Theta\left(  \cdot\right)  $ as in (\ref{G1E3}).and:
\begin{equation}
\varphi_{R}\in\mathcal{Z}_{\bar{p}}^{\sigma;\frac{1}{2}}\left(  T\right)
\ \label{G2E6}%
\end{equation}

\end{proposition}

The proof of Proposition \ref{PropAsympt} is based in dealing with the
operator ${\mathcal{L}}_{f_{0}}$ as a perturbation of the operator $L.$ To
this end, we rewrite (\ref{G2E3}) as:%
\[
\varphi_{\tau}=L\left[  \varphi\right]  +\left(  {\mathcal{L}}_{f_{0}%
}-L\right)  \left[  \tilde{h}\right]  +F\left(  \tau,x\right)
\]

Using variations of constants and Theorem \ref{ThFS} we have the
representation formula:%
\begin{equation}
\varphi\left(  \tau,x\right)  =\int_{0}^{\tau}ds\int_{0}^{\infty}g\left(
\left(  \tau-s\right)  x_{0}^{\frac{\lambda-1}{2}},\frac{x}{x_{0}},1\right)
\left[  \left(  {\mathcal{L}}_{f_{0}}-L\right)  \left[  \varphi\right]
\left(  s,x\right)  +F\left(  s,x\right)  \right]  \frac{dx_{0}}{x_{0}}
\label{S5E5}%
\end{equation}

In order to prove Proposition \ref{PropAsympt} we will derive some auxiliary
Lemmas. We begin estimating the term $\left(  {\mathcal{L}}_{f_{0}}-L\right)
\left[  \varphi\right]  $ (cf. (\ref{S5E5})). Most of the estimates in the
next Lemma have been already obtained in \cite{EV2}, but we recall them here
for convenience.

\begin{lemma}
\label{Le3}Suppose that $f_{0}$ satisfies (\ref{Z1E1a})-(\ref{Z1E2b}) and
$\varphi\in\mathcal{E}_{T;\sigma}$. Then:%
\begin{align}
N_{\infty}\left(  \left(  {\mathcal{L}}_{f_{0}}-L\right)  \left[
\varphi\right]  ;\tau_{0},R\right)   &  \leq\frac{C\left\vert \left\vert
\left\vert \varphi\right\vert \right\vert \right\vert _{\sigma}}{R^{2+r}%
}\ \ ,\ \ \tau_{0}\in\left(  0,T\right)  \ \ ,\ \ R\geq1\label{M3E1}\\
M_{\infty}\left(  \left(  {\mathcal{L}}_{f_{0}}-L\right)  \left[
\varphi\right]  ;R\right)   &  \leq\frac{C\left\vert \left\vert \left\vert
\varphi\right\vert \right\vert \right\vert _{\sigma}}{R^{\frac{3}{2}}%
}\ ,\ \ R\leq1 \label{M3E2}%
\end{align}%
\begin{align}
N_{2,\sigma}\left(  \left(  {\mathcal{L}}_{f_{0}}-L\right)  \left[
\varphi\right]  ,\tau_{0},R\right)   &  \leq\frac{C}{R^{2+r}}\left\Vert
\varphi\right\Vert _{\mathcal{Z}_{\left(  \frac{3+\lambda}{2}+r\right)
}^{\sigma;\frac{1}{2}}\left(  T\right)  }\ \ ,\ \ \tau_{0}\in\left(
0,T\right)  \ \ ,\ \ R\geq1\label{M3E1a}\\
M_{2,\sigma}\left(  \left(  {\mathcal{L}}_{f_{0}}-L\right)  \left[
\varphi\right]  ,R\right)   &  \leq\frac{C}{R^{\frac{3}{2}}}\left\Vert
\varphi\right\Vert _{\mathcal{Z}_{\left(  \frac{3+\lambda}{2}+r\right)
}^{\sigma;\frac{1}{2}}\left(  T\right)  }\ \ ,\ \ R\leq1\ \label{M3E2a}%
\end{align}
\bigskip
\end{lemma}

\begin{proof}
[Proof of Lemma \ref{Le3}]We write
\[
\left(  {\mathcal{L}}_{f_{0}}-L\right)  \left[  \varphi\right]  \left(
s,x_{0}\right)  =A_{1}+A_{2}%
\]
where:%
\[
A_{1}=\int_{0}^{\frac{x}{2}}\left(  H\left(  x-y\right)  -H\left(  x\right)
\right)  y^{\frac{\lambda}{2}}\varphi\left(  \tau,y\right)  dy-H\left(
x\right)  \int_{\frac{x}{2}}^{\infty}y^{\frac{\lambda}{2}}\varphi\left(
\tau,y\right)  dy-x^{\frac{\lambda}{2}}\varphi\left(  \tau,x\right)
\int_{\frac{x}{2}}^{\infty}H\left(  y\right)  dy
\]%
\[
A_{2}=\int_{0}^{\frac{x}{2}}\left(  \left(  x-y\right)  ^{\frac{\lambda}{2}%
}\varphi\left(  \tau,x-y\right)  -x^{\frac{\lambda}{2}}\varphi\left(
\tau,x\right)  \right)  H\left(  y\right)  dy
\]
and%
\[
H\left(  y\right)  =y^{\frac{\lambda}{2}}f_{0}\left(  y\right)  -y^{-\frac
{3}{2}}%
\]

It has been proved in \cite{EV2} (cf. Section 5, Lemmas 5.1, 5.2, 5.4):
\begin{equation}
\left\vert A_{1}\right\vert \leq\frac{C\left\vert \left\vert \left\vert
\varphi\right\vert \right\vert \right\vert _{\frac{3}{2},\frac{3+\lambda}{2}}%
}{x^{\frac{3}{2}}}\ \ ,\ \ x\leq1\ \ \ ,\ \ \left\vert A_{1}\right\vert
\leq\frac{C\left\vert \left\vert \left\vert \varphi\right\vert \right\vert
\right\vert _{\frac{3}{2},\frac{3+\lambda}{2}}}{x^{2+r}}\ \ ,\ \ x\geq1
\label{G1E4}%
\end{equation}%
\begin{equation}
N_{\infty}\left(  A_{2};\tau_{0},R\right)  \leq\frac{C\left\vert \left\vert
\left\vert \varphi\right\vert \right\vert \right\vert _{\sigma}}{R^{2+r}%
}\ ,\ \ \ ,\ \ R\geq1\ \ ;\ \ \ M_{\infty}\left(  A_{2};R\right)  \leq
\frac{C\left\vert \left\vert \left\vert \varphi\right\vert \right\vert
\right\vert _{\sigma}}{R^{\frac{3}{2}}}\ ,\ \ R\leq1\ \ \ ,\ \ \ \tau_{0}%
\in\left(  0,T\right)  \label{G1E6}%
\end{equation}%
\begin{align}
N_{2,\sigma}\left(  A_{1},\tau_{0},R\right)   &  \leq\frac{C}{R^{2+r}%
}\left\Vert \varphi\right\Vert _{Y_{\frac{3}{2},\frac{3+\lambda}{2}}^{\sigma}%
}\ \ ,\ \ \tau_{0}\in\left(  0,T\right)  \ \ ,\ \ R\geq1\label{G1E8}\\
M_{2,\sigma}\left(  A_{1},R\right)   &  \leq\frac{C}{R^{\frac{3}{2}}%
}\left\Vert \varphi\right\Vert _{Y_{\frac{3}{2},\frac{3+\lambda}{2}}^{\sigma}%
}\ \ ,\ \ R\leq1\ \label{G1E9}%
\end{align}
where $r$ might be chosen as in (\ref{U3E4a}) and the norms $\left\vert
\left\vert \left\vert \cdot\right\vert \right\vert \right\vert _{\frac{3}%
{2},\frac{3+\lambda}{2}},\ \left\vert \left\vert \left\vert \cdot\right\vert
\right\vert \right\vert _{\sigma}$ are as in (\ref{M2E8}), (\ref{M2E9}).

On the other hand, arguing as in the Proof of Lemma \ref{LQ2} it is possible
to prove the following estimates:%
\begin{align}
N_{2,\sigma}\left(  A_{2},\tau_{0},R\right)   &  \leq\frac{C}{R^{2+r}%
}\left\Vert \varphi\right\Vert _{\mathcal{Z}_{\left(  \frac{3+\lambda}%
{2}+r\right)  }^{\sigma;\frac{1}{2}}\left(  T\right)  }\ \ ,\ \ \tau_{0}%
\in\left(  0,T\right)  \ \ ,\ \ R\geq1\label{G1E10}\\
M_{2,\sigma}\left(  A_{2},R\right)   &  \leq\frac{C}{R^{\frac{3}{2}}%
}\left\Vert \varphi\right\Vert _{\mathcal{Z}_{\left(  \frac{3+\lambda}%
{2}+r\right)  }^{\sigma;\frac{1}{2}}\left(  T\right)  }\ \ ,\ \ R\leq
1\ \label{G1E11}%
\end{align}

The only difference in the argument is that instead of (\ref{S5E3}) the
estimate that must be used is:%
\[
\left\vert H\left(  y\right)  \right\vert \leq C\min\left(  y^{-\frac{3}{2}%
},y^{-\left(  \frac{3}{2}+r\right)  }\right)
\]
that implies that the function $H_{R}\left(  Y\right)  =R^{\left(
\frac{3+\lambda}{2}+r\right)  }H\left(  RY\right)  $ satisfies:
\[
\left\vert H_{R}\left(  Y\right)  \right\vert \leq C\min\left(  Y^{-\left(
\frac{3+\lambda}{2}+r\right)  },R^{\frac{\lambda}{2}+r}Y^{-\frac{3}{2}%
}\right)
\]

Notice that for $R<1$ we obtain estimates with the dependence $\frac
{1}{R^{\frac{3}{2}}}$ on the right hand side (cf. (\ref{G1E6}), (\ref{G1E9}),
(\ref{G1E11})) due to the fact that the term $H\left(  x\right)  \int
_{\frac{x}{2}}^{\infty}y^{\frac{\lambda}{2}}\varphi\left(  \tau,y\right)  dy$
in the definition of $A_{1}$ yields such a power law dependence for small $x.$
Combining (\ref{G1E8})-(\ref{G1E11}) the Lemma follows. 
\end{proof}

\begin{remark}
The estimate for the term $A_{2}$ cannot be improved to the decay in the norm
$H_{x}^{\sigma}$ except if we obtain instead the decay $R^{-2}.$ Such a decay
has been obtained in \cite{EV2}. The main novelty in the estimate for $A_{2}$
obtained in Lemma \ref{Le3} is the decay like $R^{-\left(  2+r\right)  }$ in
(\ref{G1E10}) for large $R,$ that can be obtained using the estimate for the
seminorm $\left[  \varphi\right]  _{\frac{3+\lambda}{2}}^{\sigma;\frac{1}{2}}$
contained in the spaces $Z_{\left(  \frac{3+\lambda}{2}+\bar{\delta}\right)
}^{\sigma;\frac{1}{2}}.$
\end{remark}

In the proof of the following results the notation will become simpler using
the following definitions:%
\[
\fint_{\tau_{0}-\rho}^{\tau_{0}}f\left(  s\right)  ds=\frac{1}{\rho}%
\int_{\left(  \tau_{0}-\rho\right)  _{+}}^{\tau_{0}}f\left(  s\right)  ds
\]

The next Lemma shows how to compute the asymptotics as $x\rightarrow\infty$ of
the solutions of:
\[
J_{\tau}=L\left[  J\right]  +F\left(  \tau,x\right)  \ \ ,\ \ J\left(
0,x\right)  =0
\]

To this end we will use the following auxiliary functional:%
\begin{equation}
\sup_{R\geq1,\ \tau_{0}\in\left(  0,T\right)  }\left[  N_{\infty}\left(
F;\tau_{0},R\right)  R^{2+\bar{\delta}}\right]  +\sup_{R<1}\left[  M_{\infty
}\left(  F;R\right)  R^{\frac{3}{2}}\right]  \equiv\mathcal{H}\left(
F\right)  \label{M3E3}%
\end{equation}

\begin{lemma}
\label{Lemma5} Let $0<\bar{\delta}<\min\left\{  \varepsilon,r\right\}  ,$ with
$\varepsilon$ as in (\ref{G1E2}) and $r$ as in (\ref{U3E4a}). Suppose that
$F\in X_{\frac{3}{2},2+\bar{\delta}}\left(  T\right)  .$ Let $J=J\left(
\tau,x\right)  $ be:%
\begin{equation}
J\left(  \tau,x\right)  =\int_{0}^{\tau}ds\int_{0}^{\infty}g\left(  \left(
\tau-s\right)  x_{0}^{\frac{\lambda-1}{2}},\frac{x}{x_{0}},1\right)  F\left(
s,x_{0}\right)  \frac{dx_{0}}{x_{0}} \label{Z2E0}%
\end{equation}

Then:%
\begin{equation}
J\left(  \tau,x\right)  -I\left(  \tau;F\right)  x^{-\frac{3+\lambda}{2}}%
\xi\left(  x\right)  =J_{R}\left(  \tau,x\right)  \label{J2E0a}%
\end{equation}
where $\xi\left(  \cdot\right)  $ is as in (\ref{Z1E2a}) and:%
\begin{equation}
I\left(  \tau;F\right)  =\int_{0}^{\infty}\frac{dx_{0}}{x_{0}}\int_{0}^{\tau
}dsF\left(  s,x_{0}\right)  \Theta\left(  \left(  \tau-s\right)  x_{0}%
^{\frac{\lambda-1}{2}}\right)  x_{0}^{\frac{3+\lambda}{2}} \label{Z2E1}%
\end{equation}
with $\Theta\left(  \cdot\right)  $ as in (\ref{G1E3}) and
\begin{equation}
\left\vert \left\vert \left\vert J_{R}\right\vert \right\vert \right\vert
_{\frac{3}{2},\frac{3+\lambda}{2}+\bar{\delta}}\leq C\left\Vert F\right\Vert
_{X_{\frac{3}{2},2+\bar{\delta}}\left(  T\right)  } \label{Z2E1c}%
\end{equation}

Moreover:%
\begin{equation}
\left\vert I\left(  \tau;F\right)  \right\vert \leq C\tau^{\frac{2\bar{\delta
}}{\lambda-1}}\left\Vert F\right\Vert _{X_{\frac{3}{2},2+\bar{\delta}}\left(
T\right)  } \label{Z2E1d}%
\end{equation}

\end{lemma}

\begin{proof}
We split the integral in (\ref{Z2E0}) as:%
\[
J=J_{1}+J_{2}+J_{3}+J_{4}%
\]%
\begin{align}
J_{1}  &  \equiv\int_{\frac{2x}{3}}^{\infty}dx_{0}\int_{\left(  \tau
-x_{0}^{-\frac{\lambda-1}{2}}\right)  _{+}}^{\tau}ds\left[  ...\right]
\ \ ,\ \ J_{2}\equiv\int_{0}^{\frac{2x}{3}}dx_{0}\int_{\left(  \tau
-x_{0}^{\frac{\lambda-1}{2}}\right)  _{+}}^{\tau}ds\left[  ...\right]
\label{Z2E1a}\\
\ J_{3}  &  \equiv\int_{\frac{2x}{3}}^{\infty}dx_{0}\int_{0}^{\left(
\tau-x_{0}^{-\frac{\lambda-1}{2}}\right)  _{+}}ds\left[  ...\right]
\ \ ,\ \ \ J_{4}\equiv\int_{0}^{\frac{2x}{3}}dx_{0}\int_{0}^{\left(
\tau-x_{0}^{-\frac{\lambda-1}{2}}\right)  _{+}}ds\left[  ...\right]
\label{Z2E1b}%
\end{align}

In the term $J_{1}$ we use the fact that (\ref{T1E3}) implies:%
\begin{align*}
\left\vert g(\tau,x,1)\right\vert  &  \leq C\tau x^{-\frac{3}{2}%
}\ \ ,\ \ 0<x\leq\frac{1}{2}\\
\left\vert g(\tau,x,1)\right\vert  &  \leq C\tau^{-2}\varphi\left(  \frac
{x-1}{\tau^{2}}\right)  \ \ \ ,\ \ \left\vert x-1\right\vert \leq\frac{1}%
{2}\ \ \text{with \ }\varphi\left(  \xi\right)  =\frac{1}{1+\xi^{\frac{3}%
{2}-\varepsilon_{1}}}%
\end{align*}
for some $\varepsilon_{1}>0.$ Then:%
\[
\left\vert J_{1}\right\vert \leq J_{1,1}+J_{1,2}%
\]
where:%
\begin{align}
J_{1,1}  &  \equiv C\int_{\frac{x}{2}}^{\frac{3x}{2}}dx_{0}\int_{\left(
\tau-x_{0}^{-\frac{\lambda-1}{2}}\right)  _{+}}^{\tau}\left(  \tau-s\right)
^{-2}dsx_{0}^{-\lambda}\varphi\left(  \frac{x-x_{0}}{\left(  \tau-s\right)
^{2}x_{0}^{\lambda}}\right)  \left\vert F\left(  s,x_{0}\right)  \right\vert
\label{S5E7a}\\
J_{1,2}  &  \equiv\frac{C}{x^{3/2}}\int_{2x}^{\infty}dx_{0}\int_{\left(
\tau-x_{0}^{-\frac{\lambda-1}{2}}\right)  _{+}}^{\tau}\left(  \tau-s\right)
dsx_{0}^{\frac{\lambda}{2}}\left\vert F\left(  s,x_{0}\right)  \right\vert
\label{S5E7}%
\end{align}

In order to estimate $J_{3}$ we use that (\ref{T1E0})-(\ref{T1E2}) implies:%
\[
\left\vert g(\tau,x,1)\right\vert \leq C\tau^{\frac{2}{\lambda-1}}\min\left\{
\left(  \tau^{\frac{2}{\lambda-1}}x\right)  ^{-\frac{3}{2}},\left(
\tau^{\frac{2}{\lambda-1}}x\right)  ^{-\frac{3+\lambda}{2}}\right\}
\ \ ,\ \ \tau\geq1
\]

Then:%
\begin{equation}
\left\vert J_{3}\right\vert \leq C\int_{\frac{2x}{3}}^{\infty}dx_{0}\int
_{0}^{\left(  \tau-x_{0}^{-\frac{\lambda-1}{2}}\right)  _{+}}ds\left\vert
F\left(  s,x_{0}\right)  \right\vert \min\left\{  \left(  \tau-s\right)
^{-\frac{1}{\lambda-1}}x^{-\frac{3}{2}},\left(  \tau-s\right)  ^{-\frac
{1+\lambda}{\lambda-1}}x^{-\frac{3+\lambda}{2}}\right\}  \label{S5E8}%
\end{equation}

To obtain the leading asymptotics of $J_{2}$ as $x\rightarrow\infty$ we will
use the detailed asymptotics of $g\left(  \tau,x,1\right)  .$ Using
(\ref{G1E1}), (\ref{G1E3}) we have:%
\begin{equation}
\left\vert J_{2}-x^{-\frac{3+\lambda}{2}}\int_{0}^{\frac{2x}{3}}dx_{0}%
\int_{\left(  \tau-x_{0}^{-\frac{\lambda-1}{2}}\right)  _{+}}^{\tau}dsF\left(
s,x_{0}\right)  \Theta\left(  \left(  \tau-s\right)  x_{0}^{\frac{\lambda
-1}{2}}\right)  x_{0}^{\frac{1+\lambda}{2}}\right\vert =J_{2,R} \label{S5E8a}%
\end{equation}
where, due to (\ref{Z3E1b}) in Proposition \ref{PropositionImproved}:%
\begin{equation}
\left\vert J_{2,R}\right\vert \leq Cx^{-\frac{3+\lambda}{2}-r}\int_{0}%
^{\frac{2x}{3}}dx_{0}\int_{\left(  \tau-x_{0}^{-\frac{\lambda-1}{2}}\right)
_{+}}^{\tau}\left(  \tau-s\right)  ds\left\vert F\left(  s,x_{0}\right)
\right\vert x_{0}^{\lambda+r} \label{S5E9}%
\end{equation}

To estimate $J_{4}$ we use (\ref{T1E0})-(\ref{T1E2}):%
\begin{equation}
\left\vert J_{4}\right\vert \leq\frac{C}{x^{\frac{3+\lambda}{2}}}\int
_{0}^{\frac{2x}{3}}dx_{0}\int_{0}^{\left(  \tau-x_{0}^{-\frac{\lambda-1}{2}%
}\right)  _{+}}\frac{ds}{\left(  \tau-s\right)  ^{\frac{\lambda+1}{\lambda-1}%
}}\left\vert F\left(  s,x_{0}\right)  \right\vert \label{S5E10}%
\end{equation}
where we use that, since $x_{0}\leq x,$ we have $\left(  \tau-s\right)
x^{\frac{\lambda-1}{2}}\geq\frac{x^{\frac{\lambda-1}{2}}}{x_{0}^{\frac
{\lambda-1}{2}}}\geq1$.
\begin{figure}
\begin{tikzpicture}
\draw[->] (-1, 0) -- (12, 0) node[right]{$x_0$};
\draw[->] (0, -0.3) -- (0, 6) node[right]{$s$};
\draw(3, 0) node {$\times$};
\draw(3, 0) node[below]{$\frac{x}{2}$};
\draw(2.5, 0) node {$\times$};
\draw(2.5, 0) node[below]{$\frac{R_{\ell_1}}{2}$};
\draw(5, 0) node {$\times$};
\draw(5, 0) node[below]{$R_{\ell_1}$};
\draw(6, 0) node {$\times$};
\draw(6, 0) node[below]{$x$};
\draw(9, 0) node {$\times$};
\draw(9, 0) node[below]{$\frac{3x}{2}$};
\draw(10, 0) node {$\times$};
\draw(10, 0) node[below]{$2 R_{\ell_1}$};
\draw[](9,4.5) parabola (3, 1.5);
\draw[] (0, 5) -- (12, 5);
\draw[dotted] (3, 0) -- (3,1.5);
\draw[] (3, 1.5) -- (3,5);
\draw[dotted] (0, 1.5) -- (3,1.5);
\draw(0, 1.5) node [left] {$\tau-(x/2)^{-\frac{\lambda-1}{2}}$};
\draw[dotted] (9, 0) -- (9, 4.5);
\draw(9, 4.5)--(9, 5);
\draw(2.5, 4.5)--(2.5, 4.7);
\draw(10, 4.5)--(10, 4.7);
\draw(2.5, 4.7)--(10,4.7);
\draw(2.5, 4.5)--(10,4.5);
\draw(1.75, 0) node {$\times$};
\draw(1.75, 0) node[below]{$\frac{R_{\ell_2}}{2}$};
\draw(3.5, 0) node {$\times$};
\draw(3.5, 0) node[below]{$R_{\ell_2}$};
\draw(7, 0) node {$\times$};
\draw(7, 0) node[below]{$2R_{\ell_2}$};
\draw(7, 4)--(7, 3.7);
\draw(1.75, 4)--(1.75, 3.7);
\draw(1.75,  4)--(7, 4);
\draw(1.75,  3.7)--(7, 3.7);
\end{tikzpicture}
\caption{Two cubes of the covering $\mathcal{C}_{x,\tau}^{\left( 1\right)  }$ }
\end{figure}
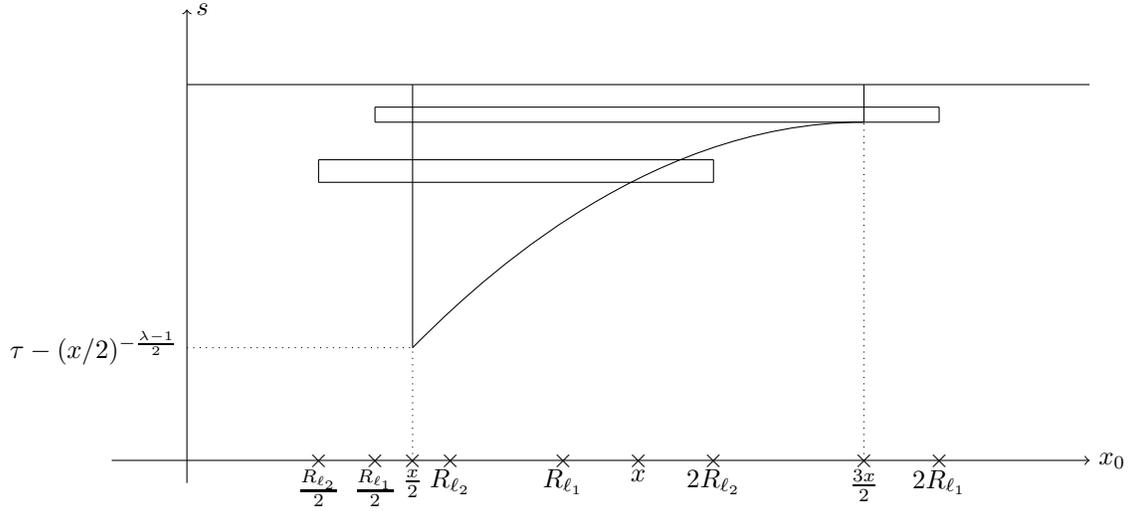
We now proceed to estimate all the terms $J_{k},\ k=1,2,3,4$ in terms of
$\mathcal{H}\left(  F\right)  $ in (\ref{M3E3}).We are interested in the
behaviour of all these quantities for large values of $x.$ We begin with the
term $J_{1,1}$ (cf. (\ref{S5E7a})). Notice that we can cover the domain
$\left\{  \left(  x_{0},s\right)  :x_{0}\in\left[  \frac{x}{2},\frac{3x}%
{2}\right]  ,\ s\in\left[  \left(  \tau-x_{0}^{-\frac{\lambda-1}{2}}\right)
_{+},\tau\right]  \right\}  $ by a set of rectangles with the form $\left[
\frac{R_{\ell}}{2},2R_{\ell}\right]  \times\left[  \tau_{0},\tau_{0}+\left(
R_{\ell}\right)  ^{-\frac{\lambda-1}{2}}\right]  ,\ $with $R_{\ell}\in\left[
\frac{x}{2},\frac{3x}{2}\right]  ,\ \tau_{0}\in\left[  \tau-2^{\frac
{\lambda-1}{2}}x^{-\frac{\lambda-1}{2}},\tau\right]  $ with $\ell
=1,...N_{x,\tau}$ and $N_{x,\tau}\leq n_{0}<\infty$ with $n_{0}$ independent
on $x,\tau.$
The fact that the number of rectangles can be estimated uniformly
on $x$ follows from the self-similarity of the problem. Let us denote such
finite covering as $\mathcal{C}_{x,,\tau}^{(1)}.$

The specific rescaling
chosen for the rectangles is due to the fact that they are the ones appearing
in the definition of the functionals $N_{\infty}\left(  F;\tau_{0},R\right)
$. We then have:%
\begin{align}
J_{1,1}  &  \leq C\sum_{\mathcal{C}_{x,,\tau}^{(1)}}\int_{\left(  \tau
-CR_{x}^{-\frac{\lambda-1}{2}}\right)  _{+}}^{\tau}\left(  \tau-s\right)
^{-2}ds\int_{\frac{x}{2}}^{\frac{3x}{2}}dx_{0}x_{0}^{-\lambda}\varphi\left(
\frac{x-x_{0}}{\left(  \tau-s\right)  ^{2}x_{0}^{\lambda}}\right)  \left\Vert
F\left(  s,\cdot\right)  \right\Vert _{L^{\infty}\left(  \frac{R_{x}}%
{2},2R_{x}\right)  }\nonumber\\
&  \leq CR_{x}^{-\frac{\lambda-1}{2}}\sum_{\mathcal{C}_{x,,\tau}^{(1)}}%
\fint_{\left(  \tau-CR_{x}^{-\frac{\lambda-1}{2}}\right)  _{+}}^{\tau
}\left\Vert F\left(  s,\cdot\right)  \right\Vert _{L^{\infty}\left(
\frac{R_{x}}{2},2R_{x}\right)  }ds\nonumber\\
&  \leq CR_{x}^{-\frac{\lambda-1}{2}}\sup_{0\leq\tau\leq T}N_{\infty}\left(
F;\tau,R\right)  \leq Cx^{-\left(  \frac{3+\lambda}{2}+\bar{\delta}\right)
}\mathcal{H}\left(  F\right)  \label{S6E1}%
\end{align}
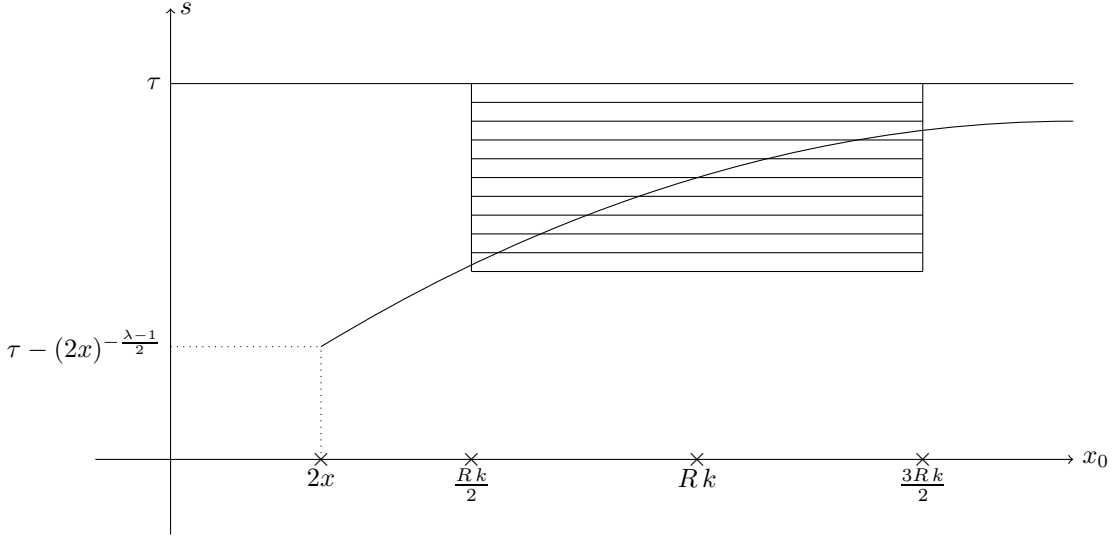
\begin{figure}
\begin{tikzpicture}
\draw[->] (-1, 0) -- (12, 0) node[right]{$x_0$};
\draw[->] (0, -1) -- (0, 6) node[right]{$s$};
\draw(0, 5) node [left] {$\tau$};
\draw(0, 1.5) node [left] {$\tau-(2x)^{-\frac{\lambda-1}{2}}$};
\draw[dotted] (0, 1.5) -- (2,1.5);
\draw(2, 0) node {$\times$};
\draw(2, 0) node[below]{$2x$};
\draw[dotted] (2,0) -- (2,1.5);
\draw(4, 0) node {$\times$};
\draw(4, 0) node[below]{$\frac{R\, k}{2}$};
\draw[] (4,2.5) -- (4, 5);
\draw(7, 0) node {$\times$};
\draw(7, 0) node[below]{$R\, k$};
\draw(10, 0) node {$\times$};
\draw(10, 0) node[below]{$\frac{3 R\, k}{2}$};
\draw[] (10,2.5) -- (10, 5);
\draw[] (0, 5) -- (12, 5);
\draw (4, 2.5)--(10, 2.5);
\draw (4, 2.75)--(10, 2.75);
\draw (4, 3)--(10, 3);
\draw (4, 3.25)--(10, 3.25);
\draw (4, 3.5)--(10, 3.5);
\draw (4, 3.75)--(10, 3.75);
\draw (4, 4)--(10, 4);
\draw (4, 4.25)--(10, 4.25);
\draw (4, 4.5)--(10, 4.5);
\draw (4, 4.75)--(10, 4.75);
\draw[](12,4.5) parabola (2, 1.5);
\end{tikzpicture}
\caption{The sub-family $\mathcal R_{k}$ of cubes of the covering $\mathcal{C}_{x,\tau}^{\left( 2\right)  }$}
\end{figure}

We now estimate the term $J_{1,2}$ in (\ref{S5E7})$.$ Let us denote as
$\mathcal{C}_{x,\tau}^{\left(  2\right)  }$ a covering of the set\newline%
$\left\{  \left(  x_{0},s\right)  :x_{0}\geq2x,\ \left(  \tau-x_{0}%
^{-\frac{\lambda-1}{2}}\right)  _{+}\leq s\leq\tau\right\}  $ by means of
boxes with the form $\left[  \frac{R}{2},2R\right]  \times\left[  \tau
_{0},\tau_{0}+\left(  R\right)  ^{-\frac{\lambda-1}{2}}\right]  ,$ with $R\geq
x,\ \tau_{0}\in\left[  0,\tau\right]  $ in which each point is covered at most
by a finite number of boxes (independent on $x,\tau$), and where the sequence
of sizes $R$ increases exponentially. 

Then:

\begin{align*}
&  \frac{1}{x^{3/2}}\int_{2x}^{\infty}dx_{0}\int_{\tau-x_{0}^{-\frac
{\lambda-1}{2}}}^{\tau}ds\left(  \tau-s\right)  x_{0}^{\frac{\lambda}{2}%
}\left\vert F\left(  s,x_{0}\right)  \right\vert \\
&  \leq\frac{C}{x^{3/2}}\sum_{\mathcal{C}_{x,\tau}^{\left(  2\right)  }%
}R^{\frac{\lambda}{2}+1}\int_{\tau_{0}-R^{-\frac{\lambda-1}{2}}}^{\tau_{0}%
}ds\left(  \tau-s\right)  \left\Vert F\left(  s,\cdot\right)  \right\Vert
_{L^{\infty}\left(  \frac{R}{2},2R\right)  }\\
&  \leq\frac{C}{x^{3/2}}\sum_{\mathcal{C}_{x,\tau}^{\left(  2\right)  }%
}R^{2-\frac{\lambda}{2}}\left(  \fint_{\tau_{0}-R^{-\frac{\lambda-1}{2}}%
}^{\tau_{0}}ds\left\Vert F\left(  s,\cdot\right)  \right\Vert _{L^{\infty
}\left(  \frac{R}{2},2R\right)  }^{2}\right)  ^{\frac{1}{2}}\\
&  \leq\frac{C}{x^{3/2}}\sum_{\mathcal{C}_{x,\tau}^{\left(  2\right)  }%
}R^{2-\frac{\lambda}{2}}N_{\infty}\left(  F;\tau,R\right)  \leq\frac
{C\mathcal{H}\left(  F\right)  }{x^{3/2}}\sum_{\mathcal{C}_{x,\tau}^{\left(
2\right)  }}R^{2-\frac{\lambda}{2}}R^{-\left(  2+\bar{\delta}\right)  }%
\leq\frac{C\mathcal{H}\left(  F\right)  }{x^{\frac{3+\lambda}{2}+\bar{\delta}%
}}%
\end{align*}
where we use the fact that the series $\sum_{\mathcal{C}_{x,\tau}^{\left(
2\right)  }}\left[  \cdot\cdot\cdot\right]  $ is a geometric series, due to
our choice of the sizes of the boxes. Therefore:
\begin{equation}
\left\vert J_{1}\right\vert \leq Cx^{-\left(  \frac{3+\lambda}{2}+\bar{\delta
}\right)  }\mathcal{H}\left(  F\right)  \label{S6E2}%
\end{equation}

We now estimate $J_{3}$ (cf. (\ref{S5E8})). To this end we introduce a new
covering $\mathcal{C}_{x,\tau}^{\left(  3\right)  }$ of the set $\mathcal{D}%
_{x,\tau}=\left\{  \left(  x_{0},s\right)  :x_{0}\geq2x,\ 0\leq s\leq\left(
\tau-x_{0}^{-\frac{\lambda-1}{2}}\right)  _{+}\right\}  $ by means of boxes of
the form $\left[  \frac{R}{2}, 2R\right]  \times\left[  \tau_{0},\tau
_{0}+\left(  R\right)  ^{-\frac{\lambda-1}{2}}\right]  $ with $R\geq
x,\ \tau_{0}\in\left[  0,\tau\right]  $. 

We will assume that the rectangles in
the set $\mathcal{C}_{x,\tau}^{\left(  3\right)  }$ have the following
properties. (i) Each point in the set $\mathcal{D}_{x,\tau}$ is contained in
at most a finite number of boxes independent on $x,\ \tau$ for $\tau\in\left[
0,T\right]  $. (ii) There exist a sequence of sizes $\left\{  R_{k}\right\}_{k\in N}  $
that increases exponentially on $k,$ such that the number of cubes with $R$
comparable to a given $R_{k}$ (i. e. $R\in\left(  \frac{R_{k}}{2}%
,2R_{k}\right)  $) is bounded by $CR_{k}^{\frac{\lambda-1}{2}},$ with $C$
independent on $x,\ \tau$ for $\tau\in\left[  0,T\right]  .$ Notice that the
construction of the covering implies that, for $\left(  x_{0},s\right)
\in\left[  \frac{R}{2},2R\right]  \times\left[  \tau_{0},\tau_{0}+\left(
R\right)  ^{-\frac{\lambda-1}{2}}\right]  $ we have $\frac{1}{2}\left(
\tau-\tau_{0}\right)  \leq\left(  \tau-s\right)  \leq2\left(  \tau-\tau
_{0}\right)  .$

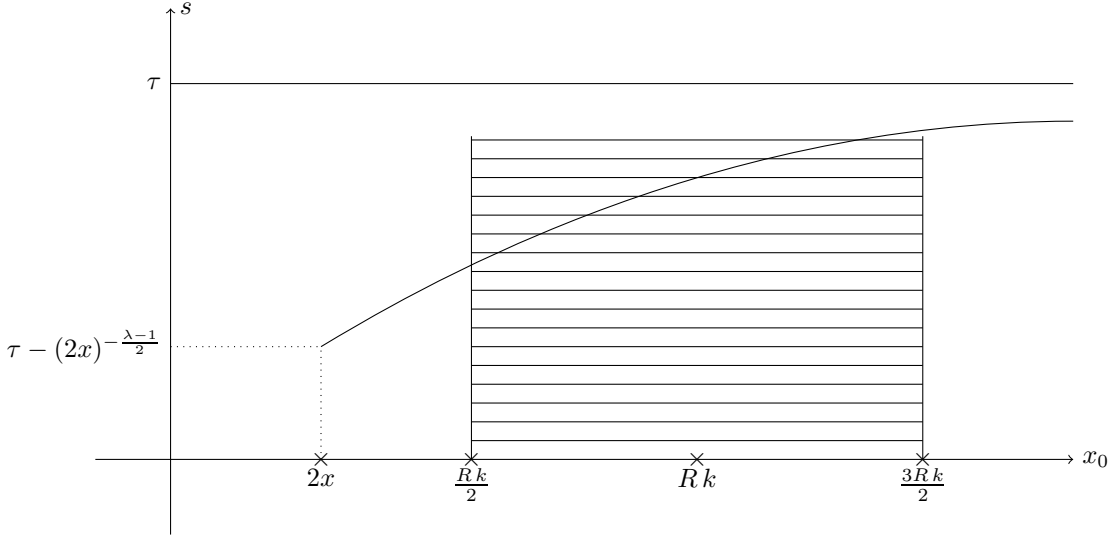
\begin{figure}
\begin{tikzpicture}
\draw[->] (-1, 0) -- (12, 0) node[right]{$x_0$};
\draw[->] (0, -1) -- (0, 6) node[right]{$s$};
\draw(0, 5) node [left] {$\tau$};
\draw(0, 1.5) node [left] {$\tau-(2x)^{-\frac{\lambda-1}{2}}$};
\draw[dotted] (0, 1.5) -- (2,1.5);
\draw(2, 0) node {$\times$};
\draw(2, 0) node[below]{$2x$};
\draw[dotted] (2,0) -- (2,1.5);
\draw(4, 0) node {$\times$};
\draw(4, 0) node[below]{$\frac{R\, k}{2}$};
\draw[] (4,0) -- (4, 4.3);
\draw(7, 0) node {$\times$};
\draw(7, 0) node[below]{$R\, k$};
\draw(10, 0) node {$\times$};
\draw(10, 0) node[below]{$\frac{3 R\, k}{2}$};
\draw[] (10,0) -- (10, 4.3);
\draw[] (0, 5) -- (12, 5);
\draw (4, 0.25)--(10, 0.25);
\draw (4, 0.5)--(10, 0.5);
\draw (4, 0.75)--(10, 0.75);
\draw (4, 1)--(10, 1);
\draw (4, 1.25)--(10, 1.25);
\draw (4, 1.5)--(10, 1.5);
\draw (4, 1.75)--(10, 1.75);
\draw (4, 2)--(10, 2);
\draw (4, 2.25)--(10, 2.25);
\draw (4, 2.5)--(10, 2.5);
\draw (4, 2.75)--(10, 2.75);
\draw (4, 3)--(10, 3);
\draw (4, 3.25)--(10, 3.25);
\draw (4, 3.5)--(10, 3.5);
\draw (4, 3.75)--(10, 3.75);
\draw (4, 4)--(10, 4);
\draw (4, 4.25)--(10, 4.25);
\draw[](12,4.5) parabola (2, 1.5);
\end{tikzpicture}
\caption{The sub-family $\mathcal R_{k}$ of cubes of the covering $\mathcal{C}_{x,\tau}^{\left(  3\right)  }$}
\end{figure}

Then, using $\min\left\{  1,\frac{1}{\left(  \tau-s\right)  ^{\frac{2}%
{\lambda-1}}x}\right\}  \leq1:$%
\begin{align*}
\left\vert J_{3}\right\vert  &  \leq\frac{C}{x^{\frac{3}{2}}}\sum
_{\mathcal{C}_{x,\tau}^{\left(  3\right)  }}\int_{\tau_{0}}^{\tau
_{0}+R^{-\frac{\lambda-1}{2}}}ds\int_{\frac{R}{2}}^{2R}dx_{0}\frac{\left\Vert
F\left(  s,\cdot\right)  \right\Vert _{L^{\infty}\left(  \frac{R}%
{2},2R\right)  }}{\left(  \tau-\tau_{0}\right)  ^{\frac{1}{\lambda-1}}}\\
&  \leq\frac{C}{x^{\frac{3}{2}}}\sum_{\mathcal{C}_{x,\tau}^{\left(  3\right)
}}\left[  \frac{R}{\left(  \tau-\tau_{0}\right)  ^{\frac{1}{\lambda-1}}%
}R^{-\frac{\lambda-1}{2}}\right]  \fint_{\tau_{0}}^{\tau_{0}+R^{-\frac
{\lambda-1}{2}}}\left\Vert F\left(  s,\cdot\right)  \right\Vert _{L^{\infty
}\left(  \frac{R}{2},2R\right)  }ds\\
&  \leq\frac{C}{x^{\frac{3}{2}}}\sum_{\mathcal{C}_{x,\tau}^{\left(  3\right)
}}\left[  R\int_{\tau_{0}}^{\tau_{0}+R^{-\frac{\lambda-1}{2}}}\frac
{ds}{\left(  \tau-s\right)  ^{\frac{1}{\lambda-1}}}\right]  N_{\infty}\left(
F;\tau,R\right)
\end{align*}
where we have used the fact that the length of the time integration in each
box is of order $R^{-\frac{\lambda-1}{2}}$ as well as Cauchy-Schwartz
inequality in the last step. Therefore, using the definition of $\mathcal{H}%
\left(  F\right)  $ as well as the properties of the covering $\mathcal{C}%
_{x,\tau}^{\left(  3\right)  }$ and in particular the properties of the
sequence $\left\{  R_{k}\right\}  :$
\begin{align*}
\left\vert J_{3}\right\vert  &  \leq\frac{C\mathcal{H}\left(  F\right)
}{x^{\frac{3}{2}}}\sum_{\mathcal{C}_{x,\tau}^{\left(  3\right)  }}\left[
R^{-\left(  1+\bar{\delta}\right)  }\int_{\tau_{0}}^{\tau_{0}+R^{-\frac
{\lambda-1}{2}}}\frac{ds}{\left(  \tau-s\right)  ^{\frac{1}{\lambda-1}}%
}\right] \\
&  \leq\frac{C\mathcal{H}\left(  F\right)  }{x^{\frac{3}{2}}}\sum_{\left\{
R_{k}\right\}  }\left[  R_{k}^{-\left(  1+\bar{\delta}\right)  }\int_{0}%
^{\tau-CR_{k}^{-\frac{\lambda-1}{2}}}\frac{ds}{\left(  \tau-s\right)
^{\frac{1}{\lambda-1}}}\right]  \leq\frac{C\mathcal{H}\left(  F\right)
}{x^{\frac{3}{2}}}\sum_{\left\{  R_{k}\right\}  }R_{k}^{-\frac{\lambda}%
{2}-\bar{\delta}}\leq\frac{C\mathcal{H}\left(  F\right)  }{x^{\frac{3+\lambda
}{2}+\bar{\delta}}}%
\end{align*}

\texttt{... Add Figure?. ...}

Then:%
\begin{equation}
\left\vert J_{3}\right\vert \leq\frac{C\mathcal{H}\left(  F\right)  }%
{x^{\frac{3+\lambda}{2}+\bar{\delta}}} \label{S6E3}%
\end{equation}

We now estimate the term $J_{2}.$ Using (\ref{G1E3}), (\ref{S5E8a}) and
(\ref{S5E9})) we obtain $J_{2}\leq K_{2}$ where:%
\[
K_{2}\leq\int_{0}^{\infty}dx_{0}\int_{\left(  \tau-x_{0}^{-\frac{\lambda-1}%
{2}}\right)  _{+}}^{\tau}ds\left\vert F\left(  s,x_{0}\right)  \right\vert
x_{0}^{\frac{\lambda+1}{2}}%
\]

We introduce a covering $\mathcal{C}_{x,\tau}^{\left(  4\right)  }$ such that
each point of the set $\left\{  \left(  x_{0},s\right)  :x_{0}\geq1,\ \left(
\tau-x_{0}^{-\frac{\lambda-1}{2}}\right)  _{+}\leq s\leq\tau\right\}  $ is
contained in a bounded number of boxes having the form $\left[  \frac{R}%
{2},2R\right]  \times\left[  \tau_{0},\tau_{0}+\left(  R\right)
^{-\frac{\lambda-1}{2}}\right]  ,$ and where the values of the radii $R$
increase exponentially. 
\begin{figure}
\begin{tikzpicture}
\draw[->] (-1, 0) -- (12, 0) node[right]{$x_0$};
\draw[->] (0, -1) -- (0, 6) node[right]{$s$};
\draw(0, 5) node [left] {$\tau$};
\draw(0, 1.5) node [left] {$\tau-1$};
\draw[dotted] (0, 1.5) -- (2,1.5);
\draw(2, 0) node {$\times$};
\draw(2, 0) node[below]{$1$};
\draw[dotted] (2,0) -- (2,1.5);
\draw(4, 0) node {$\times$};
\draw(4, 0) node[below]{$\frac{R\, k}{2}$};
\draw[] (4,2.2) -- (4, 5);
\draw(7, 0) node {$\times$};
\draw(7, 0) node[below]{$R\, k$};
\draw(10, 0) node {$\times$};
\draw(10, 0) node[below]{$\frac{3 R\, k}{2}$};
\draw[] (10,2.2) -- (10, 5);
\draw[] (0, 5) -- (12, 5);
\draw (4, 2.25)--(10, 2.25);
\draw (4, 2.5)--(10, 2.5);
\draw (4, 2.75)--(10, 2.75);
\draw (4, 3)--(10, 3);
\draw (4, 3.25)--(10, 3.25);
\draw (4, 3.5)--(10, 3.5);
\draw (4, 3.75)--(10, 3.75);
\draw (4, 4)--(10, 4);
\draw (4, 4.25)--(10, 4.25);
\draw (4, 4.5)--(10, 4.5);
\draw (4, 4.75)--(10, 4.75);
\draw[](12,4.5) parabola (2, 1.5);
\end{tikzpicture}
\caption{The sub-family $\mathcal R_{k}$ of cubes of the covering $\mathcal{C}_{x,\tau}^{\left(  4\right)  }$}
\end{figure}
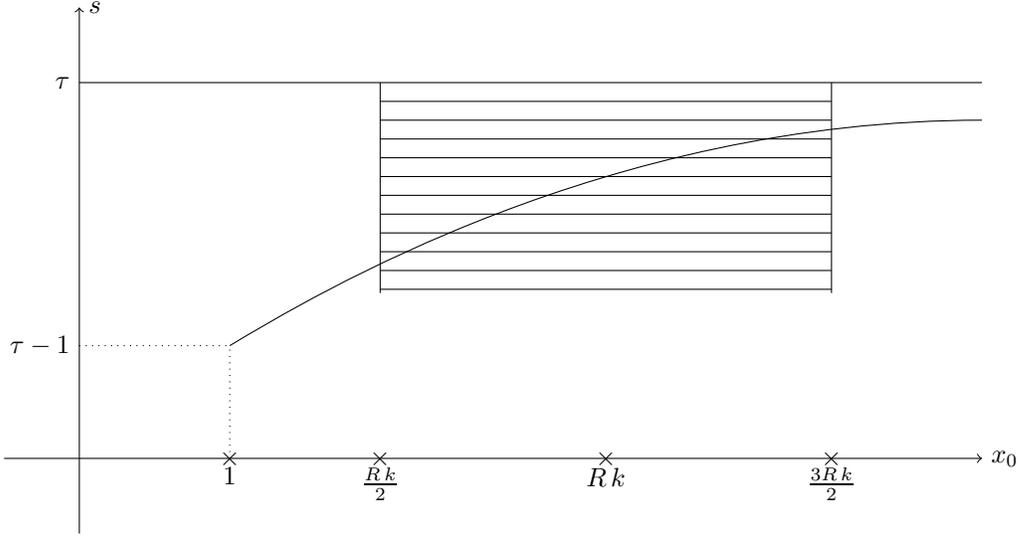

Then:%
\[
K_{2}\leq\int_{0}^{\tau}ds\int_{0}^{1}dx_{0}\left\vert F\left(  s,x_{0}%
\right)  \right\vert x_{0}^{\frac{\lambda+1}{2}}+\sum_{\mathcal{C}_{x,\tau
}^{\left(  4\right)  }}\int_{\frac{R}{2}}^{2R}dx_{0}\int_{\left(  \tau
-x_{0}^{-\frac{\lambda-1}{2}}\right)  _{+}}^{\tau}ds\left\vert F\left(
s,x_{0}\right)  \right\vert x_{0}^{\frac{\lambda+1}{2}}=K_{2,1}+K_{2,2}%
\]

Using Cauchy-Schwartz we then estimate $K_{2,2}$ as:%
\[
K_{2,2}\leq C\sum_{\mathcal{C}_{x,\tau}^{\left(  4\right)  }}R^{2}N_{\infty
}\left(  F;\tau,R\right)  \leq C\mathcal{H}\left(  F\right)  \sum
_{\mathcal{C}_{x,\tau}^{\left(  4\right)  }}R^{-\bar{\delta}}\leq
C\mathcal{H}\left(  F\right)
\]

On the other hand we estimate $K_{2,1}$ decomposing the interval $\left[
0,1\right]  $ as:%
\begin{equation}
\left[  0,1\right]  =\bigcup_{n=0}^{\infty}\left[  R_{n+1},R_{n}\right]
\ \ \ ,\ \ R_{n}=2^{-n}\ \ ,\ \ n=0,1,2,... \label{S6E4}%
\end{equation}

Then, using again Cauchy-Schwartz:%
\begin{align*}
K_{2,1}  &  \leq C\sum_{n=0}^{\infty}\int_{0}^{\tau}ds\int_{R_{n+1}}^{R_{n}%
}dx_{0}\left\Vert F\left(  s,\cdot\right)  \right\Vert _{L^{\infty}\left(
\frac{R_{n}}{2},R_{n}\right)  }R_{n}^{\frac{\lambda+1}{2}}\\
&  \leq C\mathcal{H}\left(  F\right)  \sum_{n=0}^{\infty}R_{n}^{\frac
{\lambda+3}{2}+1}R_{n}^{-\frac{3}{2}}\leq C\mathcal{H}\left(  F\right)
\end{align*}

Therefore:
\[
J_{2}\leq K_{2}\leq C\mathcal{H}\left(  F\right)
\]

Actually we can derive a more precise approximation for $J_{2}$ rewriting it
as $J_{2}=I_{2}\left(  \tau\right)  x^{-\frac{3+\lambda}{2}}+J_{2;R,1}%
+J_{2;R}$ where:%
\[
I_{2}\left(  \tau\right)  =\int_{0}^{\infty}\frac{dx_{0}}{x_{0}}\int_{\left(
\tau-x_{0}^{-\frac{\lambda-1}{2}}\right)  _{+}}^{\tau}dsF\left(
s,x_{0}\right)  \Theta\left(  \left(  \tau-s\right)  x_{0}^{\frac{\lambda
-1}{2}}\right)  x_{0}^{\frac{3+\lambda}{2}}%
\]%
\[
J_{2;R,1}=-x^{-\frac{3+\lambda}{2}}\int_{\frac{2x}{3}}^{\infty}\frac{dx_{0}%
}{x_{0}}\int_{\left(  \tau-x_{0}^{-\frac{\lambda-1}{2}}\right)  _{+}}^{\tau
}dsF\left(  s,x_{0}\right)  \Theta\left(  \left(  \tau-s\right)  x_{0}%
^{\frac{\lambda-1}{2}}\right)  x_{0}^{\frac{3+\lambda}{2}}%
\]%
\[
J_{2,R}=\int_{0}^{\tau}ds\int_{0}^{\infty}\left[  g\left(  \left(
\tau-s\right)  x_{0}^{\frac{\lambda-1}{2}},\frac{x}{x_{0}},1\right)
-\Theta\left(  \left(  \tau-s\right)  x_{0}^{\frac{\lambda-1}{2}}\right)
\left(  \frac{x_{0}}{x}\right)  ^{\frac{3+\lambda}{2}}\right]  F\left(
s,x_{0}\right)  \frac{dx_{0}}{x_{0}}%
\]

The terms $J_{2;R}$ can be estimated using Proposition
\ref{PropositionImproved} (cf. (\ref{Z3E1}), (\ref{Z3E1a})) arguing as in the
estimate of $K_{2}$ exactly as the previous estimate for $J_{2}$ since
$r>\bar{\delta}.$ Therefore, we obtain the estimate:%
\[
\left\vert J_{2,R}\right\vert \leq Cx^{-\frac{3+\lambda}{2}-r}\int_{0}%
^{\frac{2x}{3}}\frac{dx_{0}}{\left(  x_{0}\right)  ^{1+\bar{\delta}-r}}\leq
Cx^{-\frac{3+\lambda}{2}-\bar{\delta}}%
\]
where the constants $C$ could be very large for small $\bar{\delta}.$ on the
other hand:%
\[
\left\vert J_{2;R,1}\right\vert \leq Cx^{-\frac{3+\lambda}{2}}\int_{\frac
{2x}{3}}^{\infty}\frac{dx_{0}}{x_{0}}\int_{\left(  \tau-x_{0}^{-\frac
{\lambda-1}{2}}\right)  _{+}}^{\tau}ds\left(  \tau-s\right)  x_{0}%
^{\frac{\lambda-1}{2}}x_{0}^{\frac{3+\lambda}{2}}\left(  x_{0}\right)
^{-2-\bar{\delta}}\leq Cx^{-\frac{3+\lambda}{2}-\bar{\delta}}%
\]
whence:%
\begin{equation}
\left\vert J_{2}-I_{2}\left(  \tau\right)  x^{-\frac{3+\lambda}{2}}\right\vert
\leq Cx^{-\frac{3+\lambda}{2}-\bar{\delta}} \label{M3E6}
\end{equation}

We now estimate the integrand in $J_{4}$ in (\ref{S5E10}). In order to apply
Lebesgue's Theorem we need to prove that:%
\[
K_{4}=\int_{0}^{\infty}dx_{0}\int_{0}^{\left(  \tau-x_{0}^{-\frac{\lambda
-1}{2}}\right)  _{+}}\frac{ds}{\left(  \tau-s\right)  ^{\frac{\lambda
+1}{\lambda-1}}}\left\vert F\left(  s,x_{0}\right)  \right\vert <\infty
\]

We define a new covering $\mathcal{C}_{x,\tau}^{\left(  5\right)  }$\ of the
set $\left\{  \left(  x_{0},s\right)  :x_{0}\geq1,\ 0\leq s\leq\left(
\tau-x_{0}^{-\frac{\lambda-1}{2}}\right)  _{+}\right\}  $ having the same
properties as the covering $\mathcal{C}_{x,\tau}^{\left(  3\right)  }.$ 

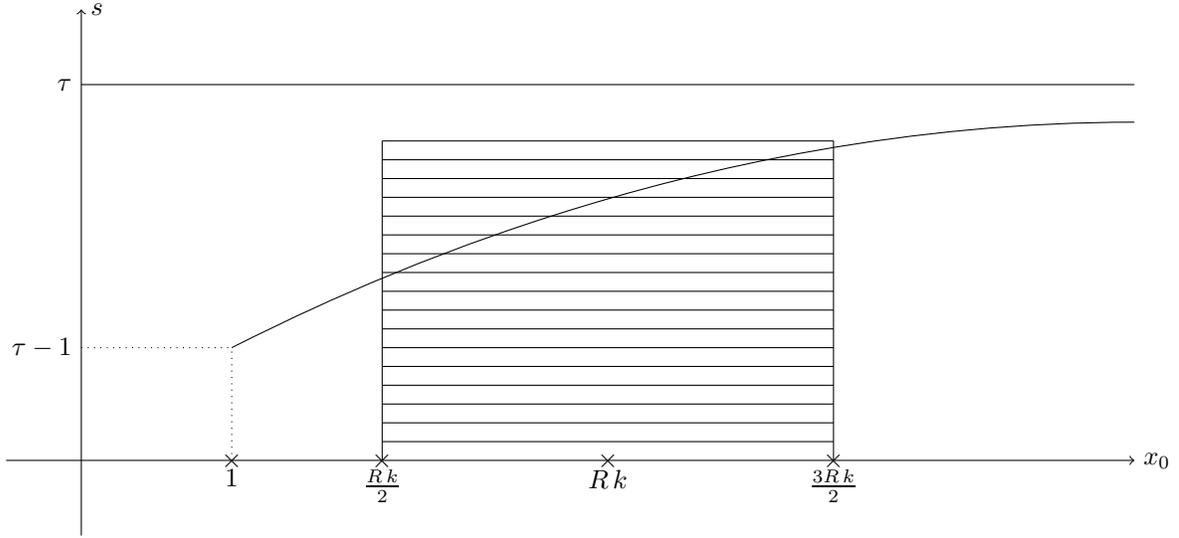
\begin{figure}
\begin{tikzpicture}
\draw[->] (-1, 0) -- (14, 0) node[right]{$x_0$};
\draw[->] (0, -1) -- (0, 6) node[right]{$s$};
\draw(0, 5) node [left] {$\tau$};
\draw(0, 1.5) node [left] {$\tau-1$};
\draw[dotted] (0, 1.5) -- (2,1.5);
\draw(2, 0) node {$\times$};
\draw(2, 0) node[below]{$1$};
\draw[dotted] (2,0) -- (2,1.5);
\draw(4, 0) node {$\times$};
\draw(4, 0) node[below]{$\frac{R\, k}{2}$};
\draw[] (4,0) -- (4, 4.25);
\draw(7, 0) node {$\times$};
\draw(7, 0) node[below]{$R\, k$};
\draw(10, 0) node {$\times$};
\draw(10, 0) node[below]{$\frac{3 R\, k}{2}$};
\draw[] (10,0) -- (10, 4.25);
\draw[] (0, 5) -- (14, 5);
\draw (4, 0.25)--(10, 0.25);
\draw (4, 0.5)--(10, 0.5);
\draw (4, 0.75)--(10, 0.75);
\draw (4, 1)--(10, 1);
\draw (4, 1.25)--(10, 1.25);
\draw (4, 1.5)--(10, 1.5);
\draw (4, 1.75)--(10, 1.75);
\draw (4, 2)--(10, 2);
\draw (4, 2.25)--(10, 2.25);
\draw (4, 2.5)--(10, 2.5);
\draw (4, 2.75)--(10, 2.75);
\draw (4, 3)--(10, 3);
\draw (4, 3.25)--(10, 3.25);
\draw (4, 3.5)--(10, 3.5);
\draw (4, 3.75)--(10, 3.75);
\draw (4, 4)--(10, 4);
\draw (4, 4.25)--(10, 4.25);
\draw[](14,4.5) parabola (2, 1.5);
\end{tikzpicture}
\caption{The sub-family $\mathcal R_{k}$ of cubes of the covering $\mathcal{C}_{x,\tau}^{\left(  5\right)  }$}
\end{figure}
Then,
using that the integrand is empty for $x_{0}\leq1$ due to the fact that
$\tau\leq T\leq1$, as well as Cauchy-Schwartz and the properties of the
covering $\mathcal{C}_{x,\tau}^{\left(  5\right)  }$ we obtain:%
\begin{align*}
K_{4}  &  \leq C\sum_{\mathcal{C}_{x,\tau}^{\left(  5\right)  }}%
\frac{R^{1-\frac{\lambda-1}{2}}}{\left(  \tau-\tau_{0}\right)  ^{\frac
{\lambda+1}{\lambda-1}}}\fint_{\tau_{0}}^{\left(  \tau_{0}-CR^{-\frac
{\lambda-1}{2}}\right)  _{+}}\left\Vert F\left(  s,\cdot\right)  \right\Vert
_{L^{\infty}\left(  \frac{R}{2},2R\right)  }ds\\
&  \leq C\mathcal{H}\left(  F\right)  \sum_{\mathcal{C}_{x,\tau}^{\left(
5\right)  }}\frac{R^{-1-\frac{\lambda-1}{2}-\bar{\delta}}}{\left(  \tau
-\tau_{0}\right)  ^{\frac{\lambda+1}{\lambda-1}}}\\
&  \leq C\mathcal{H}\left(  F\right)  \sum_{\mathcal{C}_{x,\tau}^{\left(
5\right)  }}R_{k}^{-\left(  1+\bar{\delta}\right)  }\int_{\tau_{0}}^{\left(
\tau_{0}-CR_{k}^{-\frac{\lambda-1}{2}}\right)  _{+}}\left(  \tau-s\right)
^{-\frac{\lambda+1}{\lambda-1}}ds\\
&  \leq C\mathcal{H}\left(  F\right)  \sum_{\left\{  R_{k}\right\}  }%
R_{k}^{-\bar{\delta}}\leq C\mathcal{H}\left(  F\right)
\end{align*}
whence:%
\begin{equation}
K_{4}\leq C\mathcal{H}\left(  F\right)  \label{M3E4}%
\end{equation}

In order to obtain (\ref{J2E0a}) we need to substract from $J_{4}$ the leading
contribution and to estimate the remainder. Using (\ref{G1E2}) we obtain:%
\begin{align*}
J_{4}  &  =\int_{0}^{\frac{2x}{3}}\frac{dx_{0}}{x_{0}}\int_{0}^{\left(
\tau-x_{0}^{-\frac{\lambda-1}{2}}\right)  _{+}}dsg\left(  \left(
\tau-s\right)  x_{0}^{\frac{\lambda-1}{2}},\frac{x}{x_{0}},1\right)  F\left(
s,x_{0}\right) \\
&  =\int_{0}^{\frac{2x}{3}}\frac{dx_{0}}{x_{0}}\int_{0}^{\left(  \tau
-x_{0}^{-\frac{\lambda-1}{2}}\right)  _{+}}ds\Theta\left(  \left(
\tau-s\right)  x_{0}^{\frac{\lambda-1}{2}}\right)  \left(  x\right)
^{-\frac{3+\lambda}{2}}F\left(  s,x_{0}\right)  +J_{4,R}%
\end{align*}
where:%
\[
\left\vert J_{4,R}\right\vert \leq C\int_{0}^{\frac{2x}{3}}\frac{dx_{0}}%
{x_{0}}\int_{0}^{\left(  \tau-x_{0}^{-\frac{\lambda-1}{2}}\right)  _{+}%
}ds\left(  \left(  \tau-s\right)  x_{0}^{\frac{\lambda-1}{2}}\right)
^{-\frac{2\varepsilon}{\lambda-1}-\frac{\lambda+1}{\lambda-1}}\left(  \frac
{x}{x_{0}}\right)  ^{-\frac{3+\lambda}{2}-\varepsilon}\left\vert F\left(
s,x_{0}\right)  \right\vert
\]

The argument yielding (\ref{M3E4}) implies the existence of the integral:%
\[
I_{4}\left(  \tau\right)  =\int_{0}^{\infty}\frac{dx_{0}}{x_{0}}\int
_{0}^{\left(  \tau-x_{0}^{-\frac{\lambda-1}{2}}\right)  _{+}}ds\Theta\left(
\left(  \tau-s\right)  x_{0}^{\frac{\lambda-1}{2}}\right)  \left(
\tau-s\right)  ^{-\frac{3+\lambda}{\lambda-1}}F\left(  s,x_{0}\right)
\]

Therefore:%
\begin{equation}
J_{4}=I_{4}\left(  \tau\right)  x^{-\frac{3+\lambda}{2}}+\int_{\frac{2x}{3}%
}^{\infty}\frac{dx_{0}}{x_{0}}\int_{0}^{\left(  \tau-x_{0}^{-\frac{\lambda
-1}{2}}\right)  _{+}}ds\Theta\left(  \left(  \tau-s\right)  x_{0}%
^{\frac{\lambda-1}{2}}\right)  \left(  \tau-s\right)  ^{-\frac{3+\lambda
}{\lambda-1}}F\left(  s,x_{0}\right)  +J_{4,R} \label{M3E5}%
\end{equation}

We then estimate the remainders in (\ref{M3E5}). Notice that the bounds for
$\Theta$ and $F$ yield:%
\begin{align*}
&  \left\vert \int_{\frac{2x}{3}}^{\infty}\frac{dx_{0}}{x_{0}}\int
_{0}^{\left(  \tau-x_{0}^{-\frac{\lambda-1}{2}}\right)  _{+}}ds\Theta\left(
\left(  \tau-s\right)  x_{0}^{\frac{\lambda-1}{2}}\right)  \left(
\tau-s\right)  ^{-\frac{3+\lambda}{\lambda-1}}F\left(  s,x_{0}\right)
\right\vert \\
&  \leq C\left\vert \int_{\frac{2x}{3}}^{\infty}\frac{dx_{0}}{x_{0}}\int
_{0}^{\left(  \tau-x_{0}^{-\frac{\lambda-1}{2}}\right)  _{+}}ds\left(
\tau-s\right)  ^{\frac{2}{\lambda-1}}x_{0}\left(  \tau-s\right)
^{-\frac{3+\lambda}{\lambda-1}}\left(  x_{0}\right)  ^{-\left(  2+\bar{\delta
}\right)  }\right\vert \\
&  \leq Cx^{-\bar{\delta}}%
\end{align*}

We estimate the term $J_{4,R}$ in (\ref{M3E5}) as:%
\[
\left\vert J_{4,R}\right\vert \leq Cx^{-\frac{3+\lambda}{2}-\varepsilon}%
\int_{0}^{\frac{2x}{3}}\frac{dx_{0}}{x_{0}}\int_{0}^{\left(  \tau
-x_{0}^{-\frac{\lambda-1}{2}}\right)  _{+}}ds\left(  \left(  \tau-s\right)
x_{0}^{\frac{\lambda-1}{2}}\right)  ^{-\frac{2\varepsilon+\lambda+1}%
{\lambda-1}}\left(  x_{0}\right)  ^{\frac{\lambda-1}{2}+\varepsilon
-\bar{\delta}}\leq Cx^{-\frac{3+\lambda}{2}-\bar{\delta}}%
\]
using $\varepsilon>\bar{\delta}$ whence:%
\begin{equation}
\left\vert J_{4}-I_{4}\left(  \tau\right)  x^{-\frac{3+\lambda}{2}}\right\vert
\leq Cx^{-\frac{3+\lambda}{2}-\bar{\delta}} \label{M3E7}%
\end{equation}

Combining (\ref{S6E2}), (\ref{S6E3}), (\ref{M3E6}), (\ref{M3E7}) we obtain
that the function $J_{R}\left(  \tau,x\right)  $ defined in (\ref{J2E0a})
satisfies%
\begin{equation}
\left\vert J_{R}\left(  \tau,x\right)  \right\vert \leq\frac{C\left\Vert
F\right\Vert _{X_{\frac{3}{2},2+\bar{\delta}}\left(  T\right)  }}%
{x^{\frac{3+\lambda}{2}+\bar{\delta}}}\ \ \ \text{for\ \ }x\geq1 \label{M3E8}%
\end{equation}

On the other hand, in order to estimate $J_{R}=J$ for $x\leq1$ we argue as
follows. We recall that $\left\vert J_{1}\right\vert \leq J_{1,1}+J_{1,2}$
with $J_{1,1},\ J_{1,2}$ as in (\ref{S5E7a}), (\ref{S5E7}). We have for
$R\leq1$ and $x\in\left(  \frac{R}{2},2R\right)  :$%
\[
J_{1,1}\left(  \tau,x\right)  \leq C\int_{0}^{\tau}ds\left\Vert F\left(
s,\cdot\right)  \right\Vert _{L^{\infty}\left(  \frac{R}{4},4R\right)  }\leq
M_{\infty}\left(  F;R\right)  \leq C\left\Vert F\right\Vert _{X_{\frac{3}%
{2},2+\bar{\delta}}\left(  T\right)  }%
\]

We estimate $J_{1,2}$ covering the set $\left(  s,x_{0}\right)  \in\left(
0,T\right)  \times\left(  0,1\right)  $ by means of union of the rectangles
$\left(  0,T\right)  \times\left[  \frac{1}{2^{n+1}},\frac{1}{2^{n}}\right]
,\ n=0,1,2,...$ . Let us denote this covering as $\mathcal{C}^{\left(
6\right)  }.$ Then:%
\begin{align*}
J_{1,2}\left(  \tau,x\right)   &  \leq\frac{C}{x^{3/2}}\int_{2x}^{1}dx_{0}%
\int_{0}^{\tau}\left(  \tau-s\right)  dsx_{0}^{\frac{\lambda}{2}}\left\vert
F\left(  s,x_{0}\right)  \right\vert +\\
&  +\frac{C}{x^{3/2}}\int_{1}^{\infty}dx_{0}\int_{\left(  \tau-x_{0}%
^{-\frac{\lambda-1}{2}}\right)  _{+}}^{\tau}\left(  \tau-s\right)
dsx_{0}^{\frac{\lambda}{2}}\left\vert F\left(  s,x_{0}\right)  \right\vert \\
&  =J_{1,2,1}\left(  \tau,x\right)  +J_{1,2,2}\left(  \tau,x\right)
\end{align*}%
\begin{align*}
J_{1,2,1}\left(  \tau,x\right)   &  \leq\frac{C}{x^{3/2}}\sum_{\mathcal{C}%
^{\left(  6\right)  }}\int_{\frac{1}{2^{n+1}}}^{\frac{1}{2^{n}}}dx_{0}\int
_{0}^{T}\left(  \tau-s\right)  dsx_{0}^{\frac{\lambda}{2}}\left\Vert F\left(
s,\cdot\right)  \right\Vert _{L^{\infty}\left(  \frac{1}{2^{n+1}},\frac
{1}{2^{n}}\right)  }\\
&  \leq\frac{C}{x^{3/2}}\sum_{\mathcal{C}^{\left(  6\right)  }}\int_{\frac
{1}{2^{n+1}}}^{\frac{1}{2^{n}}}dx_{0}\left(  \frac{1}{2^{n}}\right)
^{\frac{\lambda}{2}}M_{\infty}\left(  F;\frac{1}{2^{n}}\right) \\
&  \leq\frac{C\left\Vert F\right\Vert _{X_{\frac{3}{2},2+\bar{\delta}}\left(
T\right)  }}{x^{3/2}}\sum_{\mathcal{C}^{\left(  6\right)  }}\frac
{1}{2^{n\left(  \frac{\lambda-1}{2}\right)  }}\leq\frac{C\left\Vert
F\right\Vert _{X_{\frac{3}{2},2+\bar{\delta}}\left(  T\right)  }}{x^{3/2}}%
\end{align*}

The integral term in $J_{1,2,2}\left(  \tau,x\right)  $ can be estimated using
the covering $\mathcal{C}_{x,\tau}^{\left(  2\right)  }$ exactly in the same
way as in the estimate of $J_{1,2}$ for $x\geq1.$ It then follows that
$J_{1,2,2}\left(  \tau,x\right)  \leq C\left\Vert F\right\Vert _{X_{\frac
{3}{2},2+\bar{\delta}}\left(  T\right)  }x^{-3/2}$ whence:%
\begin{equation}
J_{1,2}\left(  \tau,x\right)  \leq\frac{C\left\Vert F\right\Vert _{X_{\frac
{3}{2},2+\bar{\delta}}\left(  T\right)  }}{x^{3/2}} \label{M4E1}%
\end{equation}

In order to estimate $J_{3}$ for $x\leq1$ we use (\ref{S5E8}). Notice that for
$s<\tau\leq T<1$ and $x\leq1$ we have $\left(  \tau-s\right)  ^{-\frac
{1}{\lambda-1}}x^{-\frac{3}{2}}\leq\left(  \tau-s\right)  ^{-\frac{1+\lambda
}{\lambda-1}}x^{-\frac{3+\lambda}{2}}$ whence (\ref{S5E8}) yields:%
\begin{align*}
\left\vert J_{3}\right\vert  &  \leq\frac{C}{x^{\frac{3}{2}}}\int_{\frac
{2x}{3}}^{\infty}dx_{0}\int_{0}^{\left(  \tau-x_{0}^{-\frac{\lambda-1}{2}%
}\right)  _{+}}ds\frac{\left\vert F\left(  s,x_{0}\right)  \right\vert
}{\left(  \tau-s\right)  ^{\frac{1}{\lambda-1}}}\\
&  \leq\frac{C}{x^{\frac{3}{2}}}\int_{\frac{2x}{3}}^{1}dx_{0}\int_{0}^{\left(
\tau-x_{0}^{-\frac{\lambda-1}{2}}\right)  _{+}}ds\frac{\left\vert F\left(
s,x_{0}\right)  \right\vert }{\left(  \tau-s\right)  ^{\frac{1}{\lambda-1}}%
}+\\
&  +\frac{C}{x^{\frac{3}{2}}}\int_{1}^{\infty}dx_{0}\int_{0}^{\left(
\tau-x_{0}^{-\frac{\lambda-1}{2}}\right)  _{+}}ds\frac{\left\vert F\left(
s,x_{0}\right)  \right\vert }{\left(  \tau-s\right)  ^{\frac{1}{\lambda-1}}}\\
&  \equiv J_{3,1}+J_{3,2}%
\end{align*}

The term $J_{3,2}$ can be estimated in the same manner as $J_{3}$ for
$x\geq1.$ We just use the covering $\mathcal{C}_{x,\tau}^{\left(  3\right)  }$
to obtain $J_{3,2}\leq C\left\Vert F\right\Vert _{X_{\frac{3}{2},2+\bar
{\delta}}\left(  T\right)  }x^{-3/2}.$ Since $J_{3,1}=0$ for $0\leq\tau\leq
T<1$ we then obtain:%
\begin{equation}
J_{3}\leq C\left\Vert F\right\Vert _{X_{\frac{3}{2},2+\bar{\delta}}\left(
T\right)  }x^{-3/2} \label{M4E3}%
\end{equation}

In order to estimate $J_{2}$ we use (\ref{S5E8a}), (\ref{S5E9}) as well as
(\ref{G1E3})
\begin{align*}
J_{2}  &  \leq Cx^{-\frac{3+\lambda}{2}}\int_{0}^{\frac{2x}{3}}dx_{0}\int
_{0}^{\tau}dsF\left(  s,x_{0}\right)  \left(  \tau-s\right)  x_{0}^{\lambda}\\
&  \leq Cx^{-\frac{3}{2}}\int_{0}^{1}dx_{0}\int_{0}^{\tau}dsF\left(
s,x_{0}\right)  \left(  \tau-s\right)  \left(  x_{0}\right)  ^{\frac{\lambda
}{2}}%
\end{align*}
and this integral can be estimated exactly as $J_{1,2,1}\left(  \tau,x\right)
,$ using the covering $\mathcal{C}^{\left(  6\right)  }:$%
\[
J_{2}\leq\frac{C}{x^{3/2}}\sum_{\mathcal{C}^{\left(  6\right)  }}\int
_{\frac{1}{2^{n+1}}}^{\frac{1}{2^{n}}}dx_{0}\int_{0}^{T}\left(  \tau-s\right)
dsx_{0}^{\frac{\lambda}{2}}\left\Vert F\left(  s,\cdot\right)  \right\Vert
_{L^{\infty}\left(  \frac{1}{2^{n+1}},\frac{1}{2^{n}}\right)  }%
\]
whence:%
\begin{equation}
J_{2}\leq C\left\Vert F\right\Vert _{X_{\frac{3}{2},2+\bar{\delta}}\left(
T\right)  }x^{-3/2} \label{M4E2}%
\end{equation}

Finally we notice that (\ref{S5E10}) implies that $J_{4}=0$ for $x\leq
1,\ 0\leq\tau\leq T<1.$ Combining (\ref{M4E1}), (\ref{M4E3}), (\ref{M4E2}) we
obtain:%
\[
\left\vert J\left(  \tau,x\right)  \right\vert \leq C\left\Vert F\right\Vert
_{X_{\frac{3}{2},2+\bar{\delta}}\left(  T\right)  }x^{-3/2}%
\]
for $0<x\leq1.$ Using then also (\ref{M3E8}) the result follows.

To conclude the proof of Lemma \ref{Lemma5} it only remains to show
(\ref{Z2E1d}). To this end we use similar covering arguments. First we
decompose the expression of $I\left(  \tau;F\right)  $ as:%
\[
I\left(  \tau;F\right)  =I_{1}+I_{2}+I_{3}+I_{4}%
\]%
\begin{align*}
I_{1}  &  =\int_{0}^{1}\frac{dx_{0}}{x_{0}}\int_{0}^{\tau}ds\left[  \cdot
\cdot\cdot\right]  \ \ ,\ \ I_{2}=\int_{1}^{\tau^{-\frac{2}{\lambda-1}}}%
\frac{dx_{0}}{x_{0}}\int_{0}^{\tau}ds\left[  \cdot\cdot\cdot\right] \\
I_{3}  &  =\int_{\tau^{-\frac{2}{\lambda-1}}}^{\infty}\frac{dx_{0}}{x_{0}}%
\int_{\tau-x_{0}^{-\frac{\lambda-1}{2}}}^{\tau}ds\left[  \cdot\cdot
\cdot\right]  \ \ ,\ \ I_{4}=\int_{\tau^{-\frac{2}{\lambda-1}}}^{\infty}%
\frac{dx_{0}}{x_{0}}\int_{0}^{\tau-x_{0}^{-\frac{\lambda-1}{2}}}ds\left[
\cdot\cdot\cdot\right]
\end{align*}

Using (\ref{G1E3}) as well as the definition of $\left\Vert \cdot\right\Vert
_{X_{\frac{3}{2},2+\bar{\delta}}\left(  T\right)  }$ we obtain, using
Cauchy-Schwartz:%
\begin{equation}
\left\vert I_{1}\right\vert \leq C\tau^{\frac{3}{2}}\left\Vert F\right\Vert
_{X_{\frac{3}{2},2+\bar{\delta}}\left(  T\right)  }\sum_{n=0}^{\infty}\left(
\frac{1}{2^{n}}\right)  ^{\lambda-\frac{1}{2}}\leq C\tau^{\frac{3}{2}%
}\left\Vert F\right\Vert _{X_{\frac{3}{2},2+\bar{\delta}}\left(  T\right)  }
\label{est1}%
\end{equation}

Using the splitting of the domains of integrations in rectangles as above as
well as (\ref{G1E3}) we also obtain the following estimates:%
\[
\left\vert I_{2}\right\vert \leq C\left\Vert F\right\Vert _{X_{\frac{3}%
{2},2+\bar{\delta}}\left(  T\right)  }\int_{1}^{\tau^{-\frac{2}{\lambda-1}}%
}dx_{0}\left(  x_{0}\right)  ^{\lambda-\left(  2+\bar{\delta}\right)  }%
\int_{0}^{\tau}\left(  \tau-s\right)  ds
\]%
\[
\left\vert I_{3}\right\vert \leq C\left\Vert F\right\Vert _{X_{\frac{3}%
{2},2+\bar{\delta}}\left(  T\right)  }\int_{\tau^{-\frac{2}{\lambda-1}}%
}^{\infty}dx_{0}\left(  x_{0}\right)  ^{\lambda-\left(  2+\bar{\delta}\right)
}\int_{\tau-x_{0}^{-\frac{\lambda-1}{2}}}^{\tau}\left(  \tau-s\right)  ds
\]%
\[
\left\vert I_{4}\right\vert \leq C\left\Vert F\right\Vert _{X_{\frac{3}%
{2},2+\bar{\delta}}\left(  T\right)  }\int_{\tau^{-\frac{2}{\lambda-1}}%
}^{\infty}dx_{0}\left(  x_{0}\right)  ^{-\left(  2+\bar{\delta}\right)  }%
\int_{0}^{\tau-x_{0}^{-\frac{\lambda-1}{2}}}\left(  \tau-s\right)
^{-\frac{\lambda+1}{\lambda-1}}ds
\]

Therefore, since the three integrals on the right-hand side are bounded by
$\leq C\tau^{\frac{2\bar{\delta}}{\lambda-1}}\left\Vert F\right\Vert
_{X_{\frac{3}{2},2+\bar{\delta}}\left(  T\right)  }$:%
\[
\left\vert I_{2}\right\vert +\left\vert I_{3}\right\vert +\left\vert
I_{4}\right\vert \leq C\tau^{\frac{2\bar{\delta}}{\lambda-1}}\left\Vert
F\right\Vert _{X_{\frac{3}{2},2+\bar{\delta}}\left(  T\right)  }%
\]

Combining this with (\ref{est1}) we obtain (\ref{Z2E1d}).
\end{proof}

\bigskip

We now derive the asymptotics (\ref{G2E5}) for the solutions of (\ref{G2E3})
in a pointwise sense. More precise estimates for the regularity of the error
terms will be derived later.

\bigskip

\begin{lemma}
\label{Lemma6} Suppose that $\varphi\in\mathcal{Z}_{\frac{3+\lambda}{2}%
}^{\sigma;\frac{1}{2}}\left(  T\right)  $ solves (\ref{G2E3}) with $F\in
Y_{\frac{3}{2},2+\bar{\delta}}^{\sigma}\left(  T\right)  $ and $\bar{\delta
}<r.$ Then (\ref{G2E2}) holds with $\mathcal{W}\left(  \tau\right)  $ as in
(\ref{G2E5}) and:%
\begin{equation}
\left\vert \left\vert \left\vert \varphi_{R}\right\vert \right\vert
\right\vert _{\frac{3}{2},\frac{3+\lambda}{2}+\bar{\delta}}\leq C\left(
\left\Vert F\right\Vert _{X_{\frac{3}{2},2+\bar{\delta}}^{\sigma}\left(
T\right)  }+\left\Vert \varphi\right\Vert _{\mathcal{Z}_{\frac{3+\lambda}{2}%
}^{\sigma;\frac{1}{2}}\left(  T\right)  }\right)  \label{M4E4}%
\end{equation}

\end{lemma}

\begin{proof}
[Proof of Lemma \ref{Lemma6}]Lemma \ref{Le3} as well as the fact that
$r>\bar{\delta}$ implies that $\left\Vert \left(  {\mathcal{L}}_{f_{0}%
}-L\right)  \left[  \varphi\right]  \right\Vert _{X_{\frac{3}{2},2+\bar
{\delta}}\left(  T\right)  }<\infty.$ Lemma \ref{Lemma6} then follows applying
\ref{Lemma5} with $F$ replaced by%
\begin{equation}
\bar{F}=F+\left(  {\mathcal{L}}_{f_{0}}-L\right)  \left[  \varphi\right]  \ .
\label{Z2E2}%
\end{equation}

\end{proof}

In order to prove suitable regularity properties of the remainder $\varphi
_{R}$ in (\ref{G2E2}) we will need to obtain some regularity properties for
the function $I\left(  \cdot;\bar{F}\right)  ,$ where the functional $I$ is as
in (\ref{Z2E1}) and $\bar{F}$ as in (\ref{Z2E2}). The rationale behind the
argument is that an equation for $\varphi_{R}$ contains terms that can be
estimated only if some regularity for $I\left(  \tau;\bar{F}\right)  $ is available.

We have:

\begin{lemma}
\label{Lemma7}Suppose that $\sup_{\tau\in\left[  0,T\right]  }\left\Vert
\left\vert F\left(  \tau,\cdot\right)  \right\vert \right\Vert _{\frac{3}%
{2},2+\bar{\delta}}<\infty$.  Let $\Delta=\tau_{1}-\tau_{0},$ with
$0\leq\tau_{0}\leq\tau_{1}\leq T.$ Then:%
\[
\left\vert I\left(  \tau_{1};F\right)  -I\left(  \tau_{0};F\right)
\right\vert \leq C\left(  \Delta\right)  ^{\frac{2\bar{\delta}}{\lambda-1}}%
\]
with $C>0$ depending on $\sup_{\tau\in\left[  0,T\right]  }\left\Vert
\left\vert F\left(  \tau,\cdot\right)  \right\vert \right\Vert _{\frac{3}%
{2},2+\bar{\delta}},\ \lambda.$
\end{lemma}

\begin{proof}
We write, using (\ref{Z2E1}):%
\begin{align*}
&  I\left(  \tau_{1};F\right)  -I\left(  \tau_{0};F\right) \\
&  =\int_{0}^{\infty}\frac{dx_{0}}{x_{0}}\int_{\tau_{0}}^{\tau_{1}}%
ds\Theta\left(  \left(  \tau_{1}-s\right)  x_{0}^{\frac{\lambda-1}{2}}\right)
x_{0}^{\frac{3+\lambda}{2}}F\left(  s,x_{0}\right)  +\\
&  +\int_{0}^{\infty}\frac{dx_{0}}{x_{0}}\int_{0}^{\tau_{0}}ds\left[
\Theta\left(  \left(  \tau_{1}-s\right)  x_{0}^{\frac{\lambda-1}{2}}\right)
-\Theta\left(  \left(  \tau_{0}-s\right)  x_{0}^{\frac{\lambda-1}{2}}\right)
\right]  x_{0}^{\frac{3+\lambda}{2}}F\left(  s,x_{0}\right) \\
&  =L_{1}+L_{2}%
\end{align*}

We have the global estimate:%
\[
\left\vert F\left(  s,x_{0}\right)  \right\vert \leq Cx_{0}^{-\left(
2+\bar{\delta}\right)  }%
\]

Then:%
\begin{align*}
\left\vert L_{1}\right\vert  &  \leq C\int_{0}^{\infty}\frac{dx_{0}}{x_{0}%
}\int_{\tau_{0}}^{\tau_{1}}ds\left\vert \Theta\left(  \left(  \tau
_{1}-s\right)  x_{0}^{\frac{\lambda-1}{2}}\right)  \right\vert x_{0}%
^{\frac{3+\lambda}{2}}x_{0}^{-\left(  2+\bar{\delta}\right)  }\\
&  =C\int_{0}^{\infty}\frac{dx_{0}}{\left(  x_{0}\right)  ^{1+\bar{\delta}}%
}\int_{0}^{\Delta x_{0}^{\frac{\lambda-1}{2}}}\left\vert \Theta\left(
s\right)  \right\vert ds=C\left(  \Delta\right)  ^{\frac{2\delta}{\lambda-1}%
}\int_{0}^{\infty}\frac{dy}{\left(  y\right)  ^{1+\bar{\delta}}}\int
_{0}^{y^{\frac{\lambda-1}{2}}}\left\vert \Theta\left(  s\right)  \right\vert
ds
\end{align*}

Using the fact that $\Psi_{1}\left(  y\right)  =\int_{0}^{y^{\frac{\lambda
-1}{2}}}\left\vert \Theta\left(  s\right)  \right\vert ds$ satisfies
$\left\vert \Psi_{1}\left(  y\right)  \right\vert \leq C\min\left\{
y^{\lambda-1},1\right\}  $ (cf. (\ref{G1E3})) we obtain:%
\begin{equation}
\left\vert L_{1}\right\vert \leq C\left(  \Delta\right)  ^{\frac{2\bar{\delta
}}{\lambda-1}} \label{B1}%
\end{equation}

We now estimate $\left\vert L_{2}\right\vert .$ We have:%
\begin{align*}
\left\vert L_{2}\right\vert  &  \leq\int_{0}^{\infty}\frac{dx_{0}}{x_{0}}%
\int_{0}^{\tau_{0}}ds\left\vert \Theta\left(  \left(  \tau_{1}-s\right)
x_{0}^{\frac{\lambda-1}{2}}\right)  -\Theta\left(  \left(  \tau_{0}-s\right)
x_{0}^{\frac{\lambda-1}{2}}\right)  \right\vert x_{0}^{\frac{3+\lambda}{2}%
}x_{0}^{-\left(  2+\bar{\delta}\right)  }\\
&  =\int_{0}^{\infty}\frac{dx_{0}}{\left(  x_{0}\right)  ^{1+\bar{\delta}}%
}\int_{0}^{\frac{\tau_{0}}{\Delta}}d\mathcal{\rho}\left\vert \Theta\left(
\left(  1+\mathcal{\rho}\right)  \Delta x_{0}^{\frac{\lambda-1}{2}}\right)
-\Theta\left(  \mathcal{\rho}\Delta x_{0}^{\frac{\lambda-1}{2}}\right)
\right\vert x_{0}^{\frac{\lambda-1}{2}}\Delta\\
&  \leq\left(  \Delta\right)  ^{\frac{2\bar{\delta}}{\lambda-1}}\int
_{0}^{\infty}\frac{\Psi_{2}\left(  y\right)  dy}{\left(  y\right)
^{1+\bar{\delta}}}\int_{0}^{\infty}dZ\left\vert \Theta\left(  y^{\frac
{\lambda-1}{2}}+Z\right)  -\Theta\left(  Z\right)  \right\vert
\end{align*}
where $\Psi_{2}\left(  y\right)  =\int_{0}^{\infty}dZ\left\vert \Theta\left(
y^{\frac{\lambda-1}{2}}+Z\right)  -\Theta\left(  Z\right)  \right\vert $. We
have used the change of variables $y=\left(  \Delta\right)  ^{\frac{2}%
{\lambda-1}}x_{0},$ and $Z=y^{\frac{\lambda-1}{2}}\mathcal{\rho}.$ Using
(\ref{G1E3}) we obtain $\Psi_{2}\left(  y\right)  \leq C\min\left\{
y^{\frac{\lambda-1}{2}},1\right\}  .$ Then:%
\begin{equation}
\left\vert L_{2}\right\vert \leq C\left(  \Delta\right)  ^{\frac{2\bar{\delta
}}{\lambda-1}} \label{B2}%
\end{equation}

Lemma \ref{Lemma7} follows combining (\ref{B1}), (\ref{B2}).
\end{proof}

\bigskip

We can now prove Proposition \ref{PropAsympt} deriving suitable regularity
estimates for the function $\varphi_{R}.$

\begin{proof}
[Proof of Proposition \ref{PropAsympt}]Given the function $\varphi_{R}$
defined in (\ref{G2E2}) our goal is to prove (\ref{G2E6}). To this end, given
$I\left(  \tau;F\right)  $ in (\ref{G2E5}) defined in $0\leq\tau\leq T$ we
extend it to $\mathbb{R}$ defining $I\left(  \tau;\bar{F}\right)  =I\left(
0;\bar{F}\right)  $ for $\tau\leq0$ and $I\left(  \tau;\bar{F}\right)
=I\left(  T;\bar{F}\right)  $ for $\tau\geq T.$we define a function $\tilde
{I}_{R}\left(  \tau\right)  $ for $0\leq\tau\leq T$ by means of:%
\begin{equation}
\tilde{I}_{R}\left(  \tau;\bar{F}\right)  =\left(  I\left(  \cdot;\bar
{F}\right)  \ast\chi_{R}\right)  \left(  \tau\right)  \label{G2E7}%
\end{equation}
with $\chi_{R}\left(  \tau\right)  =R^{\frac{\lambda-1}{2}}\chi_{1}\left(
R^{\frac{\lambda-1}{2}}\tau\right)  $ where the nonnegative function $\chi
_{1}\in C^{\infty}\left(  \mathbb{R}\right)  $ is compactly supported in
$\left[  -1,1\right]  $ and it satisfies $\int_{\mathbb{R}}\chi_{1}\left(
\tau\right)  d\tau=1.$ We extend $I\left(  \tau;\bar{F}\right)  $ in
(\ref{G2E7}) as It can be easily seen, using Lemma \ref{Lemma7} with
$\bar{\delta}<r$ with $r$ as in (\ref{U3E4a}) that:%
\begin{equation}
\left\vert \tilde{I}_{R}\left(  \tau;\bar{F}\right)  -I\left(  \tau;\bar
{F}\right)  \right\vert \leq CR^{-\bar{\delta}}\ \ ,\ \ \tau\in\left[
0,T\right]  \ \ ,\ \ \left\vert \frac{d\tilde{I}_{R}\left(  \tau;\bar
{F}\right)  }{d\tau}\right\vert \leq R^{\frac{\lambda-1}{2}-\bar{\delta}%
}\ \ ,\ \ \tau\in\left[  0,T\right]  \label{G2E9}%
\end{equation}

We define $\tilde{\varphi}_{R}\left(  \tau,x\right)  =\varphi\left(
\tau,x\right)  -\tilde{I}_{R}\left(  \tau;\bar{F}\right)  x^{-\frac{3+\lambda
}{2}}\xi\left(  x\right)  $ where $\xi$ is as in (\ref{Z1E2a}). Then, using
(\ref{G2E2}) we obtain $\left\vert \tilde{\varphi}_{R}\left(  \tau,x\right)
\right\vert \leq\left\vert \varphi_{R}\left(  \tau,x\right)  \right\vert
+\left\vert \tilde{I}_{R}\left(  \tau;\bar{F}\right)  -I\left(  \tau;\bar
{F}\right)  \right\vert \xi\left(  x\right)  x^{-\frac{3+\lambda}{2}}$ whence
(\ref{M4E4}), (\ref{G2E9}) yields:
\[
\left\vert \tilde{\varphi}_{R}\left(  \tau,x\right)  \right\vert \leq
CR^{-\left(  \frac{3+\lambda}{2}+\bar{\delta}\right)  }\ \ ,\ \ x\in\left[
\frac{R}{2},R\right]  \ \ ,\ \ \ \tau\in\left[  0,T\right]
\]
where $\tilde{\varphi}_{R}$ satisfies:%
\begin{align}
\left(  \tilde{\varphi}_{R}\right)  _{\tau}  &  =\mathcal{L}_{f_{0}}\left[
\tilde{\varphi}_{R}\right]  +\tilde{F}\left(  \tau,x\right) \label{G2E10}\\
\tilde{F}\left(  \tau,x\right)   &  =\bar{F}\left(  \tau,x\right)
-\frac{d\tilde{I}_{R}\left(  \tau;\bar{F}\right)  }{d\tau}x^{-\frac{3+\lambda
}{2}}+\tilde{I}_{R}\left(  \tau;\bar{F}\right)  \mathcal{L}_{f_{0}}\left[
x^{-\frac{3+\lambda}{2}}\xi\left(  x\right)  \right] \nonumber
\end{align}
where $\bar{F}$ is as in (\ref{Z2E2}). Using (\ref{G2E9}) as well as the fact
that $\bar{F}\in Y_{\frac{3}{2},2+\bar{\delta}}^{\sigma}$ and the fact that
$\mathcal{L}_{f_{0}}\left[  x^{-\frac{3+\lambda}{2}}\xi\left(  x\right)
\right]  $ decreases like $x^{-\left(  2+\bar{\delta}\right)  }$ as
$x\rightarrow\infty,$ including derivatives, it follows that $\tilde{F}\in
Y_{\frac{3}{2},2+\bar{\delta}}^{\sigma}$. We now use the rescaling
$\tilde{\varphi}_{R}\left(  \tau,x\right)  =R^{-\left(  \frac{3+\lambda}%
{2}+\bar{\delta}\right)  }\Phi_{R}\left(  \mathcal{\rho},X\right)  ,$ $x=RX,$
$\ \tau=\frac{\mathcal{\rho}}{R^{\frac{\lambda-1}{2}}}.$ The function
$\Phi_{R}\left(  \mathcal{\rho},X\right)  $ then satisfies an equation of the
form $\frac{\partial\Phi_{R}}{\partial\mathcal{\rho}}=\mathcal{L}_{f_{0}}%
^{R}\left[  \Phi_{R}\right]  +\tilde{F}_{R}$ where $\mathcal{L}_{f_{0}}^{R}$
is the operator obtained rescaling $\mathcal{L}_{f_{0}}$ as well as $\tilde
{F}_{R}$ that is obtained rescaling $\tilde{F}$ (cf. (5.82) in the proof of
Lemma 5.8 of \cite{EV2}).

\texttt{... See if it is worth to copy the formulas in \cite{EV2}. ...}

We can now argue as in the proof of Lemma 5.8 in \cite{EV2}, applying Theorem
\ref{S8T3-101} in Section \ref{Summary} to prove:
\[
\left\Vert \tilde{\varphi}_{R}\right\Vert _{\mathcal{Z}_{p}^{\sigma;\frac
{1}{2}}\left(  T\right)  }\leq C\left\Vert F\right\Vert _{Y_{\frac{3}%
{2},2+\bar{\delta}}^{\sigma}}%
\]

This concludes the Proof of Proposition \ref{PropAsympt}.
\end{proof}

\section{FIXED POINT ARGUMENT.}

\bigskip

In this Section we prove Theorem \ref{Th1} by means of a fixed point argument.
Our strategy is to define a function $\tilde{h}$ by means of (\ref{S2E4}) with
initial data $\tilde{h}\left(  0,x\right)  =0.$ To this end, we define two
auxiliary functions $\tilde{h}_{1},\ \tilde{h}_{2}$ that will be due
respectively to the source terms $\left(  \frac{Q\left[  h\right]  }%
{\Lambda\left(  \tau\right)  }+\Lambda\left(  \tau\right)  Q\left[
f_{0}\right]  \right)  $ and $\left(  -\Lambda_{\tau}f_{0}\left(  x\right)
\right)  .$ The goal of this splitting is to treat in a separate way the term
containing $\Lambda_{\tau}$ since this term will require some careful analysis
of the time regularity of the solutions. The function $h$ will be then defined
as $\tilde{h}\left(  t,x\right)  =\tilde{h}_{1}\left(  t,x\right)  +\tilde
{h}_{2}\left(  t,x\right)  .$

\subsection{Construction of the function $\tilde{h}_{1}.$}

We define $\tilde{h}_{1}$ as the solution of the problem:%
\begin{align}
\tilde{h}_{1,\tau}  &  ={\mathcal{L}}_{f_{0}}\left[  \tilde{h}_{1}\right]
+\frac{Q\left[  h\right]  }{\Lambda\left(  \tau\right)  }+\Lambda\left(
\tau\right)  Q\left[  f_{0}\right] \label{Z2E3}\\
\tilde{h}_{1}\left(  0,x\right)   &  =0\ \ ,\ \ x>0 \label{Z2E4}%
\end{align}

The existence, uniqueness and main regularity properties of the function
$\tilde{h}_{1}$ are given in the following result:

\begin{proposition}
\label{Propositionh_1}Given $h\in\mathcal{Z}_{\bar{p}}^{\sigma;\frac{1}{2}%
}\left(  T\right)  $ with $\bar{\delta}<r,$ $f_{0}$ as in (\ref{Z1E1a}%
)-(\ref{Z1E2b}), $0<T\leq1$ and $\Lambda\in C\left[  0,T\right]  $ satisfying
$\left\vert \Lambda\left(  \tau\right)  -1\right\vert \leq\frac{1}{4}$ there
exist a unique function $\tilde{h}_{1}\in\mathcal{Z}_{\frac{3+\lambda}{2}%
}^{\sigma;\frac{1}{2}}\left(  T\right)  $ solution of (\ref{Z2E3}),
(\ref{Z2E4}) if $T$ is sufficiently small. Moreover, there exist a function
$\mathcal{G}\left[  \cdot;h,\Lambda\right]  \in C\left[  0,T\right]  $ such
that the function $r_{1}$ defined by means of:%
\begin{equation}
\tilde{h}_{1}\left(  \tau,x\right)  -\mathcal{G}\left[  \tau;h,\Lambda\right]
\xi\left(  x\right)  x^{-\frac{3+\lambda}{2}}=r_{1}\left(  \tau,x;h,\Lambda
\right)  \label{G3E1}%
\end{equation}
satisfies $r_{1}\left(  \cdot,\cdot;h,\Lambda\right)  \in\mathcal{Z}_{\bar{p}%
}^{\sigma;\frac{1}{2}}\left(  T\right)  .$\ The function $\xi\left(
\cdot\right)  $ is as in (\ref{Z1E2a}). The mappings
\begin{equation}%
\begin{array}
[c]{c}%
\mathcal{Z}_{\bar{p}}^{\sigma;\frac{1}{2}}\left(  T\right)  \times C\left[
0,T\right]  \rightarrow C\left[  0,T\right] \\
\ \ \ \ \ \ \ \ \ \ \ \ \ \ \ \ \ \ \ \ \ \ \left(  h,\Lambda\right)
\longmapsto\mathcal{G}\left[  \cdot;h,\Lambda\right]
\end{array}
\label{G3E1a}%
\end{equation}%
\begin{equation}%
\begin{array}
[c]{c}%
\mathcal{Z}_{\bar{p}}^{\sigma;\frac{1}{2}}\left(  T\right)  \times C\left[
0,T\right]  \rightarrow\mathcal{Z}_{\bar{p}}^{\sigma;\frac{1}{2}}\left(
T\right) \\
\ \ \ \ \ \ \ \ \ \ \ \ \ \ \ \ \ \ \ \ \ \ \left(  h,\Lambda\right)
\longmapsto r_{1}\left(  \cdot,\cdot;h,\Lambda\right)
\end{array}
\label{G3E1b}%
\end{equation}
are Lipschitz continuous assuming that $C\left[  0,T\right]  $ is endowed with
the uniform topology and $T$ is sufficiently small. Moreover there exists
$\rho_{0},\ T_{0}$ such that the Lipschitz constant for the maps can be made
arbitrarily small if $T<T_{0}$ and $\left\Vert h\right\Vert _{\mathcal{Z}%
_{\bar{p}}^{\sigma;\frac{1}{2}}\left(  T\right)  }\leq\rho_{0}$.
\end{proposition}

We will prove Proposition \ref{Propositionh_1} with the help of some auxiliary
Lemmas. Estimates for $Q\left[  f_{0}\right]  \left(  x_{0}\right)  $ are
obtained in the next result.

\begin{lemma}
\label{Le1} Let $f_{0}$ be as in (\ref{Z1E1a})-(\ref{Z1E2b}). We assume that
this function is defined in $\left[  0,T\right]  \times\mathbb{R}^{+}$ as a
function independent on $\tau.$ We then have:%
\begin{equation}
\left\Vert Q\left[  f_{0}\right]  \right\Vert _{Y_{\frac{3}{2},\left(
2+r\right)  }^{\sigma}\left(  T\right)  }\leq C \label{S5E6a}%
\end{equation}%
\begin{equation}
\left\Vert Q\left[  f_{0}\right]  \right\Vert _{Y_{\frac{3}{2},\left(
2+\bar{\delta}_{1}\right)  }^{\sigma}\left(  T\right)  }\leq C\max\left\{
\sqrt{T},T^{\frac{2\left(  r-\bar{\delta}\right)  }{\lambda-1}}\right\}
\ \ \text{,\ \ }0<\bar{\delta}_{1}\leq r \label{S5E6b}%
\end{equation}%
\begin{equation}
\left\vert Q\left[  f_{0}\right]  \left(  x_{0}\right)  \right\vert \leq
\frac{C}{1+x_{0}^{2+r}}\ \ ,\ \ \ x_{0}>0 \label{S5E6}%
\end{equation}

\end{lemma}

\begin{proof}
[Proof of Lemma \ref{Le1}]We use the decomposition (\ref{S1E5})
\[
f_{0}=\bar{f}_{0}\left(  x\right)  +h_{0}\left(  x\right)
\ \ \ ,\ \ \ \ \ \bar{f}_{0}\left(  x\right)  =\frac{\xi\left(  x\right)
}{x^{\frac{3+\lambda}{2}}}%
\]
with $\xi\left(  \cdot\right)  $ as in (\ref{Z1E2a}). Then:%
\begin{equation}
\left\vert h_{0}\left(  x\right)  \right\vert +x\left\vert h_{0}^{\prime
}\left(  x\right)  \right\vert \leq\frac{CB}{x^{\frac{3+\lambda}{2}+r}%
}\ \ ,\ \ x\geq1 \label{M4E4a}%
\end{equation}

Then:%
\begin{equation}
Q\left[  f_{0}\right]  \left(  x_{0}\right)  =Q\left[  \bar{f}_{0}\right]
\left(  x_{0}\right)  +{\mathcal{L}}_{\bar{f}_{0}}\left[  h_{0}\right]
\left(  x_{0}\right)  +Q\left[  h_{0}\right]  \left(  x_{0}\right)
\label{M4E5}%
\end{equation}

\[
Q\left[  \bar{f}_{0}\right]  \left(  x\right)  =-\int_{\frac{x}{2}}^{\infty
}\left(  xy\right)  ^{\frac{\lambda}{2}}\bar{f}_{0}\left(  x\right)  \bar
{f}_{0}\left(  y\right)  dy+\int_{0}^{\frac{x}{2}}y^{\frac{\lambda}{2}}\bar
{f}_{0}\left(  y\right)  \left[  \left(  x-y\right)  ^{\frac{\lambda}{2}}%
\bar{f}_{0}\left(  x-y\right)  -x^{\frac{\lambda}{2}}\bar{f}_{0}\left(
x\right)  \right]  dy
\]

Using the fact that $\bar{f}_{0}\left(  y\right)  =y^{-\frac{3+\lambda}{2}}$
for large $y$ we obtain:%
\begin{align*}
Q\left[  \bar{f}_{0}\right]  \left(  x\right)   &  =-\int_{\frac{x}{2}%
}^{\infty}\frac{dy}{x^{\frac{3}{2}}y^{\frac{3}{2}}}+\int_{0}^{\frac{x}{2}%
}\left[  \frac{1}{\left(  x-y\right)  ^{\frac{3}{2}}}-\frac{1}{x^{\frac{3}{2}%
}}\right]  \frac{dy}{y^{\frac{3}{2}}}+\\
&  +\int_{0}^{2}\left[  \xi\left(  y\right)  -1\right]  \left[  \frac
{1}{\left(  x-y\right)  ^{\frac{3}{2}}}-\frac{1}{x^{\frac{3}{2}}}\right]
\frac{dy}{y^{\frac{3}{2}}}%
\end{align*}

The first two terms on the right-hand side cancel out. The third one can be
estimated, using Taylor's expansion, as:%
\[
\int_{0}^{2}\frac{1}{y^{\frac{3}{2}}}\left\vert \frac{1}{\left(  x-y\right)
^{\frac{3}{2}}}-\frac{1}{x^{\frac{3}{2}}}\right\vert dy\leq\frac{C}%
{x^{\frac{5}{2}}}\ \ ,\ \ x>1
\]
whence:%
\begin{equation}
\left\vert Q\left[  \bar{f}_{0}\right]  \left(  x\right)  \right\vert
\leq\frac{C}{x^{\frac{5}{2}}}\ \ ,\ \ x>1 \label{M4E6}%
\end{equation}

We can estimate $Q\left[  h_{0}\right]  \left(  x_{0}\right)  $ as:%
\[
\left\vert Q\left[  h_{0}\right]  \left(  x\right)  \right\vert \leq
\frac{CB^{2}}{x^{2+2r}}+\left\vert \int_{0}^{\frac{x}{2}}\frac{d q}{d y}\left(
y\right)  S\left(  x,y;h_{0}\right)  dy\right\vert
\]
where:%
\[
q\left(  y;h_{0}\right)  =\int_{y}^{\infty}\xi^{\frac{\lambda}{2}}h_{0}\left(
\xi\right)  d\xi\ \ ,\ \ S\left(  x,y;h_{0}\right)  =\left[  \left(
x-y\right)  ^{\frac{\lambda}{2}}h_{0}\left(  x-y\right)  -x^{\frac{\lambda}%
{2}}h_{0}\left(  x\right)  \right]
\]

Integrating by parts we obtain:
\[
\left\vert \int_{0}^{\frac{x}{2}}\frac{d q}{d y}\left(  y;h_{0}\right)  S\left(
x,y;h_{0}\right)  dy\right\vert \leq\frac{CB^{2}}{x^{2+2r}}+CB^{2}\left\vert
\int_{0}^{\frac{x}{2}}\frac{1}{y^{\frac{1}{2}+r}}\frac{1}{\left(  x-y\right)
^{\frac{5}{2}+r}}dy\right\vert \leq\frac{CB^{2}}{x^{2+2r}}%
\]

Then:%
\begin{equation}
\left\vert Q\left[  h_{0}\right]  \left(  x\right)  \right\vert \leq
\frac{CB^{2}}{x^{2+2r}}\ \ ,\ \ x>1 \label{M4E7}%
\end{equation}

Finally we write ${\mathcal{L}}_{f_{0}}\left[  h_{0}\right]  $ as:%
\begin{align*}
{\mathcal{L}}_{f_{0}}\left[  h_{0}\right]   &  =-\int_{0}^{\frac{x}{2}%
}S\left(  x,y;f_{0}\right)  q\prime\left(  y;h_{0}\right)  dy-x^{\lambda
/2}\,f_{0}(x)q\left(  \frac{x}{2};h_{0}\right)  +\\
&  -\int_{0}^{\frac{x}{2}}S\left(  x,y;h_{0}\right)  q\prime\left(
y;f_{0}\right)  dy-x^{\lambda/2}h_{0}(x)q\left(  \frac{x}{2};f_{0}\right)
\end{align*}

An inmediate computation shows that the second and fourth terms on the right
can be estimated as $CBx^{-\left(  2+r\right)  }.$ The other two terms can be
estimated integrating by parts. Then:%
\begin{equation}
\left\vert {\mathcal{L}}_{\bar{f}_{0}}\left[  h_{0}\right]  \left(  x\right)
\right\vert \leq\frac{CB}{x^{2+r}}\ ,\ \ x>1 \label{M4E8}%
\end{equation}

In the region $x\leq1$ we have trivially boundedness of $Q\left[
f_{0}\right]  \left(  x_{0}\right)  .$ Combining this with (\ref{M4E5}),
(\ref{M4E6}), (\ref{M4E7}), (\ref{M4E8}) we obtain (\ref{S5E6}). Estimate
(\ref{S5E6a}) follows similarly using the differentiability properties assumed
for $f_{0}$.

It only remains to prove (\ref{S5E6b}). To this end, notice that due to the
definition of $\left\Vert \cdot\right\Vert _{Y_{\frac{3}{2},\left(
2+\bar{\delta}\right)  }^{\sigma}\left(  T\right)  }$ we need to estimate
$L^{2}$ norms in time for $R\leq1.$ Since we have estimates in $L^{\infty}$
for $Q\left[  f_{0}\right]  $ we then obtain a dependence on $T$ like
$\sqrt{T}.$ Similar estimates can be obtained, using also the definition of
the norms $\left\Vert \cdot\right\Vert _{Y_{\frac{3}{2},\left(  2+r\right)
}^{\sigma}\left(  T\right)  }$ for all the values of $R\leq T^{-\frac
{2}{\lambda-1}},$ since for these values there is not splitting of the domain
of integration in the $t$ variable. In the region where $R>T^{-\frac
{2}{\lambda-1}}$ we use the fact that $Q\left[  f_{0}\right]  $ is pointwise
estimated by $R^{-\left(  2+r\right)  }.$ Therefore:%
\[
R^{\left(  2+\bar{\delta}\right)  }\left(  N_{\infty}\left(  Q\left[
f_{0}\right]  ;t_{0},R\right)  +N_{2;\sigma}\left(  Q\left[  f_{0}\right]
;t_{0},R\right)  \right)  \leq CR^{\bar{\delta}-r}\leq CT^{\frac{2\left(
r-\bar{\delta}\right)  }{\lambda-1}}%
\]
for $R>T^{-\frac{2}{\lambda-1}}.$ Therefore (\ref{S5E6b}) follows.
\end{proof}

As a next step we estimate the quadratic terms in (\ref{Z2E3}).

\begin{lemma}
\label{Le2} Given $h\in\mathcal{Z}_{\bar{p}}^{\sigma;\frac{1}{2}}\left(
T\right)  $ and $\Lambda$ as in Proposition \ref{Propositionh_1}.\ Then:%
\begin{equation}
\left\Vert \frac{Q\left[  h\right]  }{\Lambda\left(  \cdot\right)
}\right\Vert _{Y_{\frac{3+\lambda}{2},\left(  2+\bar{\delta}\right)  }%
^{\sigma}}\leq C\left\Vert h\right\Vert _{\mathcal{Z}_{\bar{p}}^{\sigma
;\frac{1}{2}}\left(  T\right)  }^{2} \label{M4E9}%
\end{equation}

\end{lemma}

Moreover:
\begin{align}
&  \left\Vert \frac{Q\left[  h_{1}\right]  }{\Lambda_{1}\left(  \cdot\right)
}-\frac{Q\left[  h_{2}\right]  }{\Lambda_{2}\left(  \cdot\right)  }\right\Vert
_{Y_{\frac{3+\lambda}{2},\left(  2+\bar{\delta}\right)  }^{\sigma}%
}\label{M4E10}\\
&  \leq C\left(  \sum_{k=1}^{2}\left\Vert h_{k}\right\Vert _{\mathcal{Z}%
_{\bar{p}}^{\sigma;\frac{1}{2}}\left(  T\right)  }\right)  \left(  \left\Vert
h_{1}-h_{2}\right\Vert _{\mathcal{Z}_{\bar{p}}^{\sigma;\frac{1}{2}}\left(
T\right)  }+\left\Vert \Lambda_{1}-\Lambda_{2}\right\Vert _{C\left[
0,T\right]  }\right) \nonumber
\end{align}

\begin{proof}
[Proof of Lemma \ref{Le2}]Estimates (\ref{M4E9}), (\ref{M4E10}) are just a
consequence of Propositions \ref{LemmaQuad}, \ref{LemmaQuadLip} as well as the
fact that $\frac{1}{2}\leq\Lambda\left(  \cdot\right)  \leq\frac{3}{2}%
.$\texttt{ }
\end{proof}

\begin{proof}
[Proof of Proposition \ref{Propositionh_1}]Existence and uniqueness of the
function $\tilde{h}_{1}$ follow from the results in \cite{EV2} (cf.
\ref{ExistenceLredonda}) combined with Lemmas \ref{Le1}, \ref{Le2}.\texttt{ }

On the other hand, the decay and regularity properties of the the function
$r_{1}$ defined in (\ref{G3E1}) are a consequence of Proposition
\ref{PropAsympt}. In order to apply this Proposition some regularity and decay
for the source terms $\frac{Q\left[  h\right]  }{\Lambda\left(  \tau\right)
}$ and $\Lambda\left(  \tau\right)  Q\left[  f_{0}\right]  $ are needed$.$ In
the case of $\frac{Q\left[  h\right]  }{\Lambda\left(  \tau\right)  }$ such
properties are a consequence of Proposition \ref{LemmaQuad}. The corresponding
properties for $\Lambda\left(  \tau\right)  Q\left[  f_{0}\right]  $ follow
from Lemma \ref{Le1} and the fact that $r>\bar{\delta}.$ The function
$\mathcal{G}\left[  \tau;h,\Lambda\right]  $ is given by the function
$\mathcal{W}\left(  \tau\right)  $ in (\ref{G2E5}) with source $F$ given by
$\frac{Q\left[  h\right]  }{\Lambda\left(  \tau\right)  }+\Lambda\left(
\tau\right)  Q\left[  f_{0}\right]  $. Notice that the linearity of the
equation satisfied by $\tilde{h}_{1}$ as well as the Lipschitz property for
$Q\left[  h\right]  $ in Proposition \ref{LemmaQuadLip} implies that the map
$h\rightarrow\tilde{h}_{1}$ is Lipschitz in $h$ in the space $\mathcal{Z}%
_{\frac{3+\lambda}{2}}^{\sigma;\frac{1}{2}}\left(  T\right)  .$ Moreover, due
to Lemma \ref{Le3}, the map $h\rightarrow\left(  {\mathcal{L}}_{f_{0}%
}-L\right)  \left(  \tilde{h}_{1}\right)  $ from $\mathcal{Z}_{\bar{p}%
}^{\sigma;\frac{1}{2}}\left(  T\right)  $ to $Y_{\frac{3}{2},\frac{3+\lambda
}{2}+\bar{\delta}}^{\sigma}$ is Lipschitz.$.$ Therefore the map in
(\ref{G3E1a}) has the Lipschitz dependence stated in Proposition
\ref{Propositionh_1}. The Lipschitz property for the map in (\ref{G3E1b}) is
again a consequence of Proposition \ref{LemmaQuadLip}, the linearity of the
problem under consideration and Proposition \ref{PropAsympt}.

It only remains to check that the Lipschitz constant of the maps
(\ref{G3E1a}), (\ref{G3E1b}) can be made small if $T\leq T_{0}$ and $T_{0}$ is
small enough. Indeed, given two couples $\left(  h^{\left(  1\right)
},\Lambda^{\left(  1\right)  }\right)  ,\ \left(  h^{\left(  2\right)
},\Lambda^{\left(  2\right)  }\right)  $ satisfying the hypothesis of the
Proposition, let us denote as $F^{\left(  1\right)  },\ F^{\left(  2\right)
}$, $\tilde{h}_{1}^{\left(  1\right)  },\ \tilde{h}_{1}^{\left(  2\right)  }$
and $\mathcal{W}^{\left(  1\right)  },\ \mathcal{W}^{\left(  2\right)  }$ the
corresponding functions $F,\ \tilde{h}_{1},\ \mathcal{W}$ respectively. The
stated Lipschitz properties yield:%
\begin{align}
&  \left\Vert F^{\left(  1\right)  }-F^{\left(  2\right)  }\right\Vert
_{Y_{\frac{3}{2},\left(  2+\bar{\delta}\right)  }^{\sigma}}+\left\Vert \left(
{\mathcal{L}}_{f_{0}}-L\right)  \left(  \tilde{h}_{1}^{\left(  1\right)
}-\tilde{h}_{1}^{\left(  2\right)  }\right)  \right\Vert _{Y_{\frac{3}%
{2},\left(  2+\bar{\delta}\right)  }^{\sigma}}\label{H1E9}\\
&  \leq C\left(  \left\Vert \Lambda^{\left(  1\right)  }-\Lambda^{\left(
2\right)  }\right\Vert _{C\left[  0,T\right]  }+\left\Vert h^{\left(
1\right)  }-h^{\left(  2\right)  }\right\Vert _{\mathcal{Z}_{\bar{p}}%
^{\sigma;\frac{1}{2}}\left(  T\right)  }\right) \nonumber
\end{align}

Using the inequality (\ref{H1E9}) combined with (\ref{G2E5}) and (\ref{Z2E1d})
we then obtain that the difference $\mathcal{W}^{\left(  1\right)
}-\mathcal{W}^{\left(  2\right)  }$ can be estimated as:%
\[
\left\vert \mathcal{W}^{\left(  1\right)  }\left(  \tau\right)  -\mathcal{W}%
^{\left(  2\right)  }\left(  \tau\right)  \right\vert \leq CT^{\frac
{2\bar{\delta}}{\lambda-1}}\left(  \left\Vert \Lambda^{\left(  1\right)
}-\Lambda^{\left(  2\right)  }\right\Vert _{C\left[  0,T\right]  }+\left\Vert
h^{\left(  1\right)  }-h^{\left(  2\right)  }\right\Vert _{\mathcal{Z}%
_{\bar{p}}^{\sigma;\frac{1}{2}}\left(  T\right)  }\right)
\]
for $0\leq\tau\leq T.$

The fact that the Lipschitz constant for the map (\ref{G3E1b}) can be made
small for small $T$ is just a consequence of (\ref{S5E6b}) in Lemma \ref{Le1}
and (\ref{M4E10}) in Lemma \ref{Le2} if $\rho_{0}$ in the statement of
Proposition \ref{Propositionh_1} is sufficiently small.
\end{proof}

\bigskip

\subsection{Construction of the function $\tilde{h}_{2}.$\label{Functionh_2}}

\bigskip

It would be natural to construct $\tilde{h}_{2}$ as the solution of the problem:%

\begin{equation}
\tilde{h}_{2,\tau}={\mathcal{L}}_{f_{0}}\left[  \tilde{h}_{2}\right]
-\Lambda_{\tau}f_{0}\left(  x\right)  \ \ \ \ ,\ \ \ \tilde{h}_{2}\left(
0,x\right)  =0\ \ ,\ \ x>0 \label{Z2E7}%
\end{equation}

However, since it is more convenient from the technical point of view to avoid
using the derivative $\Lambda_{\tau}$ we will use an alternative procedure
that we describe shortly here. More precisely, we will obtain a solution of
the initial value problem:%
\begin{equation}
\psi_{\tau}={\mathcal{L}}_{f_{0}}\left[  \psi\right]  \ \ \ ,\ \ \ \ \psi
\left(  0,x\right)  =f_{0}\left(  x\right)  \ \ ,\ \ x>0 \label{Z2E5}%
\end{equation}

In order to solve this problem we define the change of variables:%
\begin{equation}
\psi\left(  \tau,x\right)  =f_{0}\left(  x\right)  +\zeta\left(
\tau,x\right)  \label{U3E8}%
\end{equation}

The function $\zeta$ then solves:%
\begin{equation}
\zeta_{\tau}={\mathcal{L}}_{f_{0}}\left[  \zeta\right]  +{\mathcal{L}}_{f_{0}%
}\left[  f_{0}\right]  \ \ \ ,\ \ \ \ \zeta\left(  0,x\right)
=0\ \ ,\ \ x>0\ \label{Z2E6}%
\end{equation}

This equation can be solved, assuming (\ref{Z1E1a})-(\ref{Z1E2b}) using
Theorem \ref{ExistenceLredonda}.

Variation of constants formula then suggests that $\tilde{h}_{2}$, solution of
(\ref{Z2E7}) is given by:%
\begin{equation}
\tilde{h}_{2}\left(  \tau,x\right)  =-\int_{0}^{\tau}\psi\left(
\tau-s,x\right)  \Lambda_{\tau}\left(  s\right)  ds \label{Z2E8}%
\end{equation}
and assuming that $\psi$ is differentiable in time we would obtain:%
\begin{equation}
\tilde{h}_{2}\left(  \tau,x\right)  =-f_{0}\left(  x\right)  \Lambda\left(
\tau\right)  +\psi\left(  \tau,x\right)  -\int_{0}^{\tau}\frac{\partial\psi
}{\partial\tau}\left(  \tau-s,x\right)  \Lambda\left(  s\right)  ds
\label{Z2E9}%
\end{equation}

This representation formula avoids using $\Lambda_{\tau}$. However, it
requires to prove that $\frac{\partial\psi}{\partial\tau}$ is well defined. We
now prove the properties of $\psi$ required to give a precise meaning to
(\ref{Z2E9}).

\bigskip

\begin{proposition}
\label{Proposition_psi}Suppose that $f_{0}$ satisfies (\ref{Z1E1a}%
)-(\ref{Z1E2b}). There exist a function $\psi\in\mathcal{Z}_{\frac{3+\lambda
}{2}}^{\sigma;\frac{1}{2}}\left(  T\right)  $ defined by means of
(\ref{Z2E5}). We have:%
\begin{equation}
\psi\left(  \tau,x\right)  =a\left(  \tau\right)  \xi\left(  x\right)
x^{-\frac{3+\lambda}{2}}+r_{2}\left(  \tau,x\right)  \label{G3E2}%
\end{equation}
where $\xi\left(  \cdot\right)  $ is the cutoff in (\ref{Z1E2a}) and where:
\begin{equation}
a\left(  0\right)  =1\ \ \ \ ,\ \ \left\Vert r_{2}\right\Vert _{\mathcal{Z}%
_{\bar{p}}^{\sigma;\frac{1}{2}}\left(  T\right)  }\leq C \label{G3E2a}%
\end{equation}

\end{proposition}

Moreover:%
\begin{equation}
\left\vert \frac{da}{d\tau}\right\vert \leq C\ \ \ \ ,\ \ \ \ \left\Vert
\frac{\partial r_{2}}{\partial\tau}\right\Vert _{\mathcal{Z}_{\bar{p}}%
^{\sigma;\frac{1}{2}}\left(  T\right)  }\leq C \label{G3E3}%
\end{equation}

The proof of this result is based on the following Lemma:

\begin{lemma}
\label{Lemmaf_0}Suppose that $f_{0}$ satisfies (\ref{Z1E1a})-(\ref{Z1E2b}).
There exists a constant $C$ such that, for any $0<T\leq1:$%
\begin{equation}
\left\Vert {\mathcal{L}}_{f_{0}}\left[  f_{0}\right]  \right\Vert
_{Y_{\frac{3}{2},2+\delta}^{\sigma}\left(  T\right)  }+\left\Vert
{\mathcal{L}}_{f_{0}}\left[  {\mathcal{L}}_{f_{0}}\left[  f_{0}\right]
\right]  \right\Vert _{Y_{\frac{3}{2},2+\delta}^{\sigma}\left(  T\right)
}\leq C \label{U3E5}%
\end{equation}
where $\delta>0$ is as in (\ref{Z1E1a})-(\ref{Z1E2b}).
\end{lemma}

\begin{proof}
[Proof of Lemma \ref{Lemmaf_0}]Using (\ref{Z1E1a})-(\ref{Z1E2b}) we obtain the
asymptotics:%
\begin{equation}
{\mathcal{L}}_{f_{0}}\left[  f_{0}\right]  =Kx^{-\frac{3+\lambda}{2}}%
\xi\left(  x\right)  +w_{0,R}\left(  x\right)  \ \ \label{U3E9}%
\end{equation}
with a remainder $w_{0,R}\left(  x\right)  $ that can be estimated, together
with its derivatives as $x^{-\left(  \frac{3+\lambda}{2}+\delta\right)  }$ as
$x\rightarrow\infty.$ The main idea to keep in mind is that the operator
${\mathcal{L}}_{f_{0}}$ acting on power laws $x^{-p}$ amounts to multiply then
by $C_{p}x^{r}.$ The constant $C_{p}$ vanishes if $p=\frac{3+\lambda}{2}.$ The
estimate (\ref{U3E5}) will be then proved multiplying by the cutoff
$\eta\left(  \tau\right)  $ and taking the operator ${\mathcal{L}}_{f_{0}}$.
Since the leading power law in (\ref{U3E9}) is $x^{-\frac{3+\lambda}{2}}$ the
action of the operator ${\mathcal{L}}_{f_{0}}$ will cancel the first order and
only a remainder behaving like $x^{-\left(  2+\delta\right)  }$ will be left,
with $\delta$ as in (\ref{Z1E1a})-(\ref{Z1E2b}).

We now describe the details. The operator ${\mathcal{L}}_{f_{0}}$ is defined
in (\ref{S2E3}). We then have, rearranging the integral terms:%
\begin{equation}
\frac{1}{2}{\mathcal{L}}_{f_{0}}\left[  f_{0}\right]  \left(  x\right)
=Q\left[  f_{0}\right]  \left(  x\right)  \label{M7E1a}%
\end{equation}

Using (\ref{Z1E1a}), (\ref{Z1E2}):%
\begin{equation}
{\mathcal{L}}_{f_{0}}\left[  f_{0}\right]  \left(  x\right)  ={\mathcal{L}%
}_{f_{1;2}}\left[  f_{1;2}\right]  \left(  x\right)  +2{\mathcal{L}}_{f_{1;2}%
}\left[  f_{3}\right]  +{\mathcal{L}}_{f_{3}}\left[  f_{3}\right]  \left(
x\right)  \label{M7E2}%
\end{equation}

We can estimate ${\mathcal{L}}_{f_{1;2}}\left[  f_{3}\right]  ,$ as well as
its derivative using (\ref{Z1E1a}), (\ref{Z1E2}), (\ref{Z1E2b}) as well as the
fact that $1<\lambda<2$ and $0<\delta<r:$%
\begin{equation}
\left\vert {\mathcal{L}}_{f_{1;2}}\left[  f_{3}\right]  \right\vert +\left(
1+x\right)  \left\vert \frac{d}{dx}\left(  {\mathcal{L}}_{f_{1;2}}\left[
f_{3}\right]  \right)  \right\vert +\left(  1+x\right)  ^{2}\left\vert
\frac{d^{2}}{dx^{2}}\left(  {\mathcal{L}}_{f_{1;2}}\left[  f_{3}\right]
\right)  \right\vert \leq\frac{C}{\left(  1+x\right)  ^{2+r+\delta}%
}\ \ \ ,\ \ \ x>0 \label{M7E3}%
\end{equation}

The term ${\mathcal{L}}_{f_{3}}\left[  f_{3}\right]  \left(  x\right)  $ can
be estimated similarly:%
\begin{equation}
\left\vert \frac{1}{2}{\mathcal{L}}_{f_{3}}\left[  f_{3}\right]  \right\vert
+\left(  1+x\right)  \left\vert \frac{1}{2}\frac{d}{dx}\left(  {\mathcal{L}%
}_{f_{3}}\left[  f_{3}\right]  \right)  \right\vert +\left(  1+x\right)
^{2}\left\vert \frac{d^{2}}{dx^{2}}\left(  {\mathcal{L}}_{f_{3}}\left[
f_{3}\right]  \right)  \right\vert \leq\frac{C}{\left(  1+x\right)
^{\frac{3+\lambda}{2}+r}}\leq\frac{C}{\left(  1+x\right)  ^{2+r+\delta}}
\label{M7E4}%
\end{equation}
for $\ x>0.$ Therefore, it only remains to estimate ${\mathcal{L}}_{f_{1;2}%
}\left[  f_{1;2}\right]  $ in (\ref{M7E2}).

We then only need to approximate the term ${\mathcal{L}}_{f_{1;2}}\left[
f_{1;2}\right]  $ that might be rewritten as:%
\begin{equation}
{\mathcal{L}}_{f_{1;2}}\left[  f_{1;2}\right]  \left(  x\right)
={\mathcal{L}}_{f_{1}}\left[  f_{1}\right]  \left(  x\right)  +2{\mathcal{L}%
}_{f_{1}}\left[  f_{2}\right]  \left(  x\right)  +{\mathcal{L}}_{f_{2}}\left[
f_{2}\right]  \left(  x\right)  \label{Num}%
\end{equation}

We then have, using $r>\delta:$%
\begin{equation}
\left\vert {\mathcal{L}}_{f_{2}}\left[  f_{2}\right]  \right\vert +\left(
1+x\right)  \left\vert \frac{d}{dx}\left(  {\mathcal{L}}_{f_{2}}\left[
f_{2}\right]  \right)  \right\vert +\left(  1+x\right)  ^{2}\left\vert
\frac{d^{2}}{dx^{2}}\left(  {\mathcal{L}}_{f_{2}}\left[  f_{2}\right]
\right)  \right\vert \leq\frac{C}{\left(  1+x\right)  ^{2+r+\delta}%
}\ \ ,\ \ x>0 \label{M7E7}%
\end{equation}

We need to obtain precise asymptotics of the terms ${\mathcal{L}}_{f_{1}%
}\left[  f_{1}\right]  \left(  x\right)  ,\ 2{\mathcal{L}}_{f_{1}}\left[
f_{2}\right]  \left(  x\right)  $ in order to obtain the leading order term in
(\ref{U3E9}). Let us write $\bar{f}_{1}\left(  x\right)  =x^{-\frac{3+\lambda
}{2}}.$ Notice that:%
\begin{equation}
{\mathcal{L}}_{\bar{f}_{1}}\left[  \bar{f}_{1}\right]  =0 \label{M7E1}%
\end{equation}

Moreover, we have the following identity:%
\[
x^{\lambda/2}\,f_{1}(x)\int_{\frac{x}{2}}^{\infty}y^{\lambda/2}f_{1}%
(y)dy=x^{\lambda/2}\,\bar{f}_{1}(x)\int_{\frac{x}{2}}^{\infty}y^{\lambda
/2}\bar{f}_{1}(y)dy\ \ ,\ \ x\geq1
\]

Using (\ref{M7E1}) we then obtain for $x>2:$%
\[
{\mathcal{L}}_{f_{1}}\left[  f_{1}\right]  =\int_{0}^{2}y^{\lambda/2}\bar
{f}_{1}\left(  y\right)  \left[  \xi(y)-1\right]  \left[  (x-y)^{\lambda
/2}\bar{f}_{1}(x-y)-x^{\lambda/2}\,\bar{f}_{1}(x)\right]  dy
\]

Taylor's expansion, as well as the fact that $\delta<\frac{\lambda-2}{2}$
implies:%
\begin{equation}
\left\vert {\mathcal{L}}_{f_{1}}\left[  f_{1}\right]  \right\vert +\left(
1+x\right)  \left\vert \frac{d}{dx}\left(  {\mathcal{L}}_{f_{1}}\left[
f_{1}\right]  \right)  \right\vert +\left(  1+x\right)  ^{2}\left\vert
\frac{d^{2}}{dx^{2}}\left(  {\mathcal{L}}_{f_{1}}\left[  f_{1}\right]
\right)  \right\vert \leq\frac{C}{\left(  1+x\right)  ^{\frac{3+\lambda}%
{2}+\delta}}\ \ ,\ \ x>0 \label{M7E8}%
\end{equation}
where we use the fact that ${\mathcal{L}}_{f_{1}}\left[  f_{1}\right]  $ and
its derivatives are trivially bounded for $x$ bounded as it might be seen
using directly using the definition of ${\mathcal{L}}_{f_{1}}\left[
f_{1}\right]  $.

It remains to estimate the term $2{\mathcal{L}}_{f_{1}}\left[  f_{2}\right]
\left(  x\right)  $ in (\ref{Num}). Using the definition of ${\mathcal{L}%
}_{f_{1}}\left[  f_{2}\right]  \left(  x\right)  $ we obtain%
\begin{equation}
2{\mathcal{L}}_{f_{1}}\left[  f_{2}\right]  \left(  x\right)  =\mathcal{H}%
_{1}\left(  f_{1},f_{2}\right)  \left(  x\right)  +\mathcal{H}_{2}\left(
f_{1},f_{2}\right)  \left(  x\right)  \label{H5E4}%
\end{equation}
where:%
\begin{align*}
\mathcal{H}_{1}\left(  f_{1},f_{2}\right)  \left(  x\right)   &  =\int
_{0}^{\frac{x}{2}}y^{\lambda/2}f_{1}(y)\left[  (x-y)^{\lambda/2}%
f_{2}(x-y)-x^{\lambda/2}\,f_{2}(x)\right]  dy+\\
&  +\int_{0}^{\frac{x}{2}}y^{\lambda/2}f_{2}(y)\left[  (x-y)^{\lambda/2}%
f_{1}(x-y)-x^{\lambda/2}\,f_{1}(x)\right]  dy
\end{align*}%
\[
\mathcal{H}_{2}\left(  f_{1},f_{2}\right)  \left(  x\right)  =x^{\lambda
/2}\,f_{1}(x)\int_{\frac{x}{2}}^{\infty}y^{\lambda/2}f_{2}(y)dy+x^{\lambda
/2}\,f_{2}(x)\int_{\frac{x}{2}}^{\infty}y^{\lambda/2}f_{1}(y)dy
\]
The term $\mathcal{H}_{2}\left(  f_{1},f_{2}\right)  \left(  x\right)  $ can
be explicitly computed for large values of $x:$%
\begin{equation}
\mathcal{H}_{2}\left(  f_{1},f_{2}\right)  \left(  x\right)  =K_{1}%
x^{-\frac{3+\lambda}{2}}\ \ ,\ \ K_{1}\in\mathbb{R}\text{\ and }x>1
\label{H5E3}%
\end{equation}

In order to approximate $\mathcal{H}_{1}\left(  f_{1},f_{2}\right)  \left(
x\right)  $ we define a function $\bar{f}_{2}(x)=\frac{a_{1}}{x^{\frac
{3+\lambda}{2}+r}}.$ We then have:%
\begin{equation}
\mathcal{H}_{1}\left(  \bar{f}_{1},\bar{f}_{2}\right)  \left(  x\right)
=K_{2}x^{-\frac{3+\lambda}{2}}\ \ ,\ \ K_{1}\in\mathbb{R}\text{\ and }x>0
\label{H5E2}%
\end{equation}

On the other hand:%
\begin{align*}
&  \mathcal{H}_{1}\left(  f_{1},f_{2}\right)  \left(  x\right)  -\mathcal{H}%
_{1}\left(  \bar{f}_{1},\bar{f}_{2}\right)  \left(  x\right) \\
&  =\int_{0}^{2}y^{\lambda/2}\bar{f}_{1}(y)\left[  \xi(y)-1\right]  \left[
(x-y)^{\lambda/2}\bar{f}_{2}(x-y)-x^{\lambda/2}\bar{f}_{2}(x)\right]  dy+\\
&  +\int_{0}^{2}y^{\lambda/2}\bar{f}_{2}(y)\left[  \xi(y)-1\right]  \left[
(x-y)^{\lambda/2}\bar{f}_{1}(x-y)-x^{\lambda/2}\bar{f}_{1}(x)\right]  dy
\end{align*}

Taylor's expansion, as well as the fact that $\delta<\frac{2-\lambda}{2},$
yields:%
\begin{align}
&  \left\vert \mathcal{H}_{1}\left(  f_{1},f_{2}\right)  \left(  x\right)
-\mathcal{H}_{1}\left(  \bar{f}_{1},\bar{f}_{2}\right)  \left(  x\right)
\right\vert +\left(  1+x\right)  \left\vert \frac{d}{dx}\left(  \mathcal{H}%
_{1}\left(  f_{1},f_{2}\right)  \left(  x\right)  -\mathcal{H}_{1}\left(
\bar{f}_{1},\bar{f}_{2}\right)  \left(  x\right)  \right)  \right\vert
\label{M5E1}\\
&  \leq\frac{C}{x^{\frac{5}{2}}}\leq\frac{C}{x^{\frac{3+\lambda}{2}+\delta}%
}\ \ ,\ \ x\geq1\nonumber
\end{align}

The boundedness of ${\mathcal{L}}_{f_{1}}\left[  f_{2}\right]  \left(
x\right)  $ and its derivatives combined with (\ref{H5E4})-(\ref{M5E1})
yields:%
\begin{equation}
\left\vert 2{\mathcal{L}}_{f_{1}}\left[  f_{2}\right]  \left(  x\right)
-Kx^{-\frac{3+\lambda}{2}}\right\vert \leq\frac{C}{1+x^{\frac{3+\lambda}%
{2}+\delta}}\ \ ,\ \ x>0 \label{M7E9}%
\end{equation}

Combining (\ref{M7E2})-(\ref{M7E8}), (\ref{M7E9}) we obtain (\ref{U3E9})
where:%
\begin{equation}
\sum_{k=0}^{3}\left(  1+x\right)  ^{k}\left\vert \frac{d^{k}w_{0,R}}{dx^{k}%
}\right\vert \leq\frac{C}{1+x^{\frac{3+\lambda}{2}+\delta}}\ \ ,\ \ x>0
\label{M8E1}%
\end{equation}

Applying ${\mathcal{L}}_{f_{0}}\left[  {\mathcal{\cdot}}\right]  $ on both
sides of (\ref{U3E9}), using ${\mathcal{L}}_{f_{0}}\left[  x^{-\frac
{3+\lambda}{2}}\right]  =0,$ and arguing as in the proof of (\ref{M7E8}) we
obtain $\sum_{k=0}^{2}\left(  1+x\right)  ^{k}\left\vert \frac{d^{k}}{dx^{k}%
}\left(  {\mathcal{L}}_{f_{0}}\left[  Kx^{-\frac{3+\lambda}{2}}\xi\left(
x\right)  \right]  \right)  \right\vert \leq\frac{C}{1+x^{2+\delta}}.$ On the
other hand, the action of the operator ${\mathcal{L}}_{f_{0}}$ over functions
satistying (\ref{M8E1}) amounts to multiplying by $x^{r}$ for large values of
$x.$ Therefore $\sum_{k=0}^{2}\left(  1+x\right)  ^{k}\left\vert \frac{d^{k}%
}{dx^{k}}\left(  {\mathcal{L}}_{f_{0}}\left[  w_{0,R}\right]  \right)
\right\vert \leq\frac{C}{1+x^{2+\delta}}$ and the result follows.
\end{proof}

\begin{proof}
[Proof of Proposition \ref{Proposition_psi}]Due to (\ref{U3E5}) in Lemma
\ref{Lemmaf_0} ${\mathcal{L}}_{f_{0}}\left[  f_{0}\right]  $ is bounded in the
space $Y_{3/2,2+\delta}^{\sigma}\left(  T\right)  .$ Therefore (\ref{Z2E6})
can be solved using the results in \cite{EV2}. We obtain in this way a
solution $\zeta\in\mathcal{Z}_{\frac{3+\lambda}{2}}^{\sigma;\frac{1}{2}%
}\left(  T\right)  .$ Then, $\psi$ can be obtained by means of (\ref{U3E8}).
Therefore expansion (\ref{G3E2}), (\ref{G3E2a}) are just a consequence of
Proposition \ref{PropAsympt}.

It only remains to obtain estimates for the derivatives on time of the
functions $a,\ r_{2}.$ Formal differentiation of (\ref{Z2E6}) suggests that
$w=\frac{\partial\zeta}{\partial\tau}$ satisfies the following initial value
problem:%
\begin{equation}
\left(  w\right)  _{\tau}={\mathcal{L}}_{f_{0}}\left[  w\right]
\ \ \ ,\ \ \ \ w\left(  0,x\right)  ={\mathcal{L}}_{f_{0}}\left[
f_{0}\right]  \label{U3E1}%
\end{equation}

Actually we can use the results in \cite{EV2} to construct a solution of
(\ref{U3E1}) as follows. We define a function $W\left(  \tau,x\right)  $ by
means of:%
\begin{equation}
w\left(  \tau,x\right)  ={\mathcal{L}}_{f_{0}}\left[  f_{0}\right]  +W\left(
\tau,x\right)  \label{U3E2a}%
\end{equation}

Then, $w$ solves (\ref{U3E1}) iff $W$ solves:%
\begin{equation}
W_{\tau}={\mathcal{L}}_{f_{0}}\left[  W\right]  +{\mathcal{L}}_{f_{0}}\left[
{\mathcal{L}}_{f_{0}}\left[  f_{0}\right]  \right]  \ \ ,\ \ W\left(
0,x\right)  =0 \label{U3E6}%
\end{equation}

In order to be able to solve the problem (\ref{U3E1}) we use the hypothesis
(\ref{Z1E1a})-(\ref{Z1E2b}). Due to Lemma \ref{Lemmaf_0} we have that
$\left\Vert {\mathcal{L}}_{f_{0}}\left[  {\mathcal{L}}_{f_{0}}\left[
f_{0}\right]  \right]  \right\Vert _{Y_{\frac{3}{2},2+\delta}^{\sigma}\left(
T\right)  }$ bounded. Therefore, we can apply the results in \cite{EV2} (cf.
Theorem \ref{ExistenceLredonda}) to obtain a unique solution $W$ of
(\ref{U3E6}) satisfying $\left\Vert W\right\Vert _{\mathcal{Z}_{\frac
{3+\lambda}{2}}^{\sigma;\frac{1}{2}}\left(  T\right)  }\leq C.$ The function
$\psi,$ solution of (\ref{Z2E5}) can be obtained, using also (\ref{U3E8}) as:%
\begin{equation}
\psi\left(  \tau,x\right)  =f_{0}\left(  x\right)  +{\mathcal{L}}_{f_{0}%
}\left[  f_{0}\right]  \tau+\int_{0}^{\tau}W\left(  s,x\right)  ds
\label{U3E7a}%
\end{equation}

Using (\ref{G2E2}), (\ref{G2E6}) in Proposition \ref{PropAsympt} we obtain:%
\begin{equation}
W\left(  \tau,x\right)  =\mathcal{W}\left(  \tau\right)  x^{-\frac{3+\lambda
}{2}}\xi\left(  x\right)  +W_{R}\left(  \tau,x\right)  \label{M6E0}%
\end{equation}
with
\begin{equation}
\left\vert \mathcal{W}\left(  \tau\right)  \right\vert \leq C\ \ ,\ \ 0\leq
\tau\leq T\ \ \ ,\ \ \left\Vert W_{R}\left(  \tau,x\right)  \right\Vert
_{\mathcal{Z}_{\bar{p}}^{\sigma;\frac{1}{2}}\left(  T\right)  }\leq C
\label{M6E1}%
\end{equation}

Then $\psi\left(  \tau,x\right)  =a\left(  \tau\right)  x^{-\frac{3+\lambda
}{2}}\xi\left(  x\right)  +r_{2}\left(  \tau,x\right)  $ where:%
\[
a\left(  \tau\right)  =1+K\tau+\int_{0}^{\tau}\mathcal{W}\left(  s\right)
ds\ \ ,\ \ r_{2}\left(  \tau,x\right)  =\left[  f_{0}\left(  x\right)
-x^{-\frac{3+\lambda}{2}}\xi\left(  x\right)  \right]  +w_{0,R}\left(
x\right)  \tau+\int_{0}^{\tau}W_{R}\left(  s,x\right)  ds
\]

Using (\ref{M6E1}) we obtain (\ref{G3E3}) and the Proposition follows.
\end{proof}

\begin{remark}
\label{rationalef_0}As indicated in Section \ref{initialdata} the assumption
(\ref{Z1E1a})-(\ref{Z1E2b}) is very strong. However, the argument proving
Proposition \ref{Proposition_psi} shows that the main reason for assuming
(\ref{Z1E1a})-(\ref{Z1E2b}) is to show that $\left\vert \frac{da}{d\tau
}\right\vert ,\ \frac{\partial r_{2}}{\partial\tau}$ are bounded in a suitable
sense. It would be possible to weaken (\ref{Z1E1a}), (\ref{Z1E2}) to some
assumption with the form $f_{0}\left(  x\right)  =D_{1}x^{-\frac{3+\lambda}%
{2}}+O\left(  x^{-\frac{3+\lambda}{2}-\delta}\right)  $ as $x\rightarrow
\infty$ for some $\delta>0$. Making such an assumption the only difference in
the argument proving Proposition \ref{Proposition_psi} would be that the term
${\mathcal{L}}_{f_{0}}\left[  f_{0}\right]  $ in (\ref{U3E1}) would behave
like $O\left(  x^{-\left(  2+\delta\right)  }\right)  $ instead of $O\left(
x^{-\frac{3+\lambda}{2}}\right)  $ as $x\rightarrow\infty$ Unfortunately the
well-posedness theory developed in \cite{EV2} cannot cover such weakest rate
of decay at infinity. The expected asymptotics for $w\left(  \tau,x\right)  $
as $x\rightarrow\infty$ for small $\tau$ would have the form $w\left(
\tau,x\right)  \sim\frac{C}{\tau^{1-\frac{2\delta}{\lambda-1}}}x^{-\frac
{3+\lambda}{2}}\ $as$\ \ x\rightarrow\infty,\ \ \tau\rightarrow0.$ This type
of asymptotics has been obtained in \cite{EMV1}, \cite{EMV2} for an different
equation, namely the Uehling-Uhlenbeck equation. Unfortunately since the
well-posedness theory of classical solutions for the coagulation equation is
more difficult, we have preferred not to consider such a case, at the price of
assuming stronger regularity assumptions near the singular point. Nevertheless
it would be an interesting question to prove analogous regularizing results in time.
\end{remark}

With the previous construction we can define the function $\tilde{h}_{2}$ as follows.

\begin{definition}
\label{h2}For any $f_{0}$, $\Lambda$ satisfying the assumptions in Proposition
\ref{Propositionh_1} we define $\tilde{h}_{2}$ by means of%
\begin{equation}
\tilde{h}_{2}\left(  \tau,x\right)  =-f_{0}\left(  x\right)  \Lambda\left(
\tau\right)  +\psi\left(  \tau,x\right)  -\int_{0}^{\tau}w\left(
\tau-s,x\right)  \Lambda\left(  s\right)  ds \label{P1E1}%
\end{equation}
where $\psi$ is as in (\ref{U3E7a}) and $w$ is as in (\ref{U3E2a}).
\end{definition}

\begin{remark}
The rationale behind Definition \ref{h2} is the following. Assuming smoothness
we obtain, differentiating (\ref{P1E1}):%
\[
\left(  \tilde{h}_{2}\right)  _{\tau}=-f_{0}\left(  x\right)  \Lambda_{\tau
}\left(  \tau\right)  +\frac{\partial\psi}{\partial\tau}\left(  \tau,x\right)
-\Lambda\left(  \tau\right)  {\mathcal{L}}_{f_{0}}\left[  f_{0}\right]
-\int_{0}^{\tau}{\mathcal{L}}_{f_{0}}\left[  w\right]  \left(  \tau
-s,x\right)  \Lambda\left(  s\right)  ds
\]
where we have used (\ref{U3E1}). Exchanging the order of the integral in time
and ${\mathcal{L}}_{f_{0}}$ and using again (\ref{P1E1}) we obtain, after some
cancellations $\tilde{h}_{2,\tau}=-f_{0}\Lambda_{\tau}\left(  \tau\right)
+\frac{\partial\psi}{\partial\tau}+{\mathcal{L}}_{f_{0}}\left(  \tilde{h}%
_{2}\right)  -{\mathcal{L}}_{f_{0}}\left[  \psi\right]  .$ Using then
(\ref{Z2E5}) we obtain that $\tilde{h}_{2}$ would solve (\ref{Z2E7}).
\end{remark}

The asymptotics of the function $\tilde{h}_{2}$ as $x\rightarrow\infty$ can be
derived using the corresponding results for the functions $\psi,\ w$ in
Proposition \ref{Proposition_psi}.

\begin{lemma}
\label{Ash2}For any $\Lambda\in C\left[  0,T\right]  $ satisfying the
assumptions in Proposition \ref{Propositionh_1} we have:%
\[
\tilde{h}_{2}\left(  \tau,x\right)  =\mathcal{K}\left[  \Lambda\right]
\left(  \tau\right)  \xi\left(  x\right)  x^{-\frac{3+\lambda}{2}}+\tilde
{h}_{2,R}\left(  \tau,x;\Lambda\right)
\]
where%
\[
\mathcal{K}\left[  \Lambda\right]  \left(  \tau\right)  =-\Lambda\left(
\tau\right)  +a\left(  \tau\right)  -\int_{0}^{\tau}\mathcal{W}\left(
\tau-s\right)  \Lambda\left(  s\right)  ds
\]
with $a$ as in Proposition \ref{Proposition_psi},$\ \frac{da}{d\tau
}=\mathcal{W}$ and $\tilde{h}_{2,R}\in\mathcal{Z}_{\bar{p}}^{\sigma;\frac
{1}{2}}\left(  T\right)  .$ Moreover, the map:%
\begin{equation}
C\left[  0,T\right]  \rightarrow\mathcal{Z}_{\bar{p}}^{\sigma;\frac{1}{2}%
}\left(  T\right)  \ \ :\ \ \Lambda\rightarrow\tilde{h}_{2} \label{P1E3}%
\end{equation}
\ is Lipschitz if $T\leq T_{0},$ with $T_{0}$ sufficiently small.
\end{lemma}

\begin{proof}
It is just a consequence of the definition of $\tilde{h}_{2}$ in (\ref{P1E1})
and Proposition \ref{Proposition_psi}.
\end{proof}

\subsection{Setting of the fixed point argument. Solution of an integral
equation.}

Given $h\in\mathcal{Z}_{\bar{p}}^{\sigma;\frac{1}{2}}\left(  T\right)  $ and
$\Lambda\in C\left[  0,T\right]  $ as in Proposition \ref{Propositionh_1} we
can define a map $\left(  h,\Lambda\right)  \rightarrow\tilde{h}$ with
$\tilde{h}\left(  t,x\right)  =\tilde{h}_{1}\left(  t,x\right)  +\tilde{h}%
_{2}\left(  t,x\right)  $ where $\tilde{h}_{1}$ is as in Proposition
\ref{Propositionh_1} and $\tilde{h}_{2}$ as in (\ref{P1E1}).

We now select, for any given $h,$ the function $\Lambda$ in order to have:
\begin{equation}
\lim_{x\rightarrow\infty}\left(  x^{\frac{3+\lambda}{2}}\tilde{h}\left(
t,x\right)  \right)  =0 \label{P1E2}%
\end{equation}

Due to (\ref{G3E1}) and since $r_{1}\left(  \cdot,\cdot;h\right)
\in\mathcal{Z}_{\bar{p}}^{\sigma;\frac{1}{2}}\left(  T\right)  ,$ as well as
Lemma \ref{Ash2} it follows that (\ref{P1E2}) holds if $\Lambda$ solves:%
\begin{equation}
\Lambda\left(  \tau\right)  =\mathcal{G}\left[  \tau;h,\Lambda\right]
+a\left(  \tau\right)  -\int_{0}^{\tau}\mathcal{W}\left(  \tau-s,x\right)
\Lambda\left(  s\right)  ds \label{Z4E1}%
\end{equation}

We first show that for $h\in\mathcal{Z}_{\bar{p}}^{\sigma;\frac{1}{2}}\left(
T\right)  $ we can find $\Lambda=\Lambda\left(  \cdot;h\right)  $ such that
(\ref{Z4E1}) is satisfied.

\begin{lemma}
\label{LemmaIntEqu}There exist $T_{0}>0$ sufficiently small, such that, for
any $h\in\mathcal{Z}_{\bar{p}}^{\sigma;\frac{1}{2}}\left(  T\right)  $
satisfying $\left\Vert h\right\Vert _{\mathcal{Z}_{\bar{p}}^{\sigma;\frac
{1}{2}}\left(  T\right)  }\leq\rho_{0}$ equation (\ref{Z4E1}) has a unique
solution for $0\leq\tau\leq T,$ assuming that $T\leq T_{0}.$ Moreover, this
solution defines a mapping:%
\begin{equation}
\mathcal{Z}_{\bar{p}}^{\sigma;\frac{1}{2}}\left(  T\right)  \rightarrow
C\left[  0,T\right]  :\ \ h\longmapsto\Lambda\left(  \cdot;h\right)
\label{Z4E2}%
\end{equation}
that is contractive.
\end{lemma}

\begin{proof}
The function $\mathcal{W}$ in the integral term (\ref{Z4E1}) is uniformly
bounded due to Proposition \ref{Proposition_psi}. On the other hand, the
function $\mathcal{G}\left[  \tau;h,\Lambda\right]  $ is Lipschitz contractive
in $\Lambda$ if $T_{0}$ is sufficiently small and $\left\Vert h\right\Vert
_{\mathcal{Z}_{\bar{p}}^{\sigma;\frac{1}{2}}\left(  T\right)  }\leq\rho_{0}$
due to Proposition \ref{Propositionh_1}. It then follows from (\ref{Z4E1})
that the mapping (\ref{Z4E2}) is contractive.
\end{proof}

\begin{lemma}
\label{contract}Let us denote as $B_{\rho_{0}}$ the ball of radius $\rho_{0}$
in $\mathcal{Z}_{\bar{p}}^{\sigma;\frac{1}{2}}\left(  T\right)  ,$ with
$\rho_{0}$ as in Proposition \ref{Propositionh_1} and let us consider the
mapping from $B_{\rho_{0}}$ to $B_{\rho_{0}}$ given by $h\rightarrow
\mathcal{T}\left[  h\right]  $ where $\mathcal{T}\left[  h\right]  =\tilde
{h}_{1}+\tilde{h}_{2},$ with $\tilde{h}_{1}$as in Proposition
\ref{Propositionh_1} and $\tilde{h}_{2}$ as in (\ref{P1E1}) and with
$\Lambda=\Lambda\left(  \cdot;h\right)  $ in (\ref{Z2E3}) where $\Lambda
\left(  \cdot;h\right)  $ is chosen as in Lemma \ref{LemmaIntEqu}. Then, there
exists $T_{0}$ such that the mapping $\mathcal{T}$ is contractive in
$B_{\rho_{0}}$ if $T\leq T_{0}.$ In such a case there exists a unique fixed
point of $\mathcal{T}$ in $B_{\rho_{0}}.$
\end{lemma}

\begin{proof}
The definition of $\tilde{h},\ \tilde{h}_{1},\ \tilde{h}_{2}$ combined with
(\ref{Z4E1}) imply:%
\begin{equation}
\mathcal{T}\left[  h\right]  \left(  \tau,x\right)  =\tilde{h}\left(
\tau,x\right)  =r_{1}\left(  \tau,x;h,\Lambda\left(  \cdot;h\right)  \right)
+\tilde{h}_{2,R}\left(  \tau,x;\Lambda\left(  \cdot;h\right)  \right)
\label{Z4E5}%
\end{equation}

Notice that $\mathcal{T}$ transforms $B_{\rho_{0}}$ into $B_{\rho_{0}}$ for
$T\leq T_{0}$ small. Indeed, $r_{1}$ consists of two pieces that are due to
the contributions of the source terms $\frac{Q\left[  h\right]  }%
{\Lambda\left(  \tau\right)  }$ and $\Lambda\left(  \tau\right)  Q\left[
f_{0}\right]  $ in (\ref{Z2E3}) respectively. The norm $\mathcal{Z}_{\bar{p}%
}^{\sigma;\frac{1}{2}}\left(  T\right)  $ of the solution due to the source
term $\frac{Q\left[  h\right]  }{\Lambda\left(  \tau\right)  }$ can be bounded
as $C\rho_{0}^{2}$ due to Proposition \ref{LemmaQuad}. On the other hand in
order to estimate the contribution due to the term $\Lambda\left(
\tau\right)  Q\left[  f_{0}\right]  $ we use (\ref{S5E6b}) in Lemma \ref{Le1}.
Using then Proposition \ref{PropAsympt} it follows that the contribution due
to the source $\Lambda\left(  \tau\right)  Q\left[  f_{0}\right]  $ is smaller
than $\frac{\rho_{0}}{8}$ if $T_{0}$ is small enough.

On the other hand, in order to see that the contribution of the term
$\tilde{h}_{2,R}$ is small for small times, we use the formulas for
$\psi,\ \tilde{h}_{2}.$ Using (\ref{U3E7a}) and (\ref{P1E1}) we obtain:%
\[
\tilde{h}_{2}\left(  \tau,x\right)  =f_{0}\left(  x\right)  \left[
1-\Lambda\left(  \tau\right)  \right]  +{\mathcal{L}}_{f_{0}}\left[
f_{0}\right]  \tau+\int_{0}^{\tau}W\left(  s,x\right)  ds-\int_{0}^{\tau
}w\left(  \tau-s,x\right)  \Lambda\left(  s\right)  ds
\]

We substracts the terms behaving like $\xi\left(  x\right)  x^{-\frac
{3+\lambda}{2}}$ in all the pieces. We then obtain:%
\begin{align*}
\tilde{h}_{2,R}\left(  \tau,x\right)   &  =\left[  f_{0}\left(  x\right)
-\xi\left(  x\right)  x^{-\frac{3+\lambda}{2}}\right]  \left[  1-\Lambda
\left(  \tau\right)  \right]  +\\
&  +w_{0,R}\left(  x\right)  \tau+\int_{0}^{\tau}W_{R}\left(  s,x\right)
ds-\int_{0}^{\tau}\left[  w_{0,R}\left(  x\right)  +W_{R}\left(
\tau-s,x\right)  \right]  \Lambda\left(  s\right)  ds\\
&  \equiv\tilde{h}_{2,R,1}\left(  \tau,x\right)  +\tilde{h}_{2,R,2}\left(
\tau,x\right)  +\tilde{h}_{2,R,3}\left(  \tau,x\right)  +\tilde{h}%
_{2,R,4}\left(  \tau,x\right)
\end{align*}
(cf. (\ref{U3E9}), (\ref{U3E2a}), (\ref{M6E0})). Using (\ref{U3E9}) we obtain
$\left\Vert \tilde{h}_{2,R,2}\right\Vert _{\mathcal{Z}_{\bar{p}}^{\sigma
;\frac{1}{2}}\left(  T\right)  }\leq CT.$ We can estimate $\tilde{h}_{2,R,3}$
and $\tilde{h}_{2,R,4}$ in the space $\mathcal{Z}_{\bar{p}}^{\sigma;\frac
{1}{2}}\left(  T\right)  $ using the fact that these functions are integrals
on time of functions bounded in $\mathcal{Z}_{\bar{p}}^{\sigma;\frac{1}{2}%
}\left(  T\right)  .$ Using Lemma \ref{LemmaInteg} we obtain $\left\Vert
\tilde{h}_{2,R,3}\right\Vert _{\mathcal{Z}_{\bar{p}}^{\sigma;\frac{1}{2}%
}\left(  T\right)  }+\left\Vert \tilde{h}_{2,R,4}\right\Vert _{\mathcal{Z}%
_{\bar{p}}^{\sigma;\frac{1}{2}}\left(  T\right)  }\leq C\sqrt{T}.$ It only
remains to control the term $\tilde{h}_{2,R,1}.$ To this end, we use here the
integral equation (\ref{Z4E1}) that yields:%
\begin{equation}
\left[  1-\Lambda\left(  \tau\right)  \right]  =-\mathcal{G}\left[
\tau;h,\Lambda\right]  +\left(  1-a\left(  \tau\right)  \right)  +\int
_{0}^{\tau}\mathcal{W}\left(  \tau-s,x\right)  \Lambda\left(  s\right)  ds
\label{B3}%
\end{equation}

Due to Proposition \ref{Propositionh_1} and Lemmas \ref{Le1}, \ref{Le2} we can
estimate the contributions to $\mathcal{G}\left[  \tau;h,\Lambda\right]  $
that are due to $\Lambda\left(  \tau\right)  ,\ Q\left[  f_{0}\right]  $ and
$\frac{Q\left[  h\right]  }{\Lambda\left(  \tau\right)  }$ respectively as
$C\max\left\{  \sqrt{T},T^{\frac{2\left(  r-\bar{\delta}\right)  }{\lambda-1}%
}\right\}  $ and $C\rho_{0}^{2}.$ Therefore this contribution can be estimated
by $\frac{\rho_{0}}{8}.$ The second term on the right-hand side of (\ref{B3})
can be estimated using the differentiability of $a$ (cf. (\ref{G3E3})).
Therefore this term can be estimated as $CT.$ On the other hand, the
boundedness of $\mathcal{W}$ (cf. (\ref{M6E1})) provides a similar estimate
for the last term in (\ref{B3}). Therefore, using the regularity of $f_{0}$ we
obtain $\left\Vert \tilde{h}_{2,R,1}\right\Vert _{\mathcal{Z}_{\bar{p}%
}^{\sigma;\frac{1}{2}}\left(  T\right)  }\leq\frac{\rho_{0}}{4}.$ It then
follows that $\mathcal{T}$ transforms $B_{\rho_{0}}$ into $B_{\rho_{0}}$ if
$T_{0}$ is sufficiently small. Combining the contractivity of the map
(\ref{Z4E2}) with the Lipschitz properties of the maps (\ref{G3E1b}),
(\ref{P1E3}) we obtain the contractivity of $\mathcal{T}$ if $T\leq T_{0}$
sufficiently small.
\end{proof}

\begin{proof}
[Proof of Theorem \ref{Th1}]We define $\tilde{f}$ by means of:%
\begin{equation}
\tilde{f}\left(  \tau,x\right)  =\Lambda\left(  \tau\right)  f_{0}\left(
x\right)  +h\left(  \tau,x\right)  \label{Y1E1}%
\end{equation}
where $h\left(  \tau,x\right)  $ is the fixed point associated to the operator
$\mathcal{T}$ obtained in Lemma \ref{contract}. Notice that:%
\[
h\left(  \tau,x\right)  =\mathcal{T}\left(  h\right)  \left(  \tau,x\right)
=r_{1}\left(  \tau,x;h,\Lambda\right)  +\tilde{h}_{2,R}\left(  \tau,x\right)
=\tilde{h}_{1}\left(  \tau,x\right)  +\tilde{h}_{2}\left(  \tau,x\right)
\]
where $\tilde{h}_{1},\ \tilde{h}_{2}$ are as in Proposition
\ref{Propositionh_1} and (\ref{P1E1}) respectively. Using (\ref{Y1E1}) and
(\ref{P1E1}) we obtain:%
\begin{equation}
\tilde{f}\left(  \tau,x\right)  =\tilde{h}_{1}\left(  \tau,x\right)
+\psi\left(  \tau,x\right)  -\int_{0}^{\tau}w\left(  \tau-s,x\right)
\Lambda\left(  s\right)  ds \label{Y1E2}%
\end{equation}
where $\psi,\ w$ are as in (\ref{U3E7a}), (\ref{U3E2a}). The function
$\tilde{f}$ is differentiable with respect to $\tau$ due to Proposition
\ref{Propositionh_1}, Proposition \ref{Proposition_psi} and the continuity of
$\Lambda,$\ and differentiability of $w$ (cf. Lemma \ref{LemmaIntEqu} and
(\ref{U3E2a}), (\ref{U3E6}) respectively). Therefore $\tilde{f}$ solves:%
\begin{equation}
\tilde{f}_{\tau}=\frac{Q\left[  \tilde{f}\right]  }{\Lambda\left(
\tau\right)  } \label{Y1E3}%
\end{equation}
as it can be checked as follows. Differentiating (\ref{Y1E2}), using the fact
that $w$ solves (\ref{U3E1}) and exchanging the integration in time with the
operator $\mathcal{L}_{f_{0}}\left[  \cdot\right]  $ we obtain:%
\[
\tilde{f}_{\tau}=\left(  \tilde{h}_{1}\right)  _{\tau}+\psi_{\tau}%
-\Lambda\left(  \tau\right)  \mathcal{L}_{f_{0}}\left[  f_{0}\right]
-\mathcal{L}_{f_{0}}\left[  \int_{0}^{\tau}w\left(  \tau-s,x\right)
\Lambda\left(  s\right)  ds\right]
\]

Eliminating the integral in the last term by means of (\ref{Y1E2}):%
\[
\tilde{f}_{\tau}=\left[  \left(  \tilde{h}_{1}\right)  _{\tau}-\mathcal{L}%
_{f_{0}}\left[  \tilde{h}_{1}\right]  \right]  +\left[  \psi_{\tau
}-\mathcal{L}_{f_{0}}\left[  \psi\right]  \right]  -\Lambda\left(
\tau\right)  \mathcal{L}_{f_{0}}\left[  f_{0}\right]  +\mathcal{L}_{f_{0}%
}\left[  \tilde{f}\right]
\]

Proposition \ref{Proposition_psi} yields $\psi_{\tau}-\mathcal{L}_{f_{0}%
}\left[  \psi\right]  =0.$ Using also (\ref{Z2E3}) we obtain:%
\[
\tilde{f}_{\tau}=\frac{Q\left[  h\right]  }{\Lambda\left(  \tau\right)
}+\Lambda\left(  \tau\right)  Q\left[  f_{0}\right]  -\Lambda\left(
\tau\right)  \mathcal{L}_{f_{0}}\left[  f_{0}\right]  +\mathcal{L}_{f_{0}%
}\left[  \tilde{f}\right]
\]
where due to (\ref{Y1E1}) $h=\tilde{f}\left(  \tau,x\right)  -\Lambda\left(
\tau\right)  f_{0}\left(  x\right)  .$ Then:%
\[
\tilde{f}_{\tau}=\frac{Q\left[  \tilde{f}\right]  }{\Lambda\left(
\tau\right)  }-\mathcal{L}_{f_{0}}\left[  \tilde{f}\right]  +\Lambda\left(
\tau\right)  Q\left[  f_{0}\right]  +\Lambda\left(  \tau\right)  Q\left[
f_{0}\right]  -\Lambda\left(  \tau\right)  \mathcal{L}_{f_{0}}\left[
f_{0}\right]  +\mathcal{L}_{f_{0}}\left[  \tilde{f}\right]
\]

Using that $\mathcal{L}_{f_{0}}\left[  f_{0}\right]  =2Q\left[  f_{0}\right]
$ (cf. (\ref{M7E1a})) we obtain that $\tilde{f}$ solves (\ref{Y1E3}). Using
the time scale $t$ given by means of (\ref{time}) we deduce that $f\left(
t,x\right)  =\tilde{f}\left(  \tau,x\right)  $ solves (\ref{S1E1}),
(\ref{S1E2}). Using (\ref{Z2E4}), (\ref{Z2E5}) and (\ref{Y1E2}) we have that
$f$ satisfies (\ref{S1E3}). This concludes the Proof of the existence of the
sought-for solution.

We prove uniqueness in the class of solutions stated in Theorem \ref{Th1} as
follows. Suppose that we have two solutions $f^{\alpha},\ f^{\beta}$ of
(\ref{S1E1})-(\ref{S1E3}) such that $f^{\alpha}=\lambda^{\alpha}\left(
t\right)  f_{0}\left(  x\right)  +h^{\alpha}\ \ ,\ \ f^{\beta}=\lambda^{\beta
}\left(  t\right)  f_{0}\left(  x\right)  +h^{\beta}$ with $\lambda^{\alpha
},\ \lambda^{\beta}\in C\left[  0,T\right]  ,\ h^{\alpha},\ h^{\beta}%
\in\mathcal{Z}_{\bar{p}}^{\sigma;\frac{1}{2}}\left(  T\right)  .$

Using the change of variables (\ref{time}), for both solutions and denoting as
$\tau$ the new time scale in both case we obtain functions $\tilde{f}^{\alpha
},\ \tilde{f}^{\beta}$ satisfying (\ref{Y1E3}) with $\Lambda=\Lambda^{\alpha
}=\lambda^{\alpha}$ and $\Lambda=\Lambda^{\beta}=\lambda^{\beta}$
respectively. We will write, with a bit abuse of notation $\tilde{f}^{\alpha
}=\Lambda^{\alpha}\left(  \tau\right)  f_{0}\left(  x\right)  +h^{\alpha
}\ \ ,\ \ \tilde{f}^{\beta}=\Lambda^{\beta}\left(  \tau\right)  f_{0}\left(
x\right)  +h^{\beta}.$ We define functions $\tilde{h}_{2}^{\alpha},\ \tilde
{h}_{2}^{\beta}\in\mathcal{Z}_{\bar{p}}^{\sigma;\frac{1}{2}}\left(  T\right)
$ by means of (\ref{P1E1}) with the corresponding functions $\Lambda^{\alpha
},\ \Lambda^{\beta}.$ We define also the functions $\tilde{h}_{1}^{\alpha
},\ \tilde{h}_{1}^{\beta}\in\mathcal{Z}_{\bar{p}}^{\sigma;\frac{1}{2}}\left(
T\right)  $ by means of $\tilde{h}_{1}^{\alpha}=h^{\alpha}-\tilde{h}%
_{2}^{\alpha}\ \ ,\ \ \tilde{h}_{1}^{\beta}=h^{\beta}-\tilde{h}_{2}^{\beta}.$

Using arguments analogous to the ones used in the derivation of (\ref{Y1E3})
we obtain:%
\[
\left(  \tilde{h}_{1}^{k}\right)  _{\tau}-\mathcal{L}_{f_{0}}\left[  \tilde
{h}_{1}^{k}\right]  =\frac{Q\left[  h^{k}\right]  }{\Lambda^{k}\left(
\tau\right)  }+\Lambda^{k}\left(  \tau\right)  Q\left[  f_{0}\right]
\ \ ,\ \ \ \tilde{h}_{1}^{k}\left(  0\right)  =0\ \ ,\ \ k=\alpha,\beta
\]

Using Proposition \ref{Propositionh_1} we obtain that $\tilde{h}_{1}^{\alpha
},\ \tilde{h}_{1}^{\beta}$ have the asymptotics (\ref{G3E1}). Moreover, for
$T\leq T_{0}$ small enough we have that the operator $\mathcal{G}$ is
contractive in $h$ and $\Lambda.$ Therefore:
\[
\left\Vert \mathcal{G}\left[  \cdot;h^{\alpha},\Lambda^{\alpha}\right]
-\mathcal{G}\left[  \cdot;h^{\beta},\Lambda^{\beta}\right]  \right\Vert
_{C\left[  0,T\right]  }\leq\theta\left\Vert \Lambda^{\alpha}-\Lambda^{\beta
}\right\Vert _{C\left[  0,T\right]  }+\theta\left\Vert h^{\alpha}-h^{\beta
}\right\Vert _{\mathcal{Z}_{\bar{p}}^{\sigma;\frac{1}{2}}\left(  T\right)  }%
\]
where $0<\theta<1$ can be made arbitrarily small for $T_{0}$ sufficiently
small. Moreover, since $h^{\alpha},\ h^{\beta}\in\mathcal{Z}_{\bar{p}}%
^{\sigma;\frac{1}{2}}\left(  T\right)  ,$ the functions $\Lambda^{\alpha
},\ \Lambda^{\beta}$ solve the integral equation (\ref{Z4E1}). The function
$\Lambda$ depends in a Lifschitz manner on the function $\mathcal{G}$ with a
constant smaller than two if $T\leq T_{0}.$ Therefore $\left\Vert
\Lambda^{\alpha}-\Lambda^{\beta}\right\Vert _{C\left[  0,T\right]  }%
\leq2\theta\left\Vert \Lambda^{\alpha}-\Lambda^{\beta}\right\Vert _{C\left[
0,T\right]  }+2\theta\left\Vert h^{\alpha}-h^{\beta}\right\Vert _{\mathcal{Z}%
_{\bar{p}}^{\sigma;\frac{1}{2}}\left(  T\right)  }$ whence $\left\Vert
\Lambda^{\alpha}-\Lambda^{\beta}\right\Vert _{C\left[  0,T\right]  }%
\leq4\theta\left\Vert h^{\alpha}-h^{\beta}\right\Vert _{\mathcal{Z}_{\bar{p}%
}^{\sigma;\frac{1}{2}}\left(  T\right)  }.$ Using then the contractivity of
the mapping $\left(  h,\Lambda\right)  \rightarrow\tilde{h}_{1}$ (cf.
Proposition \ref{Propositionh_1}) we then obtain $\left\Vert \tilde{h}%
_{1}^{\alpha}-\tilde{h}_{1}^{\beta}\right\Vert _{\mathcal{Z}_{\bar{p}}%
^{\sigma;\frac{1}{2}}\left(  T\right)  }\leq\frac{1}{4}\left\Vert h^{\alpha
}-h^{\beta}\right\Vert _{\mathcal{Z}_{\bar{p}}^{\sigma;\frac{1}{2}}\left(
T\right)  }.$ On the other hand (\ref{P1E1}) yields $\left\Vert \tilde{h}%
_{2}^{\alpha}-\tilde{h}_{2}^{\beta}\right\Vert _{\mathcal{Z}_{\bar{p}}%
^{\sigma;\frac{1}{2}}\left(  T\right)  }\leq C\theta\left\Vert h^{\alpha
}-h^{\beta}\right\Vert _{\mathcal{Z}_{\bar{p}}^{\sigma;\frac{1}{2}}\left(
T\right)  }.$ Therefore, choosing $\theta$ small enough:%
\[
\left\Vert h^{\alpha}-h^{\beta}\right\Vert _{\mathcal{Z}_{\bar{p}}%
^{\sigma;\frac{1}{2}}\left(  T\right)  }\leq\left\Vert \tilde{h}_{1}^{\alpha
}-\tilde{h}_{1}^{\beta}\right\Vert _{\mathcal{Z}_{\bar{p}}^{\sigma;\frac{1}%
{2}}\left(  T\right)  }+\left\Vert \tilde{h}_{2}^{\alpha}-\tilde{h}_{2}%
^{\beta}\right\Vert _{\mathcal{Z}_{\bar{p}}^{\sigma;\frac{1}{2}}\left(
T\right)  }\leq\frac{1}{2}\left\Vert h^{\alpha}-h^{\beta}\right\Vert
_{\mathcal{Z}_{\bar{p}}^{\sigma;\frac{1}{2}}\left(  T\right)  }%
\]
whence $h^{\alpha}=h^{\beta}.$ Then $\Lambda^{\alpha}=\Lambda^{\beta}$ and the
uniqueness follows.
\end{proof}

\textbf{Acknowledgements. }ME is supported by Grants MTM2008-03541 and
IT-305-07. JJLV is supported by Grant MTM2007-61755. He also thanks
Universidad Complutense.


\begin{thebibliography}{99}                                                                                               %


\bibitem {BZ}A. M. Balk, V. E. Zakharov, Stability of Weak-Turbulence
Kolmogorov Spectra in Nonlinear Waves and Weak Turbulence, V. E. Zakharov ed.,
AMS. Translations Series 2, 182, 1-81, 1998.

\bibitem {EZH}M. H. Ernst, R. M. Zi \& E. M. Hendriks, Coagulation processes
with a phase transition, J. of Colloid and Interface Sci. 97, 266-277, (1984).

\bibitem {EMP}M. Escobedo, S. Mischler, B. Perthame, Gelation in coagulation
and fragmentation models. Commun. Math. Phys. 231(1), 157--188, (2002).

\bibitem {EMV1}M. Escobedo, S. Mischler, J.J.L. Vel\'{a}zquez, On the
fundamental solution of the linearized Uehling-Uhlenbeck equation. Arch. Rat.
Mech. Anal. 186, 309--349, (2007).

\bibitem {EMV2}M. Escobedo, S. Mischler, J.J.L. Vel\'{a}zquez, Singular
Solutions for the Uehling-Uhlenbeck Equation. Proc. Roy. Soc. Edinburgh 138A,
67--107, (2008).

\bibitem {EV1}M. Escobedo, J.J.L. Vel\'{a}zquez, On the fundamental solution
of a homogeneous linearized coagulation equation. Comm. Maths. Phys. 3, 297,
759-816, (2010).

\bibitem {EV2}M. Escobedo, J.J.L. Vel\'{a}zquez, Local well posedness for a
linear coagulation equation. Preprint.

\bibitem{FL} N. Fournier, P. Lauren\c cot, Marcus-Lushnikov processes, Smoluchowski's and Flory's models,
Stochastic Process. Appl. 119,  167--189  (2009).

\bibitem {Flory41}P. J. Flory, Molecular Size Distribution in Three
Dimensional Polymers. II. Trifunctional Branching Units. J. Am. Chem. Soc. 63,
3091-3096, (1941).


\bibitem {Jeon}I. Jeon, Existence of gelling solutions for
coagulation-fragmentation equations. Comm. Math. Phys. 194, 541-567, (1998).

\bibitem {LLPR}R. Lacaze, P. Lallemand, Y. Pomeau, S. Rica, Dynamical
formation of a Bose-Einstein condensate. Physica D, 152-153, 779--786, (2001).

\bibitem {L}P. Lauren\c{c}ot, On a class of continuous
coagulation-fragmentation equations. J. Diff. Equ. 167, 245-274, (2000).


\bibitem {Lu1}X. G. Lu, A modified Boltzmann equation for Bose-Einstein
particles: isotropic solutions and long time behavior. J. Stat. Phys. 98,
1335-1394, (2000).

\bibitem {Lu2}X. G. Lu, On isotropic distributional solutions to the Boltzmann
equation for Bose-Einstein particles. J. Stat. Phys. 116, 1597--1649, (2004).

\bibitem {McLeod}J. B. McLeod, On the scalar transport equation. Proc. London
Math. Soc. 14, N.3, 445-458, (1964).

\bibitem {RS}T. Runst, W. Sickel, Sobolev Spaces of Fractional Order,
Nemytskij Operators and Nonlinear Partial Differential Equations. De Gruyter,
Berlin, 1996.

\bibitem {ST1}D.V. Semikov, I.I. Tkachev, Kinetics of Bose condensation. Phys.
Rev. Lett. 74, 3093-3097, (1995).

\bibitem {ST2}D.V. Semikov, I.I. Tkachev, Condensation of Bosons in the
kinetic regime. Phys. Rev. D, 55, 2, 489-502, (1997).

\bibitem {Spohn}H. Spohn, Kinetics of the Bose-Einstein condensation. Physica
D, 239, 627-634, (2010).

\bibitem {St}W. H. Stockmayer, Theory of Molecular Size distribution and gel
formation in branched-chain polymers. J. Chem. Phys. 11, 2, 45-55, (1943).
\end{thebibliography}
\end{document}